\documentclass[conference]{IEEEtran}
\ifCLASSINFOpdf
\else
\fi

\usepackage{amsmath,amsthm}
\theoremstyle{problem}

\usepackage{amsfonts}
\usepackage{amsmath,amsthm}
\usepackage{graphicx,epstopdf,epsfig,multirow,epic,bm}
\usepackage{color}
\usepackage{multicol}
\usepackage{makeidx,showidx}
\usepackage{ifthen}
\usepackage{mathrsfs}
\usepackage{colortbl}
\usepackage{multicol}
\usepackage{latexsym}
\usepackage{amssymb}
\usepackage{algorithm}
\usepackage{algorithmic}
\usepackage{longtable}
\usepackage{array}
\usepackage{enumerate}
\usepackage{paralist}
\usepackage{fancyhdr}
\usepackage{CJK}

\theoremstyle{plain}
\newtheorem{thm}{Theorem}

\newtheorem{lem}[thm]{Lemma}

\theoremstyle{definition}
\newtheorem{defn}{Definition}
\theoremstyle{definition}
\newtheorem{rem}{Remark}

\theoremstyle{definition}

\theoremstyle{definition}

\theoremstyle{definition}

\newcommand{\ud}{\mathrm{deg}}
\newcommand{\qqed}{\hfill $\Box$}

\begin{document}

\pagestyle{fancy}
\pagestyle{plain}

\begin{CJK}{GBK}{song}

%
\title{Topological Coding and Topological Matrices Toward Network Overall Security}


\author{\IEEEauthorblockN{Bing Yao$^{1,6}$, Meimei  Zhao$^{2}$, ~Xiaohui Zhang$^{3}$, Yarong Mu$^{1}$, Yirong Sun$^{1}$, Mingjun Zhang$^{4}$, Sihua Yang$^{4}$\\
Fei Ma$^{5,\ddagger}$,~Jing Su$^{5}$, ~Xiaomin Wang$^{5}$,~Hongyu Wang$^{5}$, ~Hui Sun$^{5}$}
\IEEEauthorblockA{{1} College of Mathematics and Statistics,
 Northwest Normal University, Lanzhou, 730070, CHINA}
 \IEEEauthorblockA{{2} College of Science, Gansu Agricultural University, Lanzhou 730070, CHINA}
 \IEEEauthorblockA{{3} College of Mathematics and Statistics, Jishou University, Jishou, Hunan, 416000, CHINA}
 \IEEEauthorblockA{{4} School of Information Engineering, Lanzhou University of Finance and Economics, Lanzhou, 730030, CHINA}
 \IEEEauthorblockA{{5} School of Electronics Engineering and Computer Science, Peking University, Beijing, 100871, CHINA}
\IEEEauthorblockA{{6} School of Electronics and Information Engineering, Lanzhou Jiaotong University, Lanzhou, 730070, CHINA\\
$^\dagger$ Corresponding authors: Fei Ma's email: 1337725455@qq.com}

\thanks{Manuscript received June 1, 2016; revised August 26, 2016.
Corresponding author: Bing Yao, email: yybb918@163.com.}}


%


\maketitle

\begin{abstract}
A mathematical topology with matrix is a natural representation of a coding relational structure that is found in many fields of the world. Matrices are very important in computation of real applications, s ce matrices are easy saved in computer and run quickly, as well as matrices are convenient to deal with communities of current networks, such as Laplacian matrices, adjacent matrices in graph theory. Motivated from convenient, useful and powerful matrices used in computation and investigation of today's networks, we have introduced Topcode-matrices, which are matrices of order $3\times  q$ and differ from popular matrices applied in linear algebra and computer science. Topcode-matrices can use numbers, letters, Chinese characters, sets, graphs, algebraic groups \emph{etc.} as their elements. One important thing is that Topcode-matrices of numbers can derive easily number strings, since number strings are text-based passwords used in information security. Topcode-matrices can be used to describe topological graphic passwords (Topsnut-gpws) used in information security and graph connected properties for solving some problems coming in the investigation of Graph Networks and Graph Neural Networks proposed by GoogleBrain and DeepMind. Our topics, in this article, are: Topsnut-matrices, Topcode-matrices, Hanzi-matrices, adjacency ve-value matrices and pan-Topcode-matrices, and some connections between these Topcode-matrices will be proven. We will  discuss algebraic groups obtained from the above matrices,  graph groups, graph networking groups and number string groups for  encrypting different communities of dynamic networks. The operations and results on our matrices help us to set up our overall security mechanism to protect networks. \\[4pt]
\end{abstract}
\textbf{\emph{Keywords---Graph theory; coding matrix; topological matrix; Abelian additive group; network security.}}


%
\IEEEpeerreviewmaketitle

\section{Introduction}

``Since the quantum revolution, we have become increasingly aware that our world is not continuous, but discrete. We should look at the world by algebra of view.'' said by Xiaogang Wen, a member of Academician of the National Academy of Sciences in \cite{Xiaogang-Wen-2018-Sun-Yat-sen-University}. Graphs are natural representations for encoding relational structures that are encountered in many domains (Ref. \cite{Li-Gu-Dullien-Vinyals-Kohli-arXiv2019}). These are our ideas to write this article.

\subsection{Investigation background}

The calculation of graph structure data is widely used in various fields, such as molecular analysis of computational biology and chemistry, analysis of natural language understanding, or graph structure analysis of knowledge map, \emph{etc.} (Ref. \cite{Li-Gu-Dullien-Vinyals-Kohli-arXiv2019}). Cryptography is useful and important to the security of today's networks. Graphical passwords (GPWs) were proposed for a long tim (Ref. \cite{Suo-Zhu-Owen-2005, Biddle-Chiasson-van-Oorschot-2009, Gao-Jia-Ye-Ma-2013}), especially, two-dimension codes are popular and powerful used in the world. Another type of GPWs is Topsnut-gpw (the abbreviation of ``Graphical passwords based on the idea of topological structure plus number theory'') that was proposed first by Wang \emph{et al.} in \cite{Wang-Xu-Yao-2016} and \cite{Wang-Xu-Yao-Key-models-Lock-models-2016}.

Topsnut-gpws differ from the existing GPWs, because of they can be expressed by popular matrices (Ref. \cite{Bondy-2008}). Matrices are useful and powerful for saving graphs and solving many systems of linear equations, as known. However, a remarkable advantage of Topsnut-gpws is easily to produce Text-based passwords (TB-paws), since the application of TB-paws are in fashion of information networks. Six Topsnut-gpws (a)-(f) are shown in Fig.\ref{fig:1-example}.

In general, matrices are very important in computation of real applications, since matrices are easy saved in computer and run quickly, as well as matrices are convenient to deal with communities of current networks. An successful example is Laplacian matrix, also known as admittance matrix, Kirchhoff matrix or discrete Laplacian operator, that is mainly used in graph theory as a matrix representation of a graph.

A $(p,q)$-graph $G$ with vertex set $V(G)=\{u_1,u_2,\dots ,u_p\}$ and $q$ edges has its own \emph{adjacent matrix} $A(G)=(a_{ij})_{p\times p}$ with $a_{ij}=1$ if $u_i$ is adjacent to $u_j$, otherwise $a_{ij}=0$ and $a_{ii}=0$; its \emph{degree matrix} $D_G=(d_{ij})_{p\times p}$, $d_{ij}=0$ if $i\neq j$, otherwise $d_{ii}$ is equal to the \emph{degree} $\textrm{deg}_G(u_{i})$ of vertex $u_{i}$. Then $G$ has its own \emph{Laplacian matrix} as follows
$$
L(G)=(l_{ij})_{p\times p}=D_G-A(G)
$$ with $l_{ij}=d_{ij}-a_{ij}$. Moreover, a \emph{symmetric normalized Laplacian matrix} is defined as
$${
\begin{split}
L^{sym}&=(l^{sym}_{ij})_{p\times p}=(D_G)^{-1/2}L~(D_G)^{-1/2}\\
&=I-(D_G)^{-1/2}A(G)(D_G)^{-1/2}
\end{split}}
$$
with $l^{sym}_{ij}=1$ if $i=j$ and $\textrm{deg}_G(u_{i})\neq 0$; and
$$
l^{sym}_{ij}=-\frac{1}{\sqrt{\textrm{deg}_G(u_{i})\cdot \textrm{deg}_G(u_{j})}}
$$
if $i\neq j$ and $u_i$ is adjacent with $u_j$; otherwise $l^{sym}_{ij}=0$. However, our matrices that will be introduced here differ from Laplacian matrices and adjacent matrices.

Before starting our topic in this article, let us see an example shown in a matrix (\ref{eqa:example}) and Fig.\ref{fig:1-example}. There are six Topsnut-gpws (a)-(f) pictured in Fig.\ref{fig:1-example}, in which each of (b)-(f) Topsnut-gpw admits a \emph{splitting odd-edge-magic coloring}, except (a) that admits an odd-edge-magic total labelling. Clearly, the topological structure of each of these six Topsnut-gpws is not isomorphic to any one of others. However, each of these six Topsnut-gpws can be expressed by the matrix $A$ shown in (\ref{eqa:example}). It is not difficult to see that the Topsnut-gpw (a) in Fig.\ref{fig:1-example} can be split into nine edges, in other word, the matrix $A$ derives disconnected and connected Topsnut-gpws more than six. On the other hands, we can use the matrix $A$ shown in (\ref{eqa:example}) as a \emph{public key}, and users provided six Topsnut-gpws (a)-(f) as \emph{private keys} in real application. Such one-vs-more technique increase the offensive difficulty and huge times of deciphering passwords in order to drive attackers out of networks.

\begin{equation}\label{eqa:example}
\centering
{
\begin{split}
A= \left(
\begin{array}{ccccccccc}
7 & 5 & 7 &1 & 5 & 1 &1 & 1& 1\\
1 & 3 & 5 &7 & 9 & 11 &13& 15& 17\\
18& 18 & 14 &18& 12 & 14 &12& 10& 8
\end{array}
\right)
\end{split}}
\end{equation}

\begin{figure}[h]
\centering
\includegraphics[width=8.6cm]{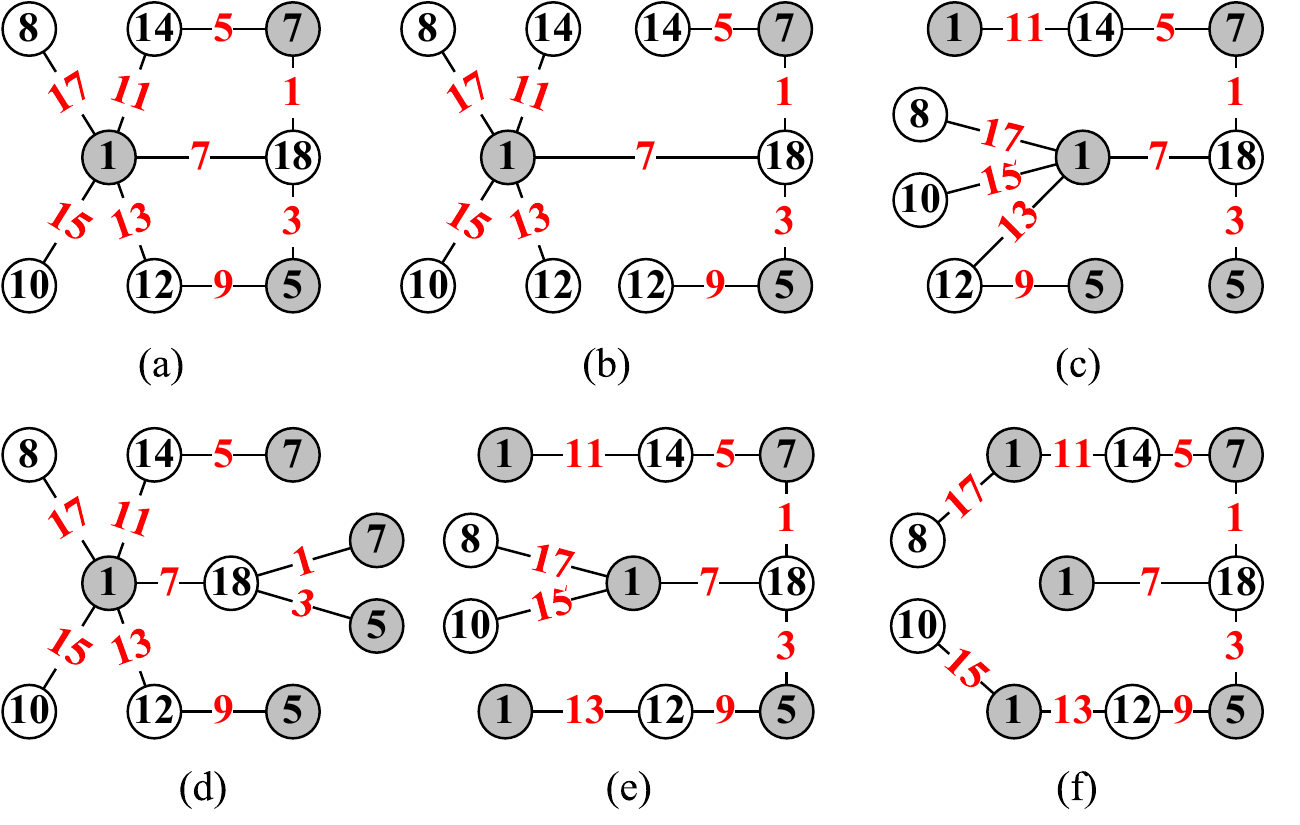}\\
\caption{\label{fig:1-example} {\small (a)-(f) are Topsnut-gpws with (splitting) odd-edge-magic colorings and corresponding the Topcode-matrix $A$, where (a) has no perfect matching, (f) has a perfect matching.}}
\end{figure}

We call the matrix $A$ shown in (\ref{eqa:example}) a \emph{topological coding matrix} (Topcode-matrix) with $q=9$ according to the following Definition \ref{defn:Topcode-matrix}. The above example derives the following questions:

\vskip 0.4cm

\begin{asparaenum}[Que-1. ]
\item What properties do Topcode-matrices have?
\item Does each Topcode-matrix derive a graph?
\item What operations exist on Topcode-matrices?
\item What applications do Topcode-matrices have?
\end{asparaenum}

\vskip 0.4cm

In \cite{Battaglia-27-authors-arXiv1806-01261v2}, the authors ask for: (1) Where do the graphs come from that graph networks operate over? Moreover, many underlying graph structures are much more \emph{sparse} than a fully connected graph, and it is an open question how to induce this \emph{sparsity}. (2) How to adaptively modify graph structures during the course of computation?

So we will do graphical study on various matrices mentioned here for exploring the above problems and ``the interpretability of the behavior of graph networks'' proposed in \cite{Battaglia-27-authors-arXiv1806-01261v2}, in which one type of matrices can be considered as mathematical models of Chinese characters, since ``\emph{If there is a country in the world, every word of her can become a poem, a painting, then there must be Chinese characters}''. Furthermore, we will build up systems of linear equations as mathematical models of Chinese characters here, another mathematical models of Chinese characters was studied in \cite{Yao-Mu-Sun-Sun-Zhang-Wang-Su-Zhang-Yang-Zhao-Wang-Ma-Yao-Yang-Xie2019}. Most of topics here need readers have learned basic knowledge of linear algebra and graph theory.

\subsection{Preliminary}
In this article, we use the following notation and terminology, the others no mentioned can be found in \cite{Bondy-2008}. To be more precise, we define these terms as:

\begin{asparaenum}[\textrm{Term}-1. ]
\item \emph{Particular sets}. The notation $[a,b]$ indicates a set $\{a,a+1,\dots, b\}$ for two integers $a,b$ with respect to range $b>a\geq 0$, $[s,t]^o$ indicates another set $\{s,s+2,\dots, t\}$ for two odd numbers $s,t$ falling into $t>s\geq 1$. Moreover, $S_{k,d}=\{k,k+d,k+2d, \dots ,k+(q-1)d\}$ and $S_{2k,2d}=\{2k+2d,2k+4d,2k+6d, \dots ,2k+2qd\}$ with integers $k,d\geq 1$ and $q\geq 2$.
\item \emph{Networks/graphs.} A network, also, is a graph in mathematics. A $(p,q)$-graph has $p$ vertices and $q$ edges.
\item \emph{Neighbor set and vertex degree.} The symbol $N_{ei}(u)$ stands for the set of all neighbors of a vertex $u$, thus, the number $\ud_G(u)=|N_{ei}(u)|$ is called the degree of $u$, where $|S|$ is the cardinality of elements of s set $S$.
\item \emph{Leaves.} A \emph{leaf} $x$ is a vertex of degree one, also, $x$ is called a \emph{appended vertex} of $y\in N_{ei}(x)=\{y\}$.
\item \emph{Sets of sets}. Let $S$ be a set, The set of all subsets is denoted as $S^2=\{X:~X\subseteq S\}$, and $S^2$ contains no empty set at all. For example, a set $S=\{a,b,c,d,e\}$, so $S^2$ has its own elements $\{a\}$, $\{b\}$, $\{c\}$, $\{d\}$, $\{e\}$, $\{a,b\}$, $\{a,c\}$, $\{a,d\}$, $\{a,e\}$, $\{b,c\}$, $\{b,d\}$, $\{b,e\}$, $\{c, d\}$, $\{c, e\}$, $\{d,e\}$, $\{a,b,c\}$, $\{a,b,d\}$, $\{a,b,e\}$, $\{a,c,d\}$, $\{a,c,e\}$, $\{a,d,e\}$, $\{b,c,d\}$, $\{b,c,e\}$, $\{a,b,c,d\}$, $\{a,b,c,e\}$, $\{a,c,d,e\}$, $\{b,c,d,e\}$ and $\{a,b,c,d,e\}$.
\end{asparaenum}

\begin{defn}\label{defn:set-ordered-odd-graceful-labelling}
(\cite{Gallian2018, Bing-Yao-Cheng-Yao-Zhao2009, Zhou-Yao-Chen-Tao2012}) Suppose that a bipartite $(p,q)$-graph $G$ with partition $(X,Y)$ admits a vertex labelling $f:V(G) \rightarrow [0,2q-1]$ (resp. $[0,q]$), such that every edge $uv$ is labeled as $f(uv)=|f(u)-f(v)|$ holding $f(E(G))=[1, 2q-1]^o$ (resp. $[1,q]$), we call $f$ an \emph{odd-graceful labelling} of $G$. Furthermore, if $f$ holds $\max\{f(x):~x\in X\}< \min\{f(y):~y\in Y\}$ ($f_{\max}(X)<f_{\min}(Y)$ for short), then $f$ is called a \emph{set-ordered odd-graceful labelling} (resp. a set-ordered graceful labelling).\qqed
\end{defn}

\begin{defn}\label{defn:splitting-(odd)graceful-coloring}
\cite{Yao-Mu-Sun-Zhang-Yang-Wang-Wang-Su-Ma-Sun-2019} Suppose that a connected $(p,q)$-graph $G$ admits a coloring $f:V(G) \rightarrow [0,q]$ (resp. $[0,2q-1]$), such that $f(u)=f(v)$ for some pairs of vertices $u,v\in V(G)$, and the edge label set $f(E(G))=\{f(uv)=|f(u)-f(v)|: ~uv\in E(G)\}=[1,~q]$ (resp. $[1,2q-1]^o$), then we call $f$ a \emph{splitting graceful coloring} (resp. splitting odd-graceful coloring). \qqed
\end{defn}

In general, we have the following definition of splitting $\epsilon$-colorings:

\begin{defn}\label{defn:splitting-(odd)graceful-coloring}
$^*$ A connected $(p,q)$-graph $G$ admits a coloring $f:S \rightarrow [a,b]$, where $S\subseteq V(G)\cup E(G)$, and there exist $f(u)=f(v)$ for some distinct vertices $u,v\in V(G)$, and the edge label set $f(E(G))$ holds an $\epsilon$-condition, so we call $f$ a \emph{splitting $\epsilon$-coloring} of $G$. \qqed
\end{defn}

Some connected graphs admitting splitting \emph{odd-edge-magic colorings} are shown in Fig.\ref{fig:1-example} (b)-(f).

\section{Topcode-matrices and other matrices}

In this section, we will work on the following two topics:

(A) Particular matrices: (i) Topsnut-matrices; (ii) Topcode-matrices; (iii) Hanzi-matrices; (iv) Adjacency ve-value matrices.

(B) Groups on the above four classes of matrices, and number string groups.

\subsection{Topcode-matrices}

We are motivated from Topsnut-matrices made by Topsnut-gpws, and define a general concept, called \emph{Topcode-matrices}.

\vskip 0.4cm

\subsubsection{Concept of Topcode-matrices}\quad

\begin{defn}\label{defn:Topcode-matrix}
$^*$ A \emph{Topcode-matrix} (topological coding matrix) is defined as
\begin{equation}\label{eqa:Topcode-matrix}
\centering
{
\begin{split}
T_{code}&= \left(
\begin{array}{ccccc}
x_{1} & x_{2} & \cdots & x_{q}\\
e_{1} & e_{2} & \cdots & e_{q}\\
y_{1} & y_{2} & \cdots & y_{q}
\end{array}
\right)=
\left(\begin{array}{c}
X\\
E\\
Y
\end{array} \right)_{3\times q}\\
&=(X~E~Y)^{T}_{3\times q}
\end{split}}
\end{equation}\\
where \emph{v-vector} $X=(x_1 ~ x_2 ~ \cdots ~x_q)$, \emph{e-vector} $E=(e_1$ ~ $e_2 $ ~ $ \cdots $ ~ $e_q)$, and \emph{v-vector} $Y=(y_1 ~ y_2 ~\cdots ~ y_q)$ consist of integers $e_i$, $x_i$ and $y_i$ for $i\in [1,q]$. We say the Topcode-matrix $T_{code}$ to be \emph{evaluated} if there exists a function $f$ such that $e_i=f(x_i,y_i)$ for $i\in [1,q]$, and call $x_i$ and $y_i$ to be the \emph{ends} of $e_i$.\qqed
\end{defn}

In Definition \ref{defn:Topcode-matrix}, we collect all different elements in two v-vectors $X$ and $Y$ into a set $(XY)^*$, and all of different elements in the e-vector $E$ into a set $E^*$, as well as $x_i\neq y_i$ with $i\in [1,q]$. Moreover, graphs have their own incident matrices like Topcode-matrices defined in Definition \ref{defn:Topcode-matrix} (Ref. \cite{Bondy-2008}).

A similar concept, \emph{Topsnut-matrix}, was introduced in \cite{Sun-Zhang-Zhao-Yao-2017, Yao-Sun-Zhao-Li-Yan-2017, Yao-Zhang-Sun-Mu-Sun-Wang-Wang-Ma-Su-Yang-Yang-Zhang-2018arXiv} for studying Topsnut-gpws and producing TB-paws. Pr\"{u}fer Code is a part of Topcode-matrices and can be used to prove Cayley's formula $\tau(K_n) = n^{n-2}$, where $\tau(K_n)$ is the number of spanning trees in a complete graph $K_n$ (Ref. \cite{Bondy-2008}).

We point that Topcode-matrices mentioned here differ from adjacent matrices of graphs of graph theory. For example, the Topcode-matrix (\ref{eqa:example}) describes two Topsnut-gpws (A) and (B) shown in Fig.\ref{fig:adjacent-no-topo}, but it differs from any one of two adjacent matrices shown in Fig.\ref{fig:adjacent-no-topo}. An adjacent matrix corresponds a unique graph, however a Topcode-matrix may correspond two or more graphs (Topsnut-gpws). It is easy to see that we need more spaces to save adjacent matrices in computer. We observe a fact as follows:

\begin{thm}\label{thm:two-Topsnut-gpws-or-more}
If a Topcode-matrix $T_{code}$ defined in Definition \ref{defn:Topcode-matrix} corresponds a connected Topsnut-gpw not being a tree, then $T_{code}$ corresponds at least two Topsnut-gpws or more.
\end{thm}

\begin{figure}[h]
\centering
\includegraphics[width=8.6cm]{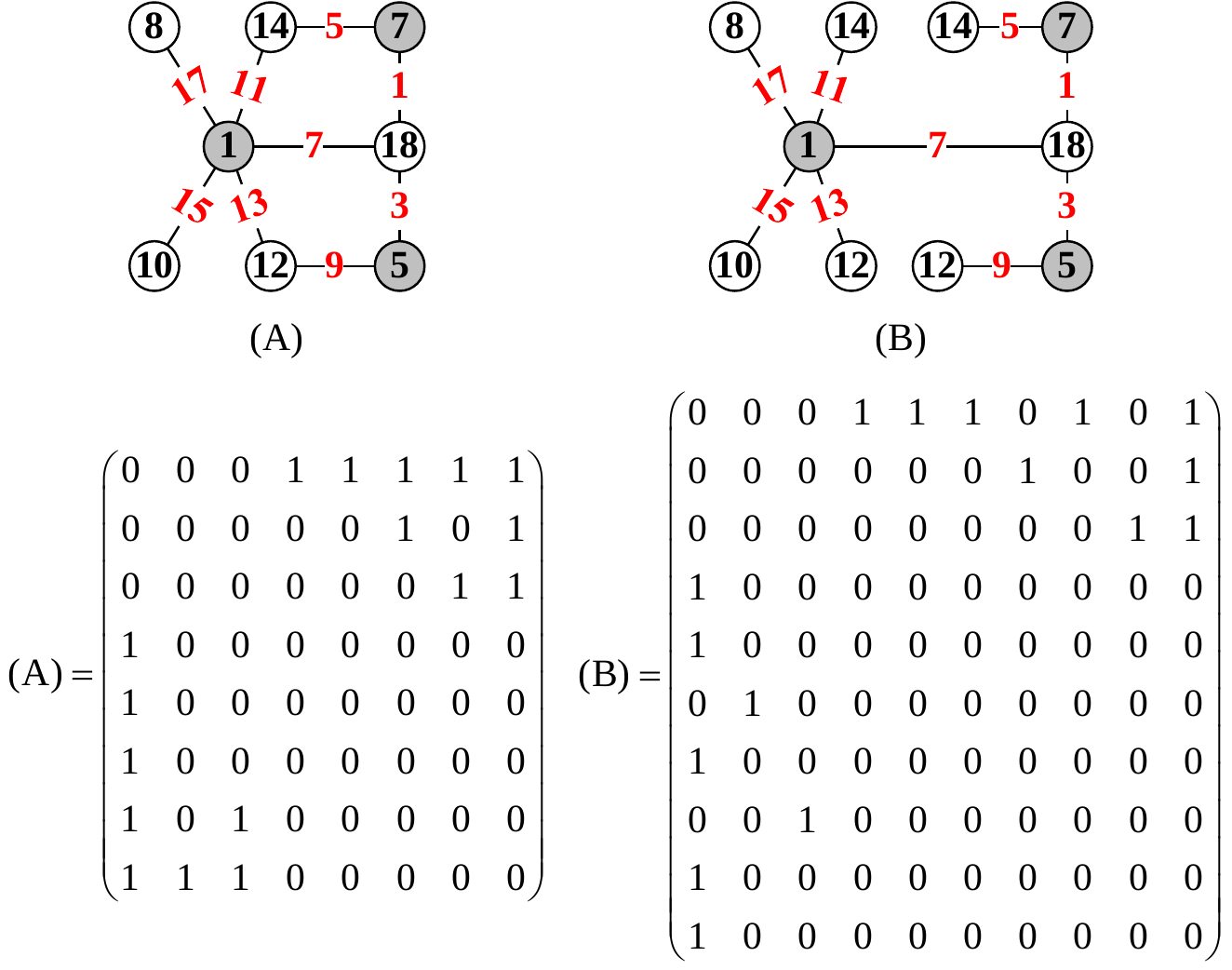}\\
\caption{\label{fig:adjacent-no-topo} {\small Two Topsnut-gpws (A),(B) and their adjacent matrices.}}
\end{figure}
\vskip 0.4cm

\subsubsection{Particular sub-Topcode-matrices}\quad (1) \emph{Perfect matching.} If we have a subset $E_S=\{e_{i_1},e_{i_2},\dots ,e_{i_m}\}\subset E^*$ in a Topcode-matrix $T_{code}$ defined in Definition \ref{defn:Topcode-matrix}, and there exists $w\in (XY)^*$ to be not a common end of $e_{i_j}$ and $e_{i_t}$ of $E_S$, and the set of all ends of $E_S$ is just equal to $(XY)^*$, then we call $E_S$ a \emph{perfect matching} of $T_{code}$, and $(X_S~E_S~Y_S)^{T}$ a \emph{perfect matching sub-Topcode-matrix} of $T_{code}$. The Topcode-matrix $A$ shown in (\ref{eqa:example}) has a perfect matching $E_S=\{5,7,9,15,17\}$, and a perfect matching sub-Topcode-matrix as follows
\begin{equation}\label{eqa:Topcode-matrix}
\centering
A_p= \left(
\begin{array}{ccccc}
1 & 7 & 1 & 5& 1\\
17 & 5 &7 & 9& 15\\
8 & 14 & 18 &12& 10
\end{array}
\right)
\end{equation}

(2) \emph{Complete Topcode-matrices.} If any $x_i\in (XY)^*$ matches with each $y_j\in (XY)^*\setminus \{x_i\}$ such that $x_i,y_j$ are the ends of some $e_{ij}\in E^*$ in a Topcode-matrix $T_{code}$ defined in Definition \ref{defn:Topcode-matrix}, and $q=\frac{p(p-1)}{2}$ with $p=|(XY)^*|$, we say $T_{code}$ a \emph{complete Topcode-matrix}.

(3) \emph{Clique Topcode-matrices.} Let $T'_{code}$ be a proper sub-Topcode-matrix of $T_{code}$ defined in Definition \ref{defn:Topcode-matrix}. If $T'_{code}$ is a complete Topcode-matrix, we call $T'_{code}$ a \emph{clique Topcode-matrix}.

(4) \emph{Paths and cycles in Topcode-matrices.} Notice that $x_i$ and $y_i$ are the \emph{ends} of $e_i$ in $T_{code}$ defined in Definition \ref{defn:Topcode-matrix}. We define two particular phenomenons in Topcode-matrices: If there are $e_{i_1},e_{i_2},\dots ,e_{i_m}\in E^*$ in $T_{code}$, such that $e_{i_j}$ and $e_{i_{j+1}}$ with $j\in [1,m-1]$ have a common end $w_{i_j}$ holding $w_{i_s}\neq w_{i_t}$ for $s\neq t$ and $1\leq s,t\leq m$. Then $T_{code}$ has a \emph{path}, denote this path as
$$P(x'_{i_1}\rightarrow y'_{i_m})=x'_{i_1}e_{i_1}w_{i_1}e_{i_2}w_{i_2}\cdots e_{i_m-1}w_{i_m-1}e_{i_m}y'_{i_m}$$
where $x'_{i_1}$ is an end of $e_{i_1}$, and $y'_{i_m}$ is an end of $e_{i_m}$, $x'_{i_1}\neq y'_{i_m}$, $x'_{i_1}\neq w_{i_j}$ and $y'_{i_m}\neq w_{i_j}$ with $j\in[1,m-1]$; and moreover if $x'_{i_1}= y'_{i_m}=ww_{i_m}$ in $P(x'_{i_1}\rightarrow y'_{i_m})$, we say that $T_{code}$ has a \emph{cycle} $C$, denoted as $$C=w_{i_m}e_{i_1}w_{i_1}e_{i_2}w_{i_2}\cdots e_{i_m-1}w_{i_m-1}e_{i_m}w_{i_m}.$$
We call $T_{code}$ to be \emph{connected} if any pair of distinct numbers $w_i,w_j\in (XY)^*$ is connected by a path $P(w_{i}\rightarrow w_{j})$. Thereby, we have:
\begin{thm}\label{thm:tree-equivalent-Topcode-matrices}
Suppose that $e_i\neq e_j$ for $i\neq j$ and $1\leq i,j\leq m$ in a Topcode-matrix $T_{code}$ defined in Definition \ref{defn:Topcode-matrix}. Then the following assertions are equivalent to each other:
\begin{asparaenum}[\emph{\textrm{Tree}}-1. ]
\item $T_{code}$ corresponds a tree $T$ of $q$ edges.
\item $T_{code}$ is connected and has no cycle.
\item Any pair of $x_i$ and $x_j$ (resp. $x_i$ and $y_j$, or $y_i$ and $x_j$, or $y_i$ and $y_j$) in $T_{code}$ is connected by a unique path $P(x_i\rightarrow x_j)$ (resp. $P(x_i\rightarrow y_j)$, or $P(y_i\rightarrow x_j)$, or $P(y_i\rightarrow y_j)$).
\item $T_{code}$ is connected and $|(XY)^*|=q+1$.
\item $T_{code}$ has no cycle and $|(XY)^*|=q+1$.
\item \emph{\cite{Yao-Zhang-Wang2010}} The number $n_1(T_{code})$ of leaves of $T_{code}$ satisfies
\begin{equation}\label{eqa:tree-leaf-number}
n_1(T_{code})=2+\sum _{w\in (XY)^*,d\geq 2}(d-2)\mu(w),
\end{equation} where $\mu(w)$ for each $w\in (XY)^*$ is the number of times $w$ appeared in $T_{code}$.
\end{asparaenum}
\end{thm}

(5) \emph{Spanning sub-Topcode-matrices.} Let $T'_{code}=(X_S$ ~$E_S$~$Y_S)^{T}$ be a sub-Topcode-matrix of a Topcode-matrix $T_{code}$ defined in Definition \ref{defn:Topcode-matrix}, where $E_S=\{e_{i_1}~e_{i_2}~\dots ~e_{i_m}\}$ with $e_{i_j}\in E^*$. If $(X_SY_S)^*=$ $(XY)^*$, we call $T'_{code}$ a \emph{spanning sub-Topcode-matrix} of $T_{code}$. Moreover, if this spanning sub-Topcode-matrix $T'_{code}$ corresponds a tree, we call it a \emph{spanning tree Topcode-matrix}.

\begin{thm}\label{thm:spanning-tree-matrix}
Suppose a Topcode-matrix $T_{code}$ defined in Definition \ref{defn:Topcode-matrix} is connected and contains no proper spanning sub-Topcode-matrix. Then $T_{code}$ corresponds a tree.
\end{thm}

(6) \emph{Euler's Topcode-matrices and Hamilton Topcode-matrices.} If the number of times each $w\in (XY)^*$ appears in a Topcode-matrix $T_{code}$ defined in Definition \ref{defn:Topcode-matrix} is even, then $T_{code}$ is called an \emph{Euler's Topcode-matrix}. Moreover, if the number of times each $w\in (XY)^*$ appears in a connected Topcode-matrix $T_{code}$ is just two, then we call $T_{code}$ a \emph{Hamilton Topcode-matrix}.

\begin{thm}\label{thm:Hamilton-Euler-matrix}
Suppose that $T_{code}$ defined in Definition \ref{defn:Topcode-matrix} is a connected Euler's Topcode-matrix. Then $T_{code}$ contains a cycle of $q$ length. Moreover, if each $w\in (XY)^*$ appears in $C$ only once, then $T_{code}$ is a Hamilton Topcode-matrix.
\end{thm}

(7) \emph{Neighbor sets of Topcode-matrices.} For each $(x_i$ $e_i$ $y_i)^{T}$ $\in T_{code}$, we call set $N_{ei}(x_i)=\{w_j:$ $(x_i~e_{s}$ $w_j)^{T}\in T_{code}\}$ as \emph{v-neighbor set} of $x_i$, $N_{ei}(y_i)=\{z_t:~(z_t$ $e_{r}~y_i)^{T}\in T_{code}\}$ as \emph{v-neighbor set} of $y_i$ and $N'_{ei}(x_i)=\{e_{s}:$ $(x_i~e_{s}~w_j)^{T}\in T_{code}\}$ as \emph{e-neighbor set} of $e_i$.

\vskip 0.4cm

\subsubsection{Traditional Topcode-matrices}\quad Based on Definition \ref{defn:Topcode-matrix} and $|(XY)^*|=p$, we have the following restrict conditions:
\begin{asparaenum}[\textrm{Cond}-1. ]
\item \label{asparae:no-same} $(XY)^*=[0,p-1]$ with $p\leq q+1$.
\item \label{asparae:full} $|(XY)^*|=p\leq q+1$.
\item \label{asparae:graceful-no-same} $(XY)^*\subset [0,q]$.
\item \label{asparae:odd-no-same} $(XY)^*\subset [0,2q-1]$.
\item \label{asparae:odd-odd-magic} $(XY)^*\subset [0,2q]$.
\item \label{asparae:edge-magic-total-1} $(XY)^*\cup E^*=[1,p+q]$.
\item \label{asparae:only-E-odd} $E^*=[1,2q-1]^o$.
\item \label{asparae:super-edge-magic-total} $(XY)^*=[1,p]$ and $E^*=[p+1,p+q]$.
\item \label{asparae:set-ordered} $T_{code}$ defined in Definition \ref{defn:Topcode-matrix} is a \emph{set-ordered Topcode-matrix} if $\max \{x_i:$ $i\in [1,q]\}<\min \{y_j:~j\in [1,q]\}$.
\item \label{asparae:graceful} $e_i=|x_i-y_i|$, $i\in [1,q]$, $E^*=[1,q]$.
\item \label{asparae:odd-graceful} $e_i=|x_i-y_i|$, $i\in [1,q]$, $E^*=[1,2q-1]^o$.
\item \label{asparae:odd-elegant} $e_i=x_i+y_i~(\bmod~2q)$, $i\in [1,q]$, $E^*=[1,2q-1]^o$.
\item \label{asparae:elegant} $e_i=x_i+y_i~(\bmod~q)$, $i\in [1,q]$, $E^*=[0,q-1]$.
\item \label{asparae:edge-magic-total-2} There exists a constant $k$, such that $x_i+e_i+y_i=k$, $i\in [1,q]$.
\item \label{asparae:edge-vs-difference} There exists a constant $k$, such that $e_i+|x_i-y_i|=k$, $i\in [1,q]$.
\end{asparaenum}

\vskip 0.4cm

Based on the above group of conditional restrictions, we have the following particular Topcode-matrices, $T_{code}$ is defined in Definition \ref{defn:Topcode-matrix}:
\begin{asparaenum}[\textrm{Topmatrix}-1. ]
\item A Topcode-matrix $T_{code}$ is called a \emph{graceful Topcode-matrix} if Cond-\ref{asparae:graceful-no-same} and Cond-\ref{asparae:graceful} hold true.
 \item A Topcode-matrix $T_{code}$ is called a \emph{set-ordered graceful Topcode-matrix} if Cond-\ref{asparae:graceful-no-same}, Cond-\ref{asparae:set-ordered} and Cond-\ref{asparae:graceful} hold true.
\item A Topcode-matrix $T_{code}$ is called an \emph{odd-graceful Topcode-matrix} if Cond-\ref{asparae:odd-no-same} and Cond-\ref{asparae:odd-graceful} hold true.
\item A Topcode-matrix $T_{code}$ is called a \emph{set-ordered odd-graceful Topcode-matrix} if Cond-\ref{asparae:odd-no-same}, Cond-\ref{asparae:set-ordered} and Cond-\ref{asparae:odd-graceful} hold true.
\item A Topcode-matrix $T_{code}$ is called an \emph{elegant Topcode-matrix} if Cond-\ref{asparae:no-same} and Cond-\ref{asparae:elegant} hold true.
\item A Topcode-matrix $T_{code}$ is called an \emph{edge-magic total Topcode-matrix} if Cond-\ref{asparae:edge-magic-total-1} and Cond-\ref{asparae:edge-magic-total-2} hold true.
\item A Topcode-matrix $T_{code}$ is called a \emph{super edge-magic total Topcode-matrix} if Cond-\ref{asparae:super-edge-magic-total} and Cond-\ref{asparae:edge-magic-total-2} hold true.
\item A Topcode-matrix $T_{code}$ is called an \emph{odd-edge-magic total Topcode-matrix} if Cond-\ref{asparae:full}, Cond-\ref{asparae:odd-no-same} and Cond-\ref{asparae:odd-graceful} hold true.
\item A Topcode-matrix $T_{code}$ is called an \emph{edge sum-difference Topcode-matrix} if Cond-\ref{asparae:edge-magic-total-1} and Cond-\ref{asparae:edge-vs-difference} hold true.
\item A Topcode-matrix $T_{code}$ is called an \emph{odd-elegant Topcode-matrix} if Cond-\ref{asparae:odd-no-same} and Cond-\ref{asparae:odd-elegant} hold true.
\item A Topcode-matrix $T_{code}$ is called a \emph{harmonious Topcode-matrix} if Cond-\ref{asparae:graceful-no-same} and Cond-\ref{asparae:elegant} hold true.
\item A Topcode-matrix $T_{code}$ is called a \emph{perfect odd-graceful Topcode-matrix} if Cond-\ref{asparae:odd-no-same}, Cond-\ref{asparae:only-E-odd} and $\{|a-b|:~a,b\in (XY)^*\}=[1,|(XY)^*|]$ hold true.
\end{asparaenum}

\vskip 0.4cm

\begin{figure}[h]
\centering
\includegraphics[width=8.6cm]{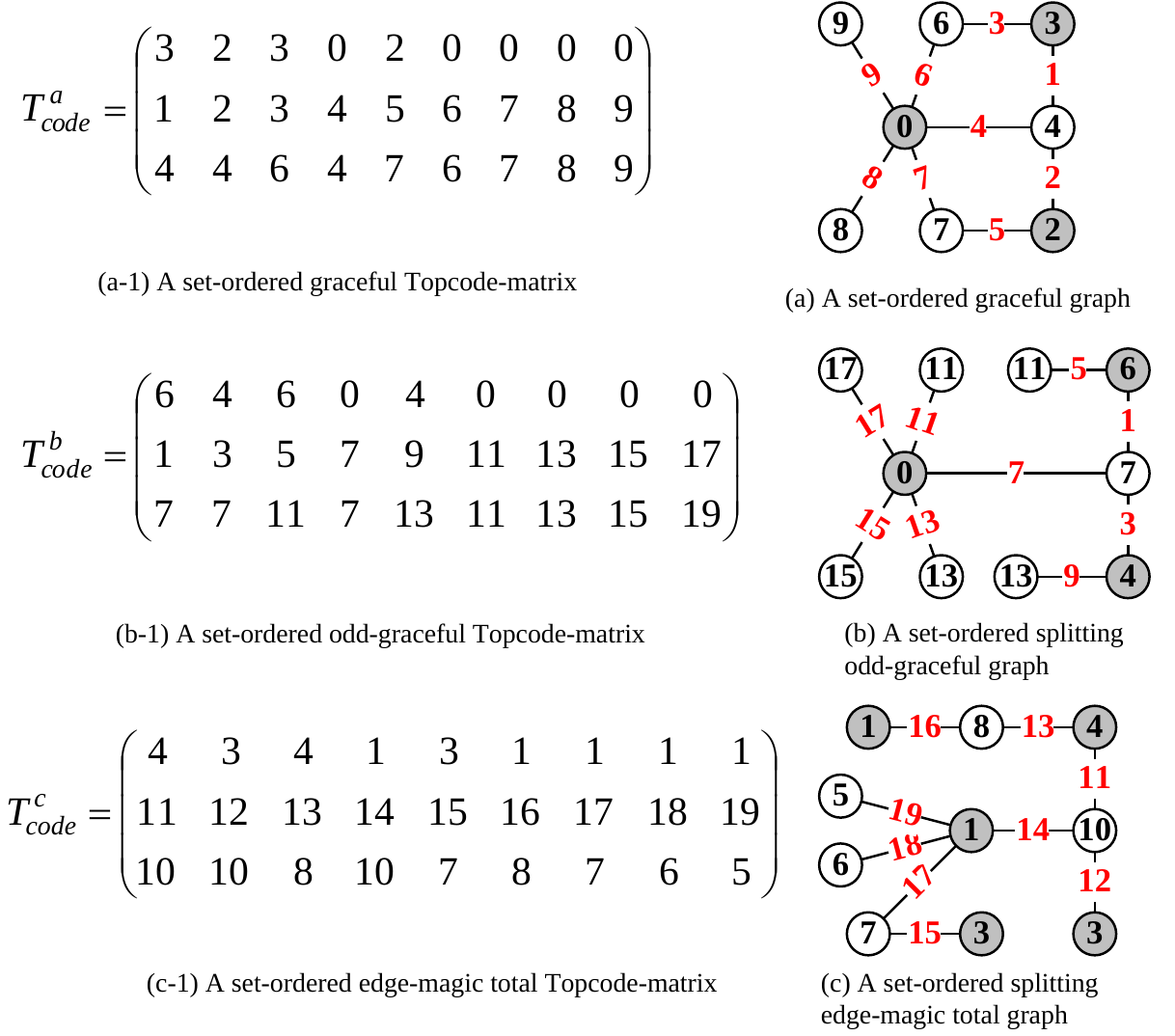}\\
\caption{\label{fig:example-equivalent-1} {\small According to six Topsnut-gpws shown in Fig.\ref{fig:1-example}, here: (a) A set-ordered graceful graph and (a-1) a set-ordered graceful Topcode-matrix; (b) a set-ordered splitting odd-graceful graph and (b-1) a set-ordered odd-graceful Topcode-matrix; (c) a set-ordered splitting edge-magic total graph and (c-1) a set-ordered edge-magic total Topcode-matrix.}}
\end{figure}

\begin{figure}[h]
\centering
\includegraphics[width=8.6cm]{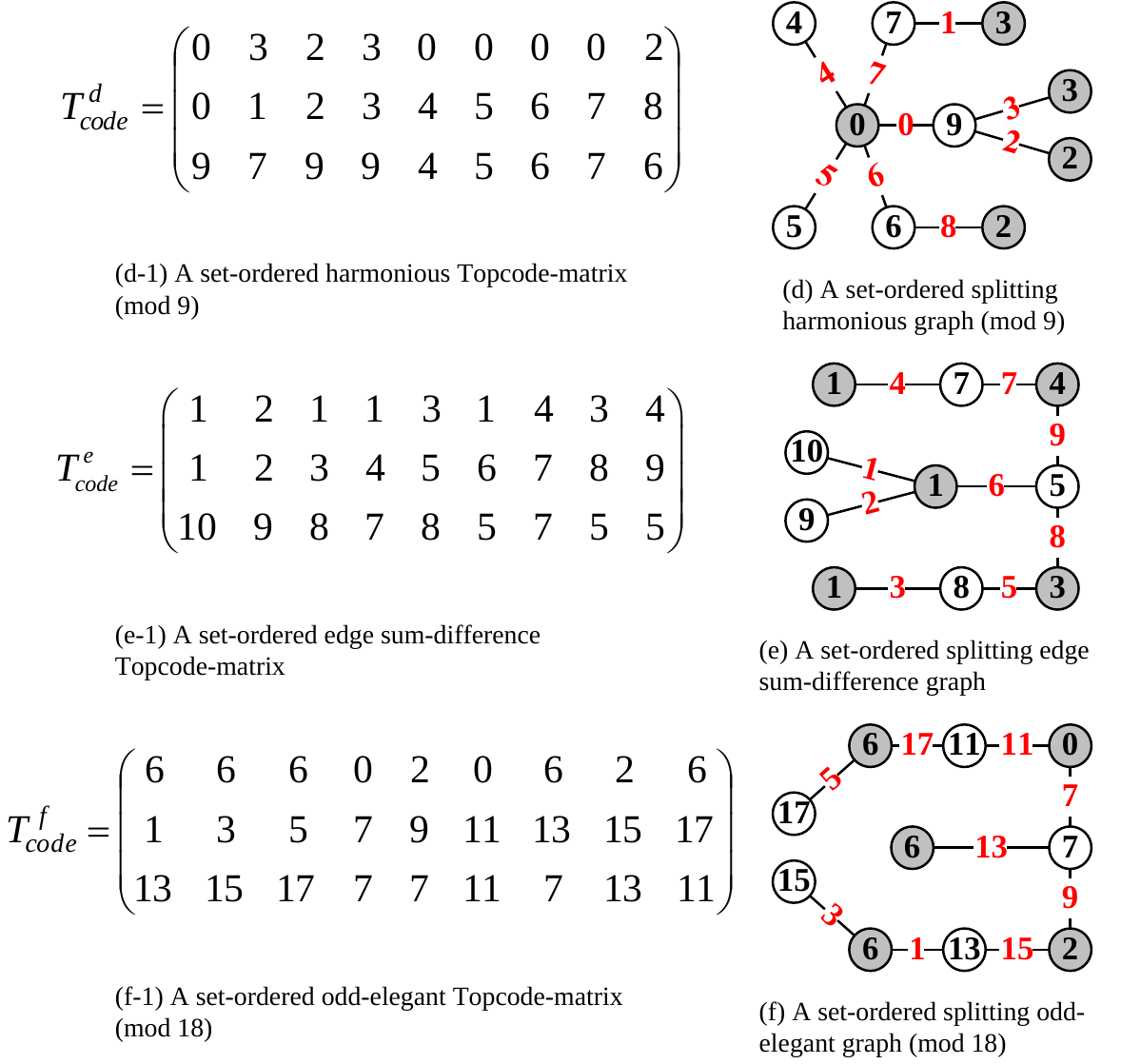}\\
\caption{\label{fig:example-equivalent-2} {\small According to six Topsnut-gpws shown in Fig.\ref{fig:1-example}, here: (d) A set-ordered splitting harmonious graph ($\bmod~9$) and (d-1) a set-ordered harmonious Topcode-matrix ($\bmod~9$); (e) a set-ordered splitting edge sum-difference graph and (e-1) a set-ordered edge sum-difference Topcode-matrix; (f) a set-ordered splitting odd-elegant graph ($\bmod~18$) and (f-1) a set-ordered odd-elegant Topcode-matrix ($\bmod~18$).}}
\end{figure}

In Fig.\ref{fig:example-equivalent-2}, a set-ordered odd-elegant Topcode-matrix ($\bmod~18$) (f-1) induces a TB-paw
$$T^1_b(\textrm{f-1})=6660206261715131197531131517771171311.$$
And we have another TB-paw
$$T^2_b(\textrm{f-1})=6113153665177702971111061371315261711$$
from the set-ordered odd-elegant Topcode-matrix ($\bmod~18$) (f-1) shown in Fig.\ref{fig:example-equivalent-2}.

\vskip 0.4cm

\subsubsection{Topcode-matrices with more restrictions}\quad Here, we show particular Topcode-matrices with more restrictions:

\begin{asparaenum}[Comp-1. ]
\item \label{asparae:strongly-matching} There exists a subset $M(T_{code})=\{e_{k_i}:i\in[1,m],e_{k_i}\in E^*\}\subset E^*$ such that each $x_i$ in $X$ is an end of some $e_{k_i}\in M(T_{code})$, and each $y_j$ in $Y$ is an end of some $e_{k_j}\in M(T_{code})$.
\item \label{asparae:strongly-constant} There exists a constant $k$ such that two ends $x_{k_i},y_{k_i}$ of each $e_{k_i}\in M(T_{code})$ hold $x_{k_i}+y_{k_i}=k$.
\item \label{asparae:(k,d)-11} Each $e_i\in E^*$ is valuated as $e_i=|x_i-y_i|$.
\item \label{asparae:relaxed} $e_i=|x_j-y_j|$ or $e_i=2M-|x_j-y_j|$ for each $i\in [1,q]$ and some $j\in [1,q]$, and a constant $M=\theta(q,|(XY)^*|)$.
\item \label{asparae:ee-difference-00} (ee-difference) Each $e_i$ with ends $x_{i}$ and $y_{i}$ matches with another $e_j$ with ends $x_{j}$ and $y_{j}$ holding $e_i=2q+|x_j-y_j|$, or $e_i=2q-|x_j-y_j|$, $e_i+e_j=2q$ for each $i\in [1,q]$ and some $j\in [1,q]$.
\item \label{asparae:(k,d)-22} Each $e_i\in E^*$ is valuated as $e_i=x_i+y_i~(\bmod~q)$.
\item \label{asparae:edge-magic-graceful} There exists a constant $k$ such that $|x_i+y_i-e_i|=k$ for each $i\in [1,q]$.
\item \label{asparae:total-graceful} $|(XY)^*|=p$, $|E^*|=q$ and $(XY)^*\cup E^*=[1,p+q]$.
\item \label{asparae:total-graceful-k-d} $|(XY)^*|=p$, $|E^*|=q$ and $(XY)^*\cup E^*\subseteq [0,k+(q-1)d]$.
\item \label{asparae:edge-plus-difference} (e-magic) There exists a constant $k$ such that $e_i+|x_i-y_i|=k$ for $i\in [1,q]$.
\item \label{asparae:edge-magic-total-graceful} (total magic) There exists a constant $k$ such that $x_i+e_i+y_i=k$ with $i\in [1,q]$.
\item \label{asparae:EV-ordered} (EV-ordered) $\min (XY)^*>\max E^*$, or $\max (XY)^*$ $<\min E^*$, or $(XY)^*\subseteq E^*$, or $E^*\subseteq(XY)^*$, or $(XY)^*$ is an odd-set and $E^*$ is an even-set.
\item \label{asparae:ee-balanced} (ee-balanced) Let $s(e_i)=|x_i-y_i|-e_i$ for $i\in [1,q]$. There exists a constant $k'$ such that each $e_i$ with $i\in [1,q]$ matches with another $e_j$ holding $s(e_i)+s(e_j)=k'$ (or $2(q+|(XY)^*|)+s(e_i)+s(e_j)=k'$, or $(|(XY)^*|+q+1)+s(e_i)+s(e_j)=k'$) true;

\item \label{asparae:ve-matching} (ve-matching) there exists a constant $k''$ such that each $e_i\in E^*$ matches with $w\in (XY)^*$, $e_i+w=k''$, and each vertex $z\in (XY)^*$ matches with $e_t\in E^*$ such that $z+e_t=k''$, except the \emph{singularity} $w=\lfloor \frac{|(XY)^*|+q+1}{2}\rfloor $.
\item \label{asparae:(k,d)-felicitous} $e_i+k=[x_i+y_i-k]~(\bmod~qd)$ for $i\in [1,q]$.
\item \label{asparae:(k,d)-00} $E^*=S_{k,d}$.
\item \label{asparae:(k,d)-33} $x_i+e_i+y_i\in S_{2k,2d}$, $i\in [1,q]$.
\item \label{asparae:(k,d)-arithmetic} $e_i=x_i+y_i\in S_{k,d}$, $i\in [1,q]$.
\item \label{asparae:v-general-e-odd} $(XY)^*\subset [0,q-1]$, $E^*=[1,2q-1]^o$.
\item \label{asparae:odd-6C} Each $e_i$ is odd, $(XY)^*\cup E^*\subset [1,4q-1]$.
\item \label{asparae:total-consecutive} $\{x_i+e_i+y_i:~i\in [1,q]\}=[a,b]$, $q=b-a+1$.
\end{asparaenum}

\vskip 0.2cm

By the above group of restrictions, we can define the following particular Topcode-matrices:
\begin{asparaenum}[\textrm{Parameter}-1. ]
\item A graceful Topcode-matrix $T_{code}$ is called a \emph{strongly graceful Topcode-matrix} if Comp-\ref{asparae:strongly-matching} and Comp-\ref{asparae:strongly-constant} hold true.
\item An odd-graceful Topcode-matrix $T_{code}$ is called a \emph{strongly odd-graceful Topcode-matrix} if Comp-\ref{asparae:strongly-matching} and Comp-\ref{asparae:strongly-constant} hold true.

\item A Topcode-matrix $T_{code}$ is called a \emph{$(k,d)$-graceful Topcode-matrix} if Comp-\ref{asparae:(k,d)-11} and Comp-\ref{asparae:(k,d)-00} hold true, see Fig.\ref{fig:k-d-4-topcode-matrices} (a).
\item A Topcode-matrix $T_{code}$ is called a \emph{$(k,d)$-felicitous Topcode-matrix} if Comp-\ref{asparae:(k,d)-felicitous} and Comp-\ref{asparae:(k,d)-00} hold true, see Fig.\ref{fig:k-d-4-topcode-matrices} (b).
\item $^*$ A Topcode-matrix $T_{code}$ is called a \emph{$(k,d)$-edge-magic total Topcode-matrix} if Comp-\ref{asparae:edge-magic-total-graceful}, Comp-\ref{asparae:ve-matching} and Comp-\ref{asparae:(k,d)-00} hold true, see Fig.\ref{fig:k-d-4-topcode-matrices} (c).

\item A Topcode-matrix $T_{code}$ is called a \emph{$(k,d)$-edge antimagic total Topcode-matrix} if it holds Comp-\ref{asparae:total-graceful-k-d}, Comp-\ref{asparae:(k,d)-00} and Comp-\ref{asparae:(k,d)-33} true, see Fig.\ref{fig:k-d-4-topcode-matrices} (d).

\begin{figure}[h]
\centering
\includegraphics[width=8.6cm]{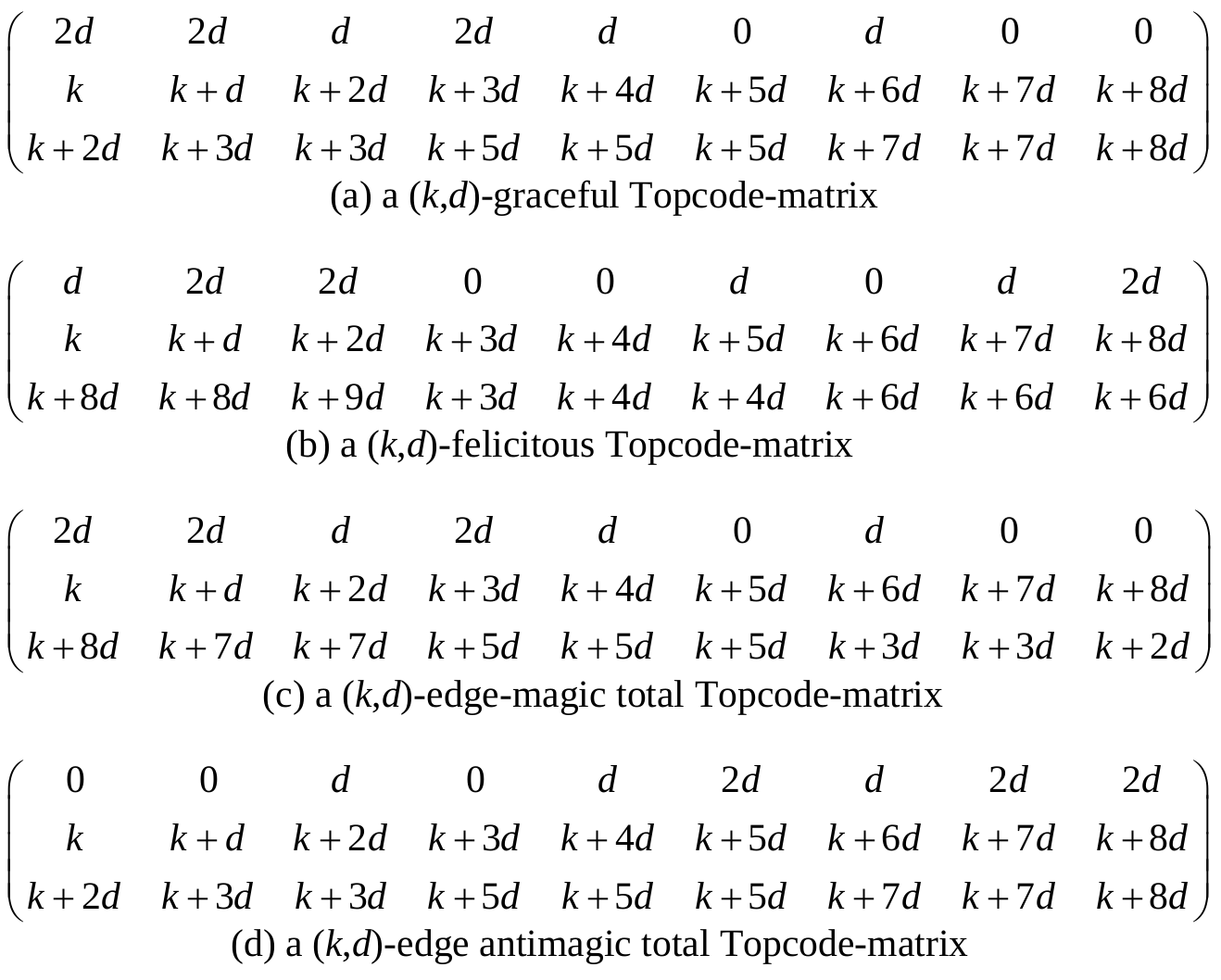}\\
\caption{\label{fig:k-d-4-topcode-matrices} {\small (a) A $(k,d)$-graceful Topcode-matrix; (b) a $(k,d)$-felicitous Topcode-matrix; (c) a $(k,d)$-edge-magic total Topcode-matrix; (d) a $(k,d)$-edge antimagic Topcode-matrix.}}
\end{figure}

\item A Topcode-matrix $T_{code}$ is called a \emph{total graceful Topcode-matrix} if Comp-\ref{asparae:(k,d)-11} and Comp-\ref{asparae:total-graceful} hold true, see Fig.\ref{fig:3-topcode-matrices} (a).
\item A Topcode-matrix $T_{code}$ is called a \emph{ve-magic total graceful Topcode-matrix} if Comp-\ref{asparae:total-graceful} and Comp-\ref{asparae:edge-plus-difference} hold true, see Fig.\ref{fig:3-topcode-matrices} (b).
\item A Topcode-matrix $T_{code}$ is called a \emph{relaxed edge-magic total Topcode-matrix} if it holds Comp-\ref{asparae:total-graceful}, Comp-\ref{asparae:edge-magic-total-graceful} and Comp-\ref{asparae:relaxed} true, see Fig.\ref{fig:3-topcode-matrices} (c).

\begin{figure}[h]
\centering
\includegraphics[width=7.6cm]{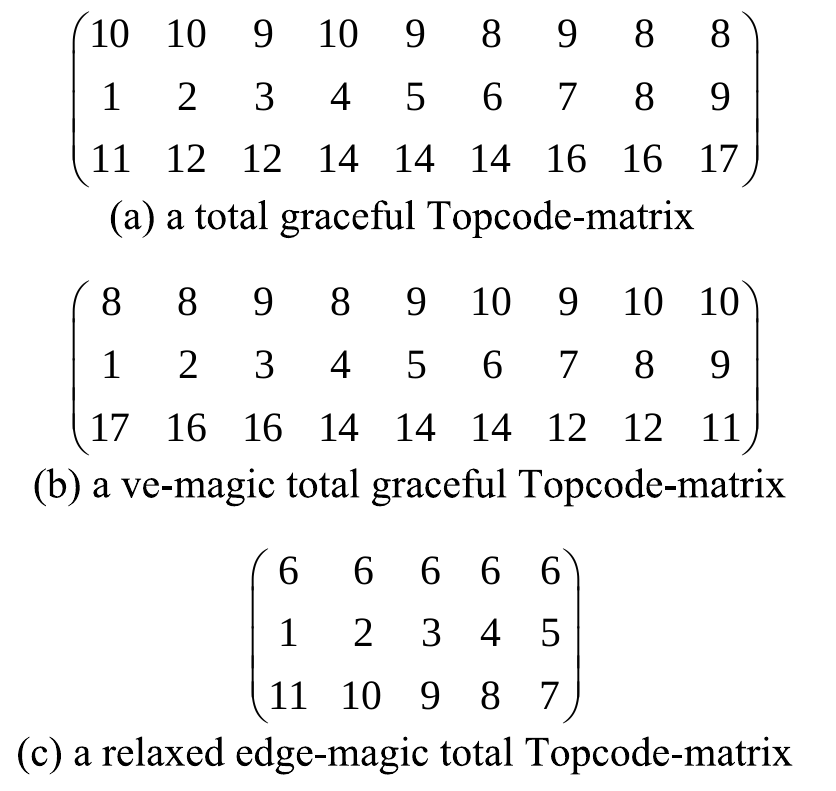}\\
\caption{\label{fig:3-topcode-matrices} {\small (a) A total graceful Topcode-matrix; (b) a ve-magic total graceful Topcode-matrix; (c) a relaxed edge-magic total Topcode-matrix.}}
\end{figure}

\item A Topcode-matrix $T_{code}$ is called an \emph{edge-magic graceful Topcode-matrix} if it holds Comp-\ref{asparae:edge-magic-graceful} and Comp-\ref{asparae:total-graceful} true.

\item A Topcode-matrix $T_{code}$ is called a \emph{6C-Topcode-matrix} if it holds Comp-\ref{asparae:relaxed}, Comp-\ref{asparae:total-graceful}, Cond-\ref{asparae:set-ordered}, Comp-\ref{asparae:edge-plus-difference}, Comp-\ref{asparae:EV-ordered}, Comp-\ref{asparae:ee-balanced} and Comp-\ref{asparae:ve-matching} true.

\vskip 0.2cm

\quad A 6C-Topcode-matrix $A$ is shown in Fig.\ref{fig:6C-matching-matrix-1}, and $A$ matches with its \emph{dual Topcode-matrix} $A^{*}$, since the sum of each element of $A$ and its corresponding element of $A^{*}$ is just 26. According to the definition of a 6C-Topcode-matrix, the Topcode-matrix $A=(X~W~Y)^{T}$ holds a \emph{6C-restriction}: (i) $e_i+|x_i-y_i|=13$; (ii) $e_i=|x_j-y_j|$; (iii) $(|x_i-y_i|-e_i)+(|x_j-y_j|-e_j)=0$; (iv) $\min (X\cup Y)>\max W$; (v) $e_i+x_s=26$ or $e_i+y_t=26$; (vi) $\min X>\max Y$. However, the dual Topcode-matrix $A^{*}=(X'~W'~Y')^{T}$ holds $e'_i-|x'_i-y'_i|=13$ and $\min (X'\cup Y')< \max W'$ only.

\begin{figure}[h]
\centering
\includegraphics[width=8.6cm]{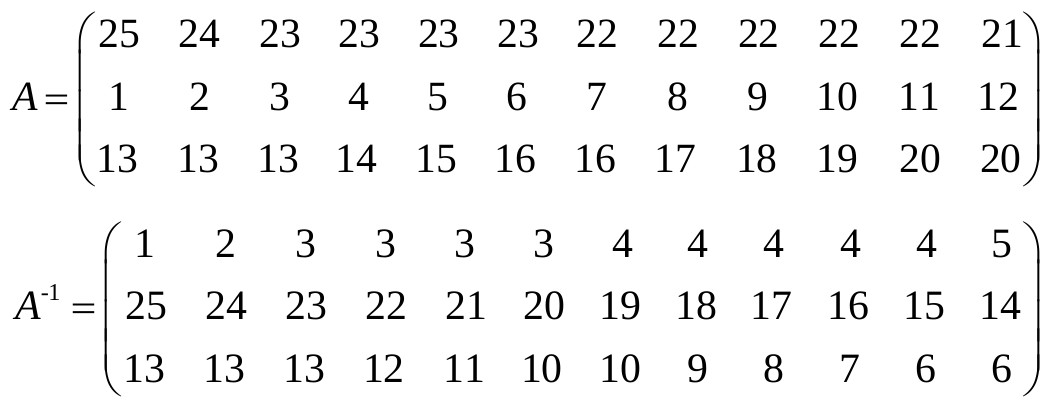}
\caption{\label{fig:6C-matching-matrix-1}{\small A 6C-Topcode-matrix $A$ and its dual Topcode-matrix $A^{*}$.}}
\end{figure}
\item A Topcode-matrix $T_{code}$ is called an \emph{odd-6C Topcode-matrix} if it holds $\{|a-b|:a,b\in (XY)^*\}=[1,2q-1]^o$, Comp-\ref{asparae:ee-difference-00}, Cond-\ref{asparae:set-ordered}, Comp-\ref{asparae:edge-plus-difference}, Comp-\ref{asparae:EV-ordered}, Comp-\ref{asparae:odd-6C}, Comp-\ref{asparae:ee-balanced}and there are two constants $k_1,k_2$ such that each $e_i$ matches with $w\in (XY)^*$ such that $e_i+w=k_1~(\textrm{or }k_2)$ true (see Fig.\ref{fig:6C-matching-matrix-odd}).
\begin{figure}[h]
\centering
\includegraphics[width=8.6cm]{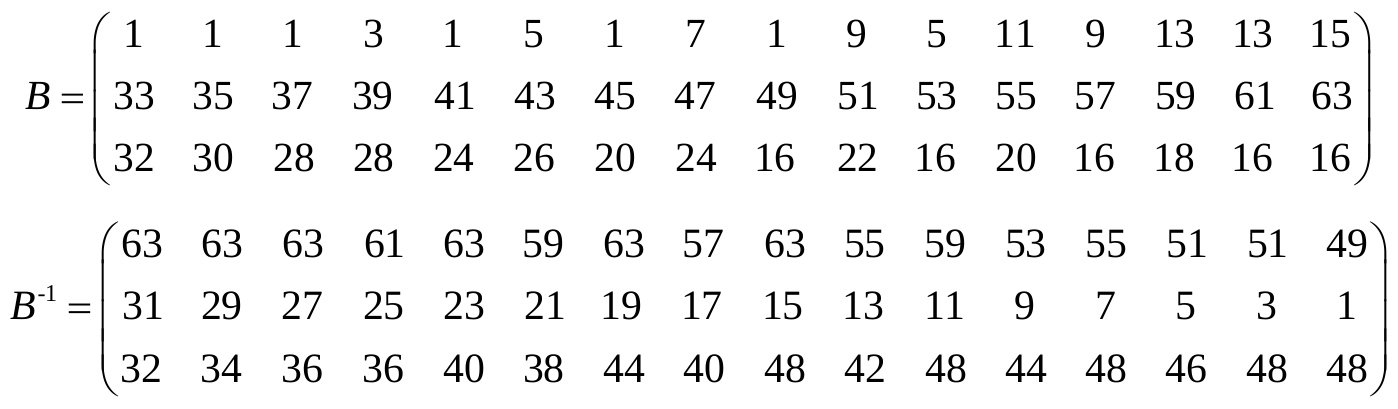}
\caption{\label{fig:6C-matching-matrix-odd}{\small An odd-6C-Topcode-matrix $B$ and its dual Topcode-matrix $B^{-1}$.}}
\end{figure}
\item A Topcode-matrix $T_{code}$ is called an \emph{ee-difference odd-edge-magic matching Topcode-matrix} if Comp-\ref{asparae:relaxed}, Comp-\ref{asparae:edge-plus-difference}, Comp-\ref{asparae:ee-balanced} and Comp-\ref{asparae:v-general-e-odd} hold true.
\item A Topcode-matrix $T_{code}$ is called an \emph{odd-edge-magic matching Topcode-matrix} if Comp-\ref{asparae:edge-magic-total-graceful} and Comp-\ref{asparae:v-general-e-odd} hold true.
\item A Topcode-matrix $T_{code}$ is called an \emph{edge-odd-graceful total Topcode-matrix} if Comp-\ref{asparae:v-general-e-odd} and Comp-\ref{asparae:total-consecutive} hold true.
\item A Topcode-matrix $T_{code}$ is called an \emph{multiple edge-meaning vertex Topcode-matrix} if $(XY)^*=[0,p-1]$ with $p=|(XY)^*|$, and there are three constants $k,k'$ and $k''$ such that (1) $E^*=[1,q]$ and $x_i+e_i+y_i=k$; (2) $E^*=[p,p+q-1]$ and $x_i+e_i+y_i=k'$; (3) $E^*=[0,q-1]$ and $e_i=x_i+y_i~(\bmod~q)$; (4) $E^*=[1,q]$ and $|x_i+y_i-e_i|=k''$; (5) $E^*=[1,2q-1]^o$, and $\{x_i+e_i+y_i:~i\in [1,q]\}=[a,b]$ with $b-a+1=q$.
\end{asparaenum}

\vskip 0.4cm

\subsubsection{Matching Topcode-matrices}\quad We show some connections between two Topcode-matrices here. Let $X_0=\{0,d,2d, \dots ,(q-1)d\}$. The set $S^*$ contains all different elements in a vector (or collection) $S$, $T_{code}$ is defined in Definition \ref{defn:Topcode-matrix} hereafter.

\begin{asparaenum}[\textrm{Matching}-1 ]
\item If two Topcode-matrices $T_{code}=(X~E~Y)^{T}$ and $\overline{T}_{code}=(U~W~V)^{T}$ hold $(X_0\cup S_{k,d})\setminus (XY)^*\cup E^*=(UV)^*\cup W^*$ true, then $\overline{T}_{code}$ is called a \emph{complementary $(k,d)$-Topcode-matrix} of $T_{code}$, and $(T_{code},\overline{T}_{code})$ is called a \emph{matching of twin $(k,d)$-Topcode-matrices}.
\item Two Topcode-matrices $T^t_{code}=(X^t$ $E^t$ $Y^t)^{T}$ with $t=1,2$ correspond a $(p,q)$-graph $G$, and satisfy that $(X^tY^t)^*\subset X_0\cup S_{k,d}$, and $e^t_i-k=[x^t_i+y^t_i-k~(\textrm{mod}~qd)]$ with $t=1,2$ and $i\in[1,q]$. If $e^1_i+e^2_i=2k+(q-1)d$ with $i\in[1,q]$, then $(T^1_{code},T^2_{code})$ is called a \emph{matching of $(k,d)$-harmonious image Topcode-matrices}.
\item Two Topcode-matrices $T^t_{code}=(X^t$ $E^t$ $Y^t)^{T}$ with $t=1,2$ correspond a $(p,q)$-graph $G$, and $e^t_i=|x^t_i-y^t_i|$ with $t=1,2$ and $i\in[1,q]$. If there exists a positive constant $k$ such that $e^1_i+e^2_i=k$ with $i\in[1,q]$, then $(T^1_{code},T^2_{code})$ is called a matching of image Topcode-matrices, and $T^t_{code}$ a \emph{mirror-image} of $T^{3-t}_{code}$ with $t=1,2$.
\item A Topcode-matrix $T_{code}=(X~E~Y)^{T}=\uplus^m_{k=1} T^k_{code}$ with $m\geq 2$, where $T^k_{code}=(X^k~E^k~Y^k)^{T}$. If $(XY)^*=[0, |(XY)^*|-1]$, and each $T^k_{code}$ is a $\varepsilon$-Topcode-matrix, then $T_{code}$ is called an \emph{multiple matching Topcode-matrix}.
\end{asparaenum}

\vskip 0.4cm

\subsubsection{Directed Topcode-matrices}\quad In Fig.\ref{fig:directed-topcode-matrix}, $\overrightarrow{T}$ is a \emph{directed Topsnut-gpw}, and it corresponds a \emph{directed Topcode-matrix} $A(\overrightarrow{T})$. In directed graph theory, the out-degree is denoted by ``$+$'', and the in-degree is denoted by ``$-$''. So, $\overrightarrow{T}$ has $d^+(22)=5$, $d^-(22)=0$, $d^-(13)=3$ and $d^+(13)=0$.

\begin{figure}[h]
\centering
\includegraphics[width=8.6cm]{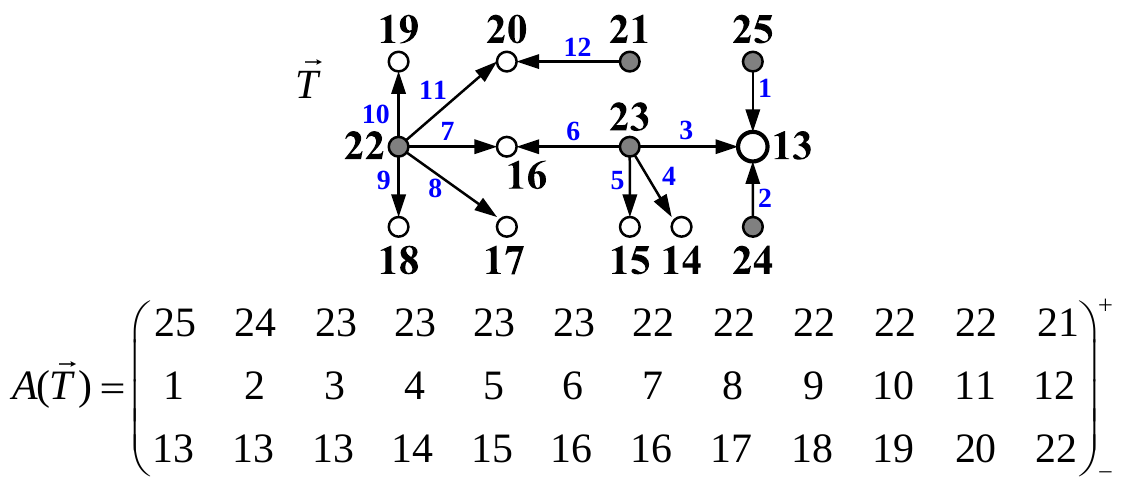}
\caption{\label{fig:directed-topcode-matrix}{\small A directed Topsnut-gpw with a directed Topcode-matrix.}}
\end{figure}

We show the definition of a directed Topcode-matrix as follows:
\begin{defn}\label{defn:directed-Topcode-matrix}
$^*$ A \emph{directed Topcode-matrix} is defined as
\begin{equation}\label{eqa:Topcode-dimatrix}
\centering
{
\begin{split}
\overrightarrow{T}_{code}&= \left(
\begin{array}{ccccc}
x_{1} & x_{2} & \cdots & x_{q}\\
e_{1} & e_{2} & \cdots & e_{q}\\
y_{1} & y_{2} & \cdots & y_{q}
\end{array}
\right)^{+}_{-}=
\left(\begin{array}{c}
X\\
\overrightarrow{E}\\
Y
\end{array} \right)^{+}_{-}\\
&=[(X~\overrightarrow{E}~Y)^{+}_{-}]^{-1}
\end{split}}
\end{equation}\\
where \emph{v-vector} $X=(x_1 ~ x_2 ~ \cdots ~x_q)$, \emph{v-vector} $Y=(y_1 $ ~ $y_2$ ~ $\cdots $ ~ $y_q)$ and \emph{di-e-vector} $\overrightarrow{E}=(e_1$ ~ $e_2 $ ~ $ \cdots $ ~ $e_q)$, such that each $e_i$ has its head $x_i$ and its tail $y_i$ with $i\in [1,q]$.\qqed
\end{defn}

The study of directed graphs is not more than that of non-directed graphs, although directed graphs are useful and powerful in real applications (Ref. \cite{Bang-Jensen-Gutin-digraphs-2007}). It may be interesting to apply directed Topcode-matrices to Graph Networks and Graph Neural Networks.

\vskip 0.4cm

\subsubsection{Pan-Topcode-matrices with elements are not numbers}\quad We consider some particular Topcode-matrices having elements being sets, Hanzis, graphs, groups and so on. Let $T_{pcode}=(P_X~P_E~P_Y~)^{T}$ be a pan-Topcode-matrix, where three vectors $P_X=(\alpha_1~\alpha_2~\dots ~\alpha_q)$, $P_E=(\gamma_1~\gamma_2~\dots ~\gamma_q)$ and $P_Y=(\beta_1~\beta_2~\dots ~\beta_q)$. We set $(P_XP_Y)^*$ to be the set of different elements of the union set $P_X\cup P_Y$, $P^*_E$ to be the set of different elements of the vector $P_E$.

\vskip 0.4cm

\textbf{Sets:}

\begin{asparaenum}[(\textrm{Set}-1) ]
\item A pan-Topcode-matrix $T_{pcode}$ has $(P_XP_Y)^*\subseteq [1,q]^2$ and $P^*_E\subseteq [1,q]^2$, and $\gamma_i=\alpha_i\cap \beta_i$. If we can select a \emph{representative} $a_{i}\in \gamma_i$ such that the representative set $\{a_{i}:~\gamma_i\in P^*_E\}=[1,q]$, then $T_{pcode}$ is called a \emph{graceful intersection total set-Topcode-matrix}.
\item A pan-Topcode-matrix $T_{pcode}$ has $(P_XP_Y)^*\subseteq [1,q]^2$ and $P^*_E\subseteq [1,q]^2$, and $\gamma_i=\alpha_i\cap \beta_i$, $|\gamma_i|=|\{e_i\}|=1$ and $\bigcup ^q_{i=1}\gamma_i=[1,q]$ (resp. $[1,2q-1]^o$), we call $T_{pcode}$ a \emph{set-graceful pan-Topcode-matrix} (resp. a \emph{set-odd-graceful Topcode-matrix}).
\item A pan-Topcode-matrix $T_{pcode}$ has $(P_XP_Y)^*\subseteq [1,q]^2$ and $P^*_E\subseteq [1,q]^2$, and $\gamma_i=\alpha_i\setminus \beta_i$ or $\gamma_i=\beta_i\setminus \alpha_i$, $\gamma_i\cap \gamma_j=\emptyset$ for $i\neq j$ and $\bigcup ^q_{i=1}\gamma_i=[a,b]$, we call $T_{pcode}$ a \emph{set-subtraction graceful pan-Topcode-matrix}.

\item A pan-Topcode-matrix $T_{pcode}$ has $(P_XP_Y)^*\subseteq [1,2q]^2$ and $P^*_E\subseteq [1,2q-1]^2$, and $\gamma_i=\alpha_i\cap \beta_i$. If we can select a \emph{representative} $a_{i}\in \gamma_i$ such that the representative set $\{a_{i}:~\gamma_i\in P^*_E\}=[1,2q-1]^o$, then $T_{pcode}$ is called an \emph{odd-graceful intersection total set-Topcode-matrix}.
\item A \emph{set-intersecting rainbow Topcode-matrix} $C$ is shown in Fig.\ref{fig:rainbow-matrix}, and it corresponds a tree $H$ with a \emph{set coloring} $h: V(H)\rightarrow S$, where $S=\{[1,1],[1,2],\dots ,[1,13]\}$, such that $h(uv)=h(u)\cap h(v)$ for each edge $uv\in E(H)$, and $h(E(H))=\{[1,1],[1,2],\dots ,[1,12]\}$. Similarly, we can have a set-union rainbow Topcode-matrix by defining $h(uv)=h(u)\cup h(v)$.

\begin{figure}[h]
\centering
\includegraphics[width=8.6cm]{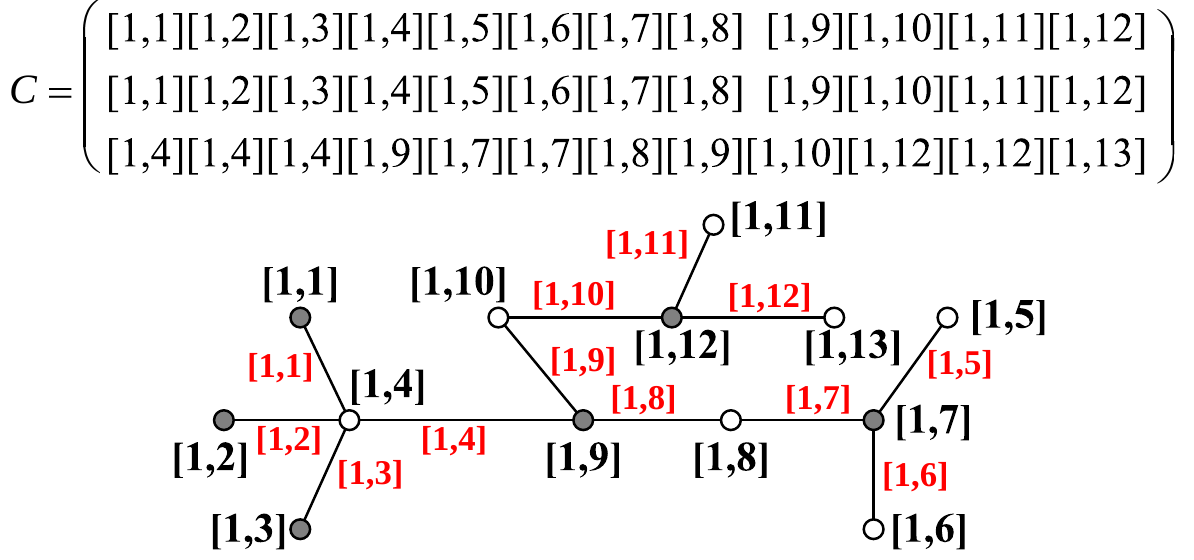}
\caption{\label{fig:rainbow-matrix}{\small A set-intersecting rainbow Topcode-matrix $C$.}}
\end{figure}

\item We call $T_{code}$ defined in Definition \ref{defn:Topcode-matrix} a \emph{v-distinguishing Topcode-matrix} if $N_{ei}(x_i)\neq N_{ei}(w_j)$ for any pair of distinct $x_i,w_j\in (XY)^*$. For each $(x_i$ $e_i$ $y_i)^{T}$ $\in T_{code}$, $T_{code}$ is called an \emph{adjacent v-distinguishing Topcode-matrix} if $N_{ei}(x_i)\neq N_{ei}(y_i)$; $T_{code}$ is called an \emph{adjacent e-distinguishing Topcode-matrix} if $N'_{ei}(x_i)\neq N'_{ei}(y_i)$; $T_{code}$ is called an \emph{adjacent total-distinguishing Topcode-matrix} if $N_{ei}(x_i)\cup N'_{ei}(x_i)\neq N_{ei}(y_i)\cup N'_{ei}(y_i)$ for each $e_i$ with $i\in [1,q]$. It is easy to see such examples, for instance, all Topcode-matrices of Topmatrix-$i$, Parameter-$j$ and Matching-$k$ are distinguishing.
\end{asparaenum}

\vskip 0.4cm

\textbf{Groups: } Topcode$^+$-matrix groups and Topcode$^-$-matrix groups are defined by the additive v-operation ``$\oplus$'' defined in (\ref{eqa:group-operation-1}) and (\ref{eqa:group-operation-2}) and the subtractive v-operation ``$\ominus$'' defined in (\ref{eqa:subtraction-group-operation-1}) and (\ref{eqa:subtraction-group-operation-2}), respectively. See a Topcode$^+$-matrix group shown in Fig.\ref{fig:topcode-matrix-group}, and an every-zero graphic group shown in Fig.\ref{fig:graph-group}.

Suppose that $\{F_{n}(H);\oplus\}$ is an every-zero graphic group based on $F_{n}(H)=\{H_i:i\in [1,n]\}$ with $n\geq q$ ($n\geq 2q-1$) and the additive operation ``$\oplus$'', and moreover $\{F_{n}(H);\ominus\}$ is an every-zero graphic group based on the subtractive operation ``$\ominus$''.

\begin{asparaenum}[\textrm{Gcond}-1. ]
\item \label{asparae:vertex-restriction} $(P_XP_Y)^*\subseteq F_{n}(H)$.
\item \label{asparae:edge-restriction} $P^*_E\subseteq F_{n}(H)$.
\item \label{asparae:edge-restriction-grace} $P^*_E=\{H_1,H_2,\dots ,H_q\}$.
\item \label{asparae:edge-restriction-odd-grace} $P^*_E=\{H_1,H_3,\dots ,H_{2q-1}\}$.
\item \label{asparae:additive-operation} $\gamma_i=\alpha_i\oplus \beta_i$.
\item \label{asparae:subtractive-operation} $\gamma_i=\alpha_i\ominus \beta_i$.
\item \label{asparae:vertex-grace} $(P_XP_Y)^*=\{H_1,H_2,\dots ,H_q\}$.
\item \label{asparae:vertex-odd-grace} $(P_XP_Y)^*=\{H_1,H_3,\dots ,H_{2q-1}\}$.
\item \label{asparae:vertex-set-ordered} $\max\{i:\alpha_i=H_i\}<\min\{j:\beta_j=H_j\}$.
\end{asparaenum}

\vskip 0.4cm

We have

\begin{asparaenum}[(\textrm{Group}-1) ]
\item A pan-Topcode-matrix $T_{pcode}$ is called an \emph{e-graceful group Topcode$^+$-matrix} if Gcond-\ref{asparae:vertex-restriction}, Gcond-\ref{asparae:edge-restriction-grace} and Gcond-\ref{asparae:additive-operation} hold true.
\item A pan-Topcode-matrix $T_{pcode}$ is called an \emph{e-odd-graceful group Topcode$^+$-matrix} if Gcond-\ref{asparae:vertex-restriction}, Gcond-\ref{asparae:edge-restriction-odd-grace} and Gcond-\ref{asparae:additive-operation} hold true.

\item A pan-Topcode-matrix $T_{pcode}$ is called an \emph{e-graceful group Topcode$^-$-matrix} if Gcond-\ref{asparae:vertex-restriction}, Gcond-\ref{asparae:edge-restriction-grace} and Gcond-\ref{asparae:subtractive-operation} hold true.
\item A pan-Topcode-matrix $T_{pcode}$ is called an \emph{e-odd-graceful group Topcode$^-$-matrix} if Gcond-\ref{asparae:vertex-restriction}, Gcond-\ref{asparae:edge-restriction-odd-grace} and Gcond-\ref{asparae:subtractive-operation} hold true.

\item An e-graceful group Topcode$^+$-matrix $T_{pcode}$ is called a \emph{ve-graceful group Topcode$^+$-matrix} if Gcond-\ref{asparae:vertex-grace} holds true.
\item An e-odd-graceful group Topcode$^+$-matrix $T_{pcode}$ is called a \emph{ve-odd-graceful group Topcode$^+$-matrix} if Gcond-\ref{asparae:vertex-odd-grace} holds true.

\item An e-graceful group Topcode$^-$-matrix $T_{pcode}$ is called a \emph{ve-graceful group Topcode$^-$-matrix} if Gcond-\ref{asparae:vertex-grace} holds true.
\item An e-odd-graceful group Topcode$^-$-matrix $T_{pcode}$ is called a \emph{ve-odd-graceful group Topcode$^-$-matrix} if Gcond-\ref{asparae:vertex-odd-grace} holds true.

\item A Topcode-matrix $T_{pcode}$ is called a \emph{set-ordered group Topcode-matrix} if Gcond-\ref{asparae:vertex-restriction}, Gcond-\ref{asparae:edge-restriction} and Gcond-\ref{asparae:vertex-set-ordered} hold true.

\item An $\epsilon$-matrix $T_{pcode}$ is called a \emph{set-ordered $\varepsilon$-matrix} if Gcond-\ref{asparae:vertex-set-ordered} holds true, where $\varepsilon\in \{$e-graceful group Topcode$^+$, e-odd-graceful group Topcode$^+$, ve-graceful group Topcode$^+$, ve-odd-graceful group Topcode$^+$, e-graceful group Topcode$^-$, e-odd-graceful group Topcode$^-$, ve-graceful group Topcode$^-$, ve-odd-graceful group Topcode$^-$$\}$.
\end{asparaenum}

\begin{rem}\label{eqa:group-Topcode-matrix}
(1) Since each every-zero graphic group $\{F_{n}(H);\oplus\}$ corresponds an every-zero Topcode-matrix group $\{F_{n}(T_{code}(H));\oplus\}$, we can define various $\varepsilon$-matrix $T_{pcode}$ based on every-zero Topcode-matrix groups, where $\varepsilon\in \{$e-graceful matrix-group Topcode$^+$, e-odd-graceful matrix-group Topcode$^+$, ve-graceful matrix-group Topcode$^+$, ve-odd-graceful matrix-group Topcode$^+$, e-graceful matrix-group Topcode$^-$, e-odd-graceful matrix-group Topcode$^-$, ve-graceful matrix-group Topcode$^-$, ve-odd-graceful matrix-group Topcode$^-$$\}$.

(2) Since each every-zero number string group can be derived from Topcode-matrices of Topsnut-gpws, similarly, we can make various $\varepsilon$-matrix $T_{pcode}$ based on every-zero number string groups, where $\varepsilon\in \{$e-graceful string-group Topcode$^+$, e-odd-graceful string-group Topcode$^+$, ve-graceful string-group Topcode$^+$, ve-odd-graceful string-group Topcode$^+$, e-graceful string-group Topcode$^-$, e-odd-graceful string-group Topcode$^-$, ve-graceful string-group Topcode$^-$, ve-odd-graceful string-group Topcode$^-$$\}$.

(3) It is possible to define more pan-Topcode-matrices $T_{pcode}$ by simulating labelling/coloring of graph theory.
\end{rem}

\begin{rem}\label{eqa:Graphicable-Topcode-matrix}
Definition \ref{defn:Topcode-matrix} enables us to define the following concepts:
\begin{asparaenum}[Def-1. ]
\item A \emph{Topsnut-gpw} is defined by a no-colored $(p,q)$-graph $G$ and an evaluated Topcode-matrix $T_{code}$ defined in Definition \ref{defn:Topcode-matrix} such that $G$ can be evaluated by $T_{code}$.
\item If a colored $(p,q)$-graph $G$ with a coloring $f$ has its own Topsnut-matrix $A_{vev}(G)$ holding $A_{vev}(G)=T_{code}$ true with each $e_i=f(x_i,y_i)$ for $i\in [1,q]$, then we say the Topcode-matrix $T_{code}$ to be \emph{graphicable}.
\item A Topcode-matrix $T_{code}$ defined in Definition \ref{defn:Topcode-matrix} is an $\varepsilon$-Topcode-matrix and graphicable, and $G$ is a graph derived from $T_{code}$. If the vertex number of $G$ holds $|V(G)|> |(XY)^*|$ true, we say $T_{code}$ is a \emph{splitting $\varepsilon$-Topcode-matrix}. For example, $T_{code}$ is a splitting graceful Topcode-matrix, or a splitting odd-edge-magic total Topcode-matrix, \emph{etc.}.
\end{asparaenum}
\end{rem}

The problem of a Topcode-matrix $T_{code}$ defined in Definition \ref{defn:Topcode-matrix} to be graphicable is one of \emph{degree sequence} of graph theory. In particular, we have the following graphicable results:
\begin{thm}\label{thm:chapter1-degree-sequence-th-0}
A Topcode-matrix $T_{code}$ defined in Definition \ref{defn:Topcode-matrix} is graphicable if and only if $2q=\sum_{x\in X^*} \alpha(x)+\sum _{y\in Y^*} \alpha(y)$, where $\alpha(x)$ (resp. $\alpha(y)$) is the number of $x$ (resp. $y$) appeared in $X$ (resp. $Y$).
\end{thm}

Let $(XY)^*=\{u_{1},u_{2},\dots ,u_{p}\}$, and each number $u_i$ appears $d_i$ times in $T_{code}$. We present the famous Erd\"{o}s-Gallai degree sequence theorem (Ref. \cite{Bondy-2008}) as follows:
\begin{thm}\label{thm:chapter1-degree-sequence-iff}
A sequence $\textbf{d}=(d_1,d_2,\dots ,d_{p})$ to be the
degree sequence of a certain graph $G$ of order $p$ if $\sum
^p_{i=1}d_i$ is even and satisfies the following inequality
\begin{equation}\label{Erdos}
\sum ^k_{i=1}d_i\leq k(k-1)+\sum ^p_{i=k+1}\min\{k,d_i\},\ i\in [1,
p].
\end{equation}
\end{thm}
\begin{proof}
Let $\ud_G(u_i)=d_i$ for $u_i\in V(G)=\{u_i: i\in [1, n]\}$. The set $V_1=\{u_1, u_2, \dots, u_k\}$ distributes $k(k-1)$ in
$\sum^k_{i=1}d_i$. Each $u_j\in V_2=\{u_{k+1},u_{k+2},\dots ,u_{n}\}$ is adjacent to at most $d_i$ or $k$ vertices of $V_1$.
Therefore, we get the desired inequality (\ref{Erdos}).

Furthermore, Erd\"{o}s and Gallai, in 1960, have shown that this necessary condition is also sufficient for the sequence \emph{\textbf{d}} to be graphic
\cite{Bondy-2008}. The sufficient proof can be found in Chapter Six of ``\emph{Graph Theory}'' written by Harary
\cite{Harary-book}.
\end{proof}

\begin{thm}\label{thm:chapter1-degree-sequence-th-1}
If there is no $x_j=x_i$ and $y_r=x_i$ for $j\neq i$ and $r\neq i$ in $T_{code}$ defined in Definition \ref{defn:Topcode-matrix}, we call $x_i$ as a \emph{leaf} of $T_{code}$, and we delete $x_i,e_i,y_i$ from $T_{code}$ to obtain $T^1_{code}$ which is a $3\times (q-1)$-order Topcode-matrix; and do such deletion operation to a leaf $x^1_i$ of $T^1_{code}$ to obtain $T^2_{code}$ which is a $3\times (q-2)$-order Topcode-matrix, until, we get $T^k_{code}$ such that there is no leaf in $T^k_{code}$ which is a $3\times (q-k)$-order Topcode-matrix. If the deletion of all leaves of $T_{code}$ produces a Topcode-matrix $T^*_{code}$ corresponding a \emph{path} of graph theory, so $T_{code}$ is called a \emph{caterpillar Topcode-matrix}.
\end{thm}

\begin{thm}\label{thm:chapter1-degree-sequence-th-3}
If the Topcode-matrix $T^*_{code}$ obtained from $T_{code}$ defined in Definition \ref{defn:Topcode-matrix} by deleting all leaves of $T_{code}$, and $T^*_{code}$ corresponds a caterpillar. Then we claim that $T_{code}$ corresponds a graph, called a \emph{lobster} in graph theory.
\end{thm}

\vskip 0.4cm

\subsection{Operations on Topcode-matrices}\quad We show the following operations on Topcode-matrices:

\emph{(i) $^*$ Dual operation.} Let $M_v=\max (XY)^*$ and $m_v=\min (XY)^*$ in an evaluated Topcode-matrix $T_{code}$ defined in Definition \ref{defn:Topcode-matrix}, $T_{code}$ admits a function $f$ such that $e_i=f(x_i,y_i)$ for $i\in [1,q]$. We make two vectors $\overline{X}=(\overline{x}_1~\overline{x}_2~\dots ,~\overline{x}_q)$ with $\overline{x}_i=M_v+m_v-x_i$ for $i\in [1,q]$ and $\overline{Y}=(\overline{y}_1~\overline{y}_2~\dots ~\overline{y}_q)$ with $\overline{y}_j=M_v+m_v-y_j$ for $j\in [1,q]$, as well as $\overline{E}$ having each element $\overline{e}_i=f(M_v+m_v-\overline{x}_i,M_v+m_v-\overline{y}_i)$ for $i\in [1,q]$. The Topcode-matrix $\overline{T}_{code}=(\overline{X}\quad \overline{E}\quad \overline{Y})^{T}_{3\times q}$ is called the \emph{dual Topcode-matrix} of the Topcode-matrix $T_{code}$.

\emph{(ii) Column-exchanging operation \cite{Yao-Sun-Zhao-Li-Yan-2017}.} We exchange the positions of two columns $(x_i~e_i~y_i)^{T}$ and $(x_j~e_j~y_j)^{T}$ in $T_{code}$ defined in Definition \ref{defn:Topcode-matrix}, so we get another Topcode-matrix $T'_{code}$. In mathematical symbol, the \emph{column-exchanging operation} $c_{(i,j)}(T_{code})=T'_{code}$ is defined by
$${
\begin{split}
&\quad c_{(i,j)}(x_1 ~ x_2 ~ \cdots ~\textcolor[rgb]{0.00,0.00,1.00}{x_i}~ \cdots ~\textcolor[rgb]{0.00,0.00,1.00}{x_j}~ \cdots ~x_q)\\
&=(x_1 ~ x_2 ~ \cdots ~\textcolor[rgb]{0.00,0.00,1.00}{x_j}~ \cdots ~\textcolor[rgb]{0.00,0.00,1.00}{x_i}~ \cdots ~x_q),
\end{split}}$$
$$
{
\begin{split}
&\quad c_{(i,j)}(e_1 ~ e_2 ~ \cdots ~\textcolor[rgb]{0.00,0.00,1.00}{e_i}~ \cdots ~\textcolor[rgb]{0.00,0.00,1.00}{e_j}~ \cdots ~e_q)\\
&=(e_1 ~ e_2 ~ \cdots ~\textcolor[rgb]{0.00,0.00,1.00}{e_j}~ \cdots ~\textcolor[rgb]{0.00,0.00,1.00}{e_i}~ \cdots ~e_q),
\end{split}}
$$

and
$$
{
\begin{split}
&\quad c_{(i,j)}(y_1 ~ y_2 ~ \cdots ~\textcolor[rgb]{0.00,0.00,1.00}{y_i}~ \cdots ~\textcolor[rgb]{0.00,0.00,1.00}{y_j}~ \cdots ~y_q)\\
&=(y_1 ~ y_2 ~ \cdots ~\textcolor[rgb]{0.00,0.00,1.00}{y_j}~ \cdots ~\textcolor[rgb]{0.00,0.00,1.00}{y_i}~ \cdots ~y_q).
\end{split}}
$$

\emph{(iii) XY-exchanging operation \cite{Yao-Sun-Zhao-Li-Yan-2017}.} We exchange the positions of $x_i$ and $y_i$ of the $i$th column of $T_{code}$ defined in Definition \ref{defn:Topcode-matrix} by an \emph{XY-exchanging operation} $l_{(i)}$ defined as:
$${
\begin{split}
&\quad l_{(i)}(x_1 ~ x_2 ~ \cdots x_{i-1}~\textcolor[rgb]{0.00,0.00,1.00}{x_i}~x_{i+1} \cdots ~x_q)\\
&=(x_1 ~ x_2 ~ \cdots x_{i-1}~\textcolor[rgb]{0.00,0.00,1.00}{y_i}~x_{i+1} \cdots ~x_q)
\end{split}}
$$
and
$${
\begin{split}
&\quad l_{(i)}(y_1 ~ y_2 ~ \cdots y_{i-1}~\textcolor[rgb]{0.00,0.00,1.00}{y_i}~y_{i+1} \cdots ~y_q)\\
&=(y_1 ~ y_2 ~ \cdots y_{i-1}~\textcolor[rgb]{0.00,0.00,1.00}{x_i}~y_{i+1} \cdots ~y_q),
\end{split}}
$$
the resultant matrix is denoted as $l_{(i)}(T_{code})$.

Now, we do a series of column-exchanging operations $c_{(i_k,j_k)}$ with $k\in [1,a]$, and a series of XY-exchanging operations $l_{(i_s)}$ with $s\in [1,b]$ to a Topcode-matrix $T_{code}$ defined in Definition \ref{defn:Topcode-matrix}, the resultant Topcode-matrix is written by $C_{(c,l)(a,b)}(T_{code})$.
\begin{lem}\label{thm:grapgicable-no-isomorphic}
Suppose $T_{code}$ and $T'_{code}$ are defined in Definition \ref{defn:Topcode-matrix} and grapgicable, a graph $G$ corresponds to $T_{code}$ and another graph $H$ corresponds to $T'_{code}$. If
\begin{equation}\label{eqa:graphs-isomorphic-Topcode-matrices}
C_{(c,l)(a,b)}(T_{code}(G))=T'_{code}(H),
\end{equation}
then two graphs $G$ and $H$ may be or not be isomorphic to each other.
\end{lem}

See some examples shown in Fig.\ref{fig:1-example} for understanding Lemma \ref{thm:grapgicable-no-isomorphic}, in which six Topsnut-gpws (a)-(f) correspond the same Topcode-matrix, although they are not isomorphic to each other.

Under the XY-exchanging operation and the column-exchanging operation, a colored graph $G$ has its own \emph{standard Topcode-matrix} (representative Topcode-matrix) $A_{vev}(G)=(X~E~Y)^{T}$ such that $e_i\leq e_{i+1}$ with $i\in[1,q-1]$, $x_j<y_j$ with $j\in[1,q]$, since there exists no case $x_j=y_j$. Thereby, this graph $G$ has $m$ standard Topcode-matrices if it admits $m$ different colorings/labellings of graphs in graph theory.

\emph{(iv) $^*$ Union-addition operation.} A \emph{single matrix} $X$ is defined as
\begin{equation}\label{eqa:one-columme-matrix}
\centering
{
\begin{split}
X&= \left(
\begin{array}{ccccc}
x_{1}\\
x_{2}\\
\cdots \\
x_{n}
\end{array}
\right)=(x_{1}~x_{2}~\cdots ~x_{n})^{T}_{n\times 1}
\end{split}}
\end{equation}
Thereby, we define an operation, called \emph{union-addition operation} and denoted as $\uplus$, between matrices $X_i=(x_{i,1}~x_{i,2}~\cdots ~x_{i,n})^{T}_{n\times 1}$ with $i\in [1,m]$ in the following way:
\begin{equation}\label{eqa:one-columme-matrix}
{
\begin{split}
\uplus ^m_{i=1}X_i&=X_1\uplus X_2\uplus \cdots \uplus X_m\\
&=\left(\begin{array}{ccccc}
x_{1,1} & x_{2,1} & \cdots & x_{m,1}\\
x_{1,2} & x_{2,2} & \cdots & x_{m,2}\\
\cdots & \cdots & \cdots & \cdots\\
x_{1,n} & x_{2,n} & \cdots & x_{m,n}
\end{array}
\right)_{n\times m}
\end{split}}
\end{equation}

Moreover, let $T^i_{code}=(X_i~E_i~Y_i)^{T}_{3\times q_i}$ defined in Definition \ref{defn:Topcode-matrix}, where v-vector $X_i=(x_{i,1} ~ x_{i,2} ~ \cdots ~x_{i,q_i})$, e-vector $E_i=(e_{i,1}~e_{i,2} ~\cdots ~e_{i,q_i})$, and v-vector $Y_i=(y_{i,1} ~ y_{i,2} ~\cdots ~ y_{i,q_i})$ for $i\in[1,m]$. We have a Topcode-matrix union as follows
\begin{equation}\label{eqa:matrix-operation}
{
\begin{split}
T_{code}&=\uplus ^m_{i=1}T^i_{code}=T^1_{code}\uplus T^2_{code}\uplus \cdots \uplus T^m_{code}\\
&=\left(\begin{array}{ccccc}
X_1 & X_2 & \cdots & X_m\\
E_1 & E_2 & \cdots & E_m\\
Y_1 & Y_2 & \cdots & Y_m
\end{array}
\right)_{3\times mA}
\end{split}}
\end{equation}
by the union-addition operation, where $A=\sum ^m_{i=1}q_i$. Clearly,
\begin{thm}\label{thm:Topcode-matrix-union-graphicable}
If each Topcode-matrix $T^i_{code}$ with $i\in[1,m]$ is graphicable, so is the Topcode-matrix union $\uplus ^m_{i=1}T^i_{code}$.
\end{thm}

\vskip 0.2cm

Theorem \ref{thm:Topcode-matrix-union-graphicable} provides us techniques for constructing new Topsnut-gpws and new \emph{combinatorial labellings}.

(i) If a graph $G$ having the union $\uplus ^m_{i=1}T^i_{code}$ of Topcode-matrices $T^i_{code}$ with $i\in[1,m]$ is connected, then $G$ is a Topsnut-gpw and an authentication for the private keys $T^1_{code},T^2_{code},\dots, T^m_{code}$.

(ii) Suppose that each Topsnut-gpw $G_i$ having the Topcode-matrix $T^i_{code}$ admits a graph labelling/coloring $f_i$ with $i\in[1,m]$, then $G$ admits a graph labelling/coloring made by a combinatorial coloring/labelling $f=\uplus ^m_{i=1}f_i$.

\vskip 0.2cm

We present an example shown in Fig.\ref{fig:combinatorial-labellings} for illustrating the above techniques, where $G=\bigcup ^4_{i=1}G_i$, where $G_1$ admits a \emph{pan-graceful labelling} $f_1$ making a pan-graceful Topcode-matrix $T_{code}(G_1)$; $G_2$ admits a \emph{graceful labelling} $f_2$ making a graceful Topcode-matrix $T_{code}(G_2)$; $G_3$ admits an \emph{odd-graceful labelling} $f_3$ making an odd-graceful Topcode-matrix $T_{code}(G_3)$; $G_4$ admits an \emph{odd-edge-magic total labelling} $f_4$ making an odd-edge-magic total Topcode-matrix $T_{code}(G_4)$. Thereby, $G$ admits a combinatorial labelling $f=\uplus ^4_{i=1}f_i$, and corresponds a Topcode-matrix $\uplus ^4_{i=1}T_{code}(G_i)$.

Observe the Topsnut-gpw $G$ shown in Fig.\ref{fig:combinatorial-labellings} carefully, we can discover a difficult problem: Splitting $G$ into the original Topsnut-gpws $G_1,G_3,G_3,G_4$ is not easy, even impossible as if the size of $G$ is quite large, that is, $G$ has thousand and thousand vertices and edges. So, it is good for producing Topsnut-gpws having high-level security as desired, but it is very difficult for attacking our Topsnut-gpws, in other words, our Topsnut-gpws are certainly \emph{computational security}.

\begin{figure}[h]
\centering
\includegraphics[width=8.6cm]{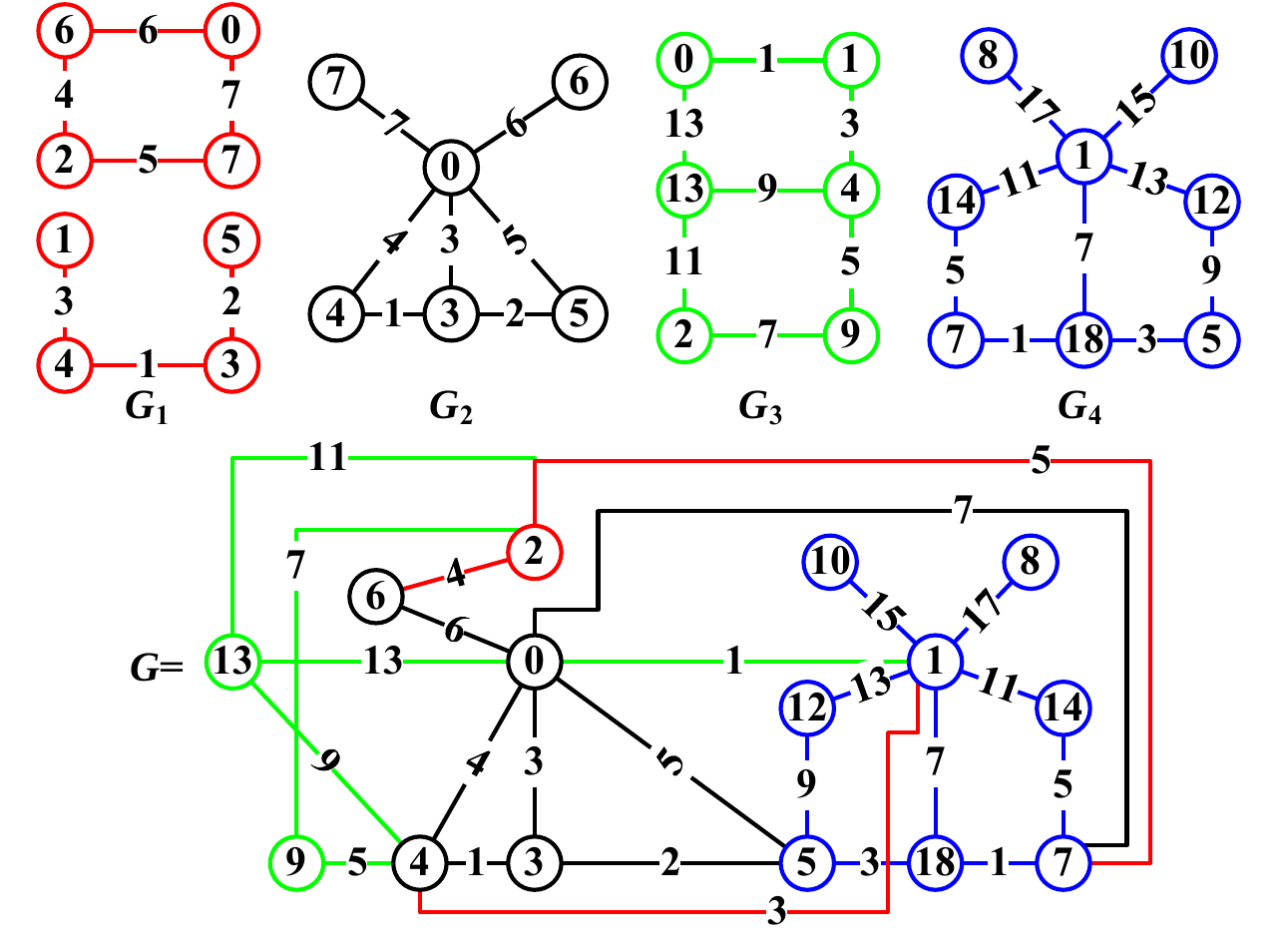}\\
\caption{\label{fig:combinatorial-labellings} {\small A Topsnut-gpw $G=\bigcup ^4_{i=1}G_i$ is connected, and each $G_i$ admits a graph labelling $f_i$ with $i\in [1,4]$.}}
\end{figure}

\subsubsection{Popular operations from sets}\quad For a Topcode-matrix $A_i=(X_i~E_i~Y_i)$ with $(x_i~e_i~y_i)^{T}$ for $i\in [1,m]$, and another Topcode-matrix $B_j=(X_j~E_j~Y_j)$ with $(x_j~e_j~y_j)^{T}$ for $j\in [1,n]$, as well as $W_{ij}=\{w_s=(x_s~e_s~y_s)^{T}:$ $w_s \textrm{ is in both $A_i$ and $B_j$}\}$, we can define the subtractive operation as $A_i\setminus B_j$ containing no any $w_s\in W_{ij}$ and no any $(x_j~e_j~y_j)^{T}$ in $B_j$; similarly, the subtractive operation as $B_j\setminus A_i$ containing no any $w_s\in W_{ij}$ and no any $(x_i~e_i~y_i)^{T}$ in $A_i$; the union operation as $A_i\cup B_j$ containing each element of $W_{ij}$, $A_i\setminus B_j$ and $B_j\setminus A_i$; and the intersection operation $A_i\cap B_j=W_{ij}$. For example,

{\footnotesize
\begin{equation}\label{eqa:subtractive-union-intersection-1}
\centering
A_i= \left(
\begin{array}{ccccccccc}
7 & 5 & 7 &1\\
1 & 3 & 5 &7 \\
18& 18 & 14 &18
\end{array}
\right)~
B_j= \left(
\begin{array}{ccccccccc}
7 &1 & 5 \\
1 &7 & 9\\
18 &18& 12
\end{array}
\right)
\end{equation}

\begin{equation}\label{eqa:subtractive-union-intersection-2}
\centering
A_i\cup B_j= \left(
\begin{array}{ccccccccc}
7 & 5 & 7 &1& 5\\
1 & 3 & 5 &7& 9 \\
18& 18 & 14 &18& 12
\end{array}
\right)~B_j\setminus A_i= \left(
\begin{array}{ccccccccc}
 5 \\
9 \\
12
\end{array}
\right)
\end{equation}

\begin{equation}\label{eqa:subtractive-union-intersection-3}
\centering
A_i\cap B_j= \left(
\begin{array}{ccccccccc}
7 & 1\\
1 &7 \\
18& 18
\end{array}
\right)~A_i\setminus B_j= \left(
\begin{array}{ccccccccc}
 5 & 7 \\
 3 & 5 \\
18 & 14
\end{array}
\right)
\end{equation}
}

\subsection{Splitting operations}

For investigation of splitting $G$ into the original Topsnut-gpws $G_1,G_3,G_3,G_4$ above Fig.\ref{fig:combinatorial-labellings}, we present a group of splitting operations, coincident operations and contracting operations as follows:

\begin{defn}\label{defn:split-operation-combinatoric}
\cite{Yao-Mu-Sun-Sun-Zhang-Wang-Su-Zhang-Yang-Zhao-Wang-Ma-Yao-Yang-Xie2019} Let $xw$ be an edge of a $(p,q)$-graph $G$, such that neighbor sets $N_{ei}(x)=\{w,u_1,u_2,\dots ,u_i, v_1,v_2,\dots ,v_j\}$ and $N_{ei}(w)=\{x,w_1,w_2,\dots ,w_m\}$, and $xw\in E(G)$.
\begin{asparaenum}[Op-1. ]
\item A \emph{half-edge split operation} is defined by deleting the edge $xw$, and then splitting the vertex $x$ into two vertices $x',x''$ and joining $x'$ with these vertices $w,u_1,u_2,\dots ,u_i$, and finally joining $x''$ with these vertices $w,v_1,v_2,\dots ,v_j$. The resultant graph is denoted as $G \wedge ^{1/2}xw$, named as a \emph{half-edge split graph}, and $N_{ei}(x')\cap N_{ei}(x'')=\{w\}$ in $H$. (see Fig.\ref{fig:split-operation-half-edge})

\begin{figure}[h]
\centering
\includegraphics[height=2.6cm]{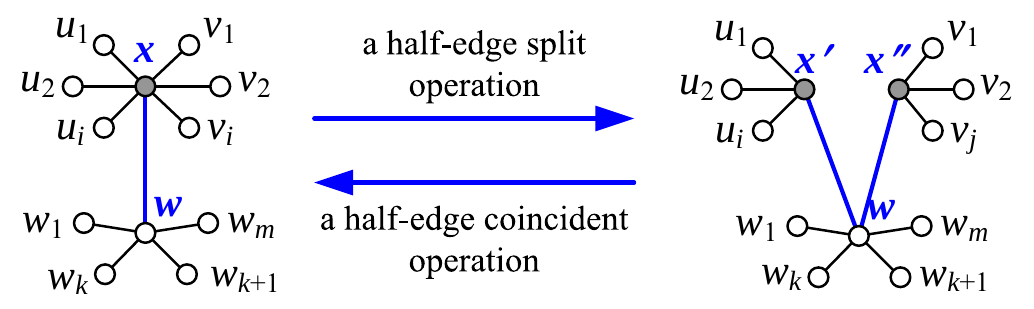}\\
\caption{\label{fig:split-operation-half-edge} {\small A scheme for a half-edge split
operation from left to right, and a half-edge coincident operation from right to left.}}
\end{figure}
\item A \emph{half-edge coincident operation} is defined as: Suppose that $N_{ei}(x')\cap N_{ei}(x'')=\{w\}$, we coincide $x'$ with $x''$ into one, denoted as $x=(x',x'')$, such that $N_{ei}(x)=N_{ei}(x')\cap N_{ei}(x'')$, that is, delete one multiple edge. The resultant graph is denoted as $G(x'w\odot x'w)$, called a \emph{half-edge coincident graph}. (see Fig.\ref{fig:split-operation-half-edge})
\item \cite{Yao-Sun-Zhang-Mu-Sun-Wang-Su-Zhang-Yang-Yang-2018arXiv} A \emph{vertex-split operation} is defined in the way: Split $x$ into two vertices $x',x''$ such that $N_{ei}(x')=\{w,u_1,u_2,\dots ,u_i\}$ and $N_{ei}(x'')=\{v_1,v_2,\dots ,v_j\}$ with $N_{ei}(x')$ $\cap $ $N_{ei}(x'')=\emptyset$; the resultant graph is written as $G\wedge x$, named as a \emph{vertex-split graph}. (see Fig.\ref{fig:split-operation-vertex-edge} from (b) to (a))
\item \cite{Yao-Sun-Zhang-Mu-Sun-Wang-Su-Zhang-Yang-Yang-2018arXiv} A \emph{vertex-coincident operation} is defined by coinciding two vertices $x'$ and $x''$ in to one $x=(x',x'')$ such that $N_{ei}(x)=N_{ei}(x')\cup N_{ei}(x'')$; the resultant graph is written as $G(x'\odot x'')$, called a \emph{vertex-coincident graph}. (see Fig.\ref{fig:split-operation-vertex-edge} from (b) to (a))

\begin{figure}[h]
\centering
\includegraphics[height=5.6cm]{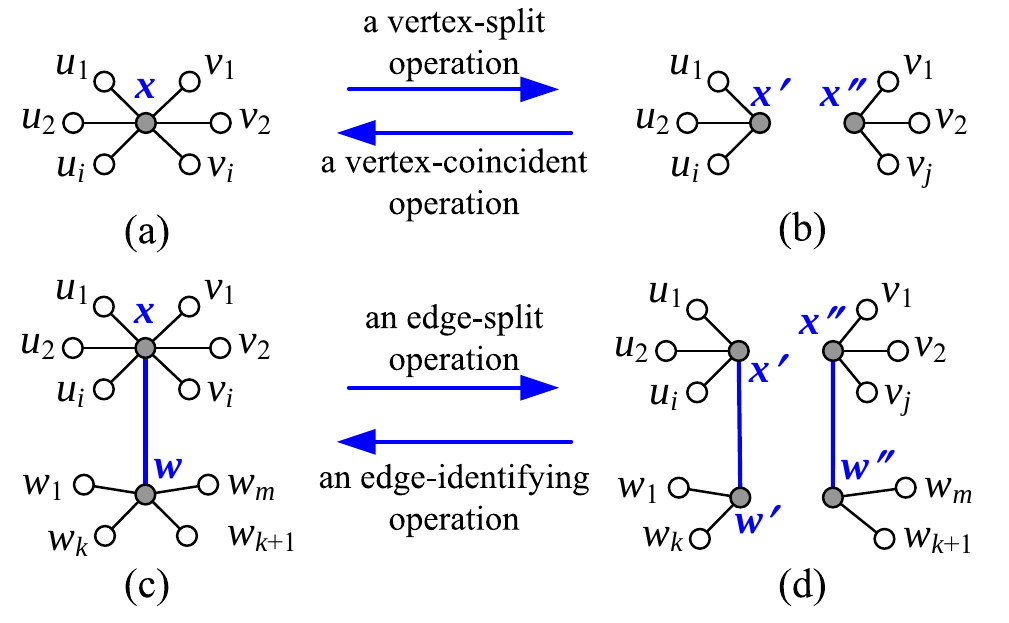}\\
\caption{\label{fig:split-operation-vertex-edge} {\small (a) A \emph{vertex-coincident graph} $G(x'\odot x'')$ obtained by a vertex-coincident operation from (b) to (a); (b) a \emph{vertex-split graph} $G\wedge x$ obtained by a vertex-split
operation from (a) to (b); (c) an \emph{edge-coincident graph} $G(x'w'\odot x''w'')$ obtained by an edge-coincident operation from (d) to (c); (d) an \emph{edge-split graph} $G\wedge xw$ obtained by an edge-split operation from (c) to (d).}}
\end{figure}

\item \cite{Yao-Sun-Zhang-Mu-Sun-Wang-Su-Zhang-Yang-Yang-2018arXiv} An \emph{edge-split operation} is defined as: Split the edge $xw$ into two edges $x'w'$ and $x''w''$ such that $N_{ei}(x')=\{w',u_1,u_2,\dots ,u_i\}$ and $N_{ei}(x'')=\{w''$, $v_1,v_2,\dots ,v_j\}$, $N_{ei}(w')=\{x',w_1,w_2,\dots ,w_k\}$ and $N_{ei}(w'')=$ $\{x'',w_{k+1},w_{k+2},\dots ,w_{m}\}$; the resultant graph is written as $G\wedge xw$, named as an \emph{edge-split graph}, with $N_{ei}(w')\cap N_{ei}(x')=\emptyset$, $N_{ei}(w')\cap N_{ei}(x'')=\emptyset$, and $N_{ei}(x')\cap N_{ei}(w')=\emptyset$, $N_{ei}(x')\cap N_{ei}(w'')=\emptyset$. (see Fig.\ref{fig:split-operation-vertex-edge} from (c) to (d))
\item \cite{Yao-Sun-Zhang-Mu-Sun-Wang-Su-Zhang-Yang-Yang-2018arXiv} An \emph{edge-coincident operation} is defined by coinciding two edges $x'w'$ and $x''w''$ into one edge $xw$ such that $N_{ei}(x)=N_{ei}(x')\cup N_{ei}(x'')\cup \{w=(w',w'')\}$ and $N_{ei}(w)=N_{ei}(w')\cup N_{ei}(w'')\cup \{x=(x',x'')\}$; the resultant graph is written as $G(x'w'\odot x''w'')$, called an \emph{edge-coincident graph}. (see Fig.\ref{fig:split-operation-vertex-edge} from (d) to (c))
\item \cite{Bondy-2008} In Fig.\ref{fig:split-operation-contract-subdivision}, an \emph{edge-contracting operation} is shown as: Delete the edge $xy$ first, and then coincide $x$ with $y$ into one vertex $w=(x,y)$ such that $N_{ei}(w)=[N_{ei}(x)\setminus \{y\}]\cup [N_{ei}(y)\setminus \{x\}]$. The resultant graph is denoted as $G\triangleleft xy$, named as an \emph{edge-contracted graph}.
 \item \cite{Bondy-2008} In Fig.\ref{fig:split-operation-contract-subdivision}, an \emph{edge-subdivided operation} is defined in the way: Split the vertex $w$ into two vertices $x,y$, and join $x$ with $y$ by a new edge $xy$, such that $N_{ei}(x)\cap N_{ei}(y)=\emptyset$, $y\in N_{ei}(x)$, $x\in N_{ei}(y)$ and $N_{ei}(w)=[N_{ei}(x)\setminus \{y\}]\cup [N_{ei}(y)\setminus \{x\}]$. The resultant graph is denoted as $G\triangleright w$, called an \emph{edge-subdivided graph}.\qqed

\begin{figure}[h]
\centering
\includegraphics[height=1.8cm]{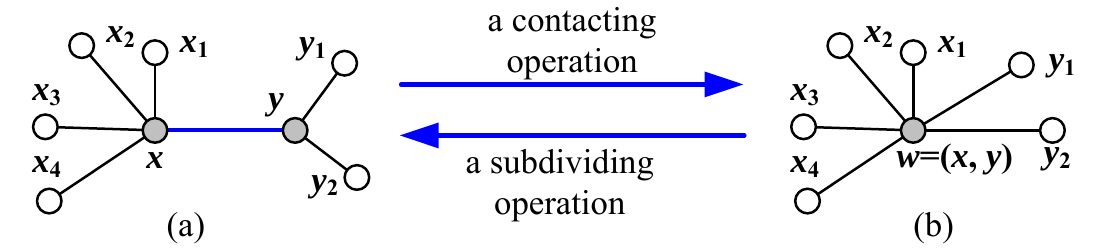}\\
\caption{\label{fig:split-operation-contract-subdivision} {\small (a) An \emph{edge-subdivided graph} $G\triangleright w$ obtained by subdividing a vertex $w=(x,y)$ into an edge $xy$ from (b) to (a); (b) an \emph{edge-contracted graph} $G\triangleleft xy$ obtained by contracting an edge $xy$ to a vertex $(x,y)$ from (a) to (b).}}
\end{figure}
\end{asparaenum}
\end{defn}

We, for understanding vertex-splitting Topcode-matrices, show examples shown in Fig.\ref{fig:v-split-matrix-1}. The Topcode-matrix $A$ shown in (\ref{eqa:example}) can be v-split into $A=A_{11}\uplus A_{12}$ (see Fig.\ref{fig:v-split-matrix-1} (a)), and $A=A_{21}\uplus A_{22}$ (see Fig.\ref{fig:v-split-matrix-1} (b)), where (a-1) and (b-1) are graphical illustration of doing vertex-splitting operations on the Topcode-matrix $A$.

\begin{figure}[h]
\centering
\includegraphics[width=8.6cm]{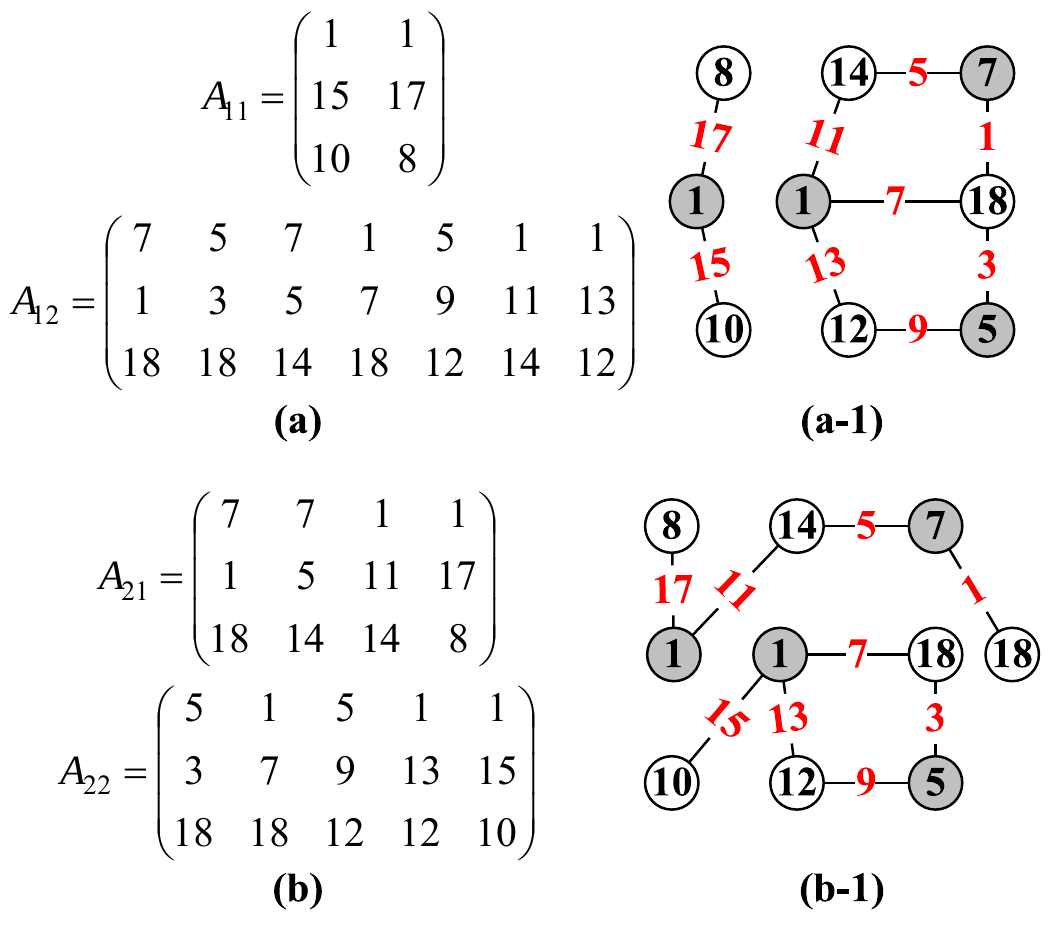}\\
\caption{\label{fig:v-split-matrix-1} {\small Two examples for understanding vertex-splitting Topcode-matrices.}}
\end{figure}

Again we consider do edge-splitting operations on Topcode-matrices. In Fig.\ref{fig:e-split-matrix-1}, we implement an half-edge-splitting operation to the $e=7$ of the Topcode-matrix $A$ shown in (\ref{eqa:example}), the resultant Topcode-matrix is denoted as $A\wedge^{1/2} \{7\}=A_{31}$ (see Fig.\ref{fig:e-split-matrix-1} (c)), and moreover we do an edge-splitting operation to the $e=7$ of the Topcode-matrix $A$ shown in (\ref{eqa:example}), the resultant T[opcode-matrix is denoted as $A\wedge \{7\}=A_{41}\uplus A_{42}$ (see Fig.\ref{fig:e-split-matrix-1} (d)). We use two Topsnut-gpws (c-1) and (d-1) to explain the edge-splitting operation to the Topcode-matrix $A$.

We are ready to present the following definitions of connectivity on Topcode-matrices:

\begin{defn}\label{defn:v-splitting-Topcode-matrices}
$^*$ Let $T_{code}$ be a Topcode-matrix defined in Definition \ref{defn:Topcode-matrix}. If there exists a number $w_i$ that is the common end of $e_{i_1},e_{i_2},\dots ,e_{i_{m}}$ of $E$ with $i_m\geq 2$ such that there are two sub-Topcode-matrices $T'_{code}$ and $T''_{code}$ of $T_{code}$ hold $e_{i_1},e_{i_2},\dots ,e_{i_k}$ are in $T'_{code}$ with $i_k\geq 1$, and $e_{i_k+1},e_{i_k+2},\dots ,e_{i_{m}}$ are in $T''_{code}$ with $i_m-(i_k+1)\geq 1$, and $T_{code}=T'_{code}\uplus T''_{code}$. We call the process of obtaining $T_{code}=T'_{code}\uplus T''_{code}$ as a \emph{v-splitting operation} on Topcode-matrices, and particularly write
\begin{equation}\label{eqa:v-splitting-Topcode-matrices-0}
T_{code}\wedge \{w_i\}=T'_{code}\uplus T''_{code}.
\end{equation}
\end{defn}

Let $\{w_1,w_2,\dots ,w_n\}=\{w_i\}^n_1\subset (XY)^*$ of $T_{code}$ defined in Definition \ref{defn:Topcode-matrix}, and $w_i$ be the common end of $e_{i_1},e_{i_2},\dots ,e_{i_{m}}$ in $E$ with $i_m\geq 2$ for $i\in [1,n]$. Based on Definition \ref{defn:v-splitting-Topcode-matrices}, we do a series of v-splitting operations to each of $\{w_i\}^n_1$, the resultant Topcode-matrix is denoted as:
\begin{equation}\label{eqa:v-splitting-Topcode-matrices-1}
T_{code}\wedge \{w_i\}^n_1=\uplus^m_{j=1} T^j_{code}.
\end{equation}
If $T_{code}$ and each $T^j_{code}$ with $j\in [1,m]$ are connected, we call $T_{code}\wedge \{w_i\}^n_1$ to be \emph{$n$-connected}, and moreover, the smallest number $k$ of $n$ for which $T_{code}\wedge \{w_i\}^n_1$ is $n$-connected is denoted as $\kappa(T_{code})=k$, and we say that $T_{code}$ is the \emph{v-$k$-code connectivity}.

\begin{defn}\label{defn:half-e-splitting-Topcode-matrices}
$^*$ Let $T_{code}$ be a Topcode-matrix defined in Definition \ref{defn:Topcode-matrix}, and let $e_i\in E^*$. Doing a v-splitting operation to one end $x_i$ of $e_i$ and adding $(x_i~e'_i~y_i)^{T}$ produces
\begin{equation}\label{eqa:v-splitting-Topcode-matrices-3}
T_{code}\wedge ^{1/2}\{e_i\}=T_{code}\uplus (x_i~e'_i~y_i)^{T}.
\end{equation}
We call this process a \emph{half-e-splitting operation}, it splits $e_i$ into $e_i$ and $e'_i$ such that both $e_i$ and $e'_i$ have a common end $y_i$.\qqed
\end{defn}

See an example of the half-e-splitting operation shown in Fig.\ref{fig:e-split-matrix-1} (c) and (c-1).

\begin{defn}\label{defn:e-splitting-Topcode-matrices}
$^*$ Let $T_{code}$ be a Topcode-matrix defined in Definition \ref{defn:Topcode-matrix}, and let $e_i\in E^*$. Doing two v-splitting operations to the ends $x_i,y_i$ of $e_i$ produces two sub-Topcode-matrices $T^1_{code}$ and $T^2_{code}$ of $T_{code}$ such that $T_{code}=T^1_{code}\uplus T^2_{code}$, and $e_i$ is in $T^1_{code}$ and not in $T^2_{code}$, but $x_i,y_i$ are in both $T^1_{code}$ and $T^2_{code}$. We add $(x_i~e'_i~y_i)^{T}$ to $T^2_{code}$ to form a new Topcode-matrix $T^*_{code}=T^2_{code}\uplus (x_i~e'_i~y_i)^{T}$. This process is called an \emph{e-splitting operation}, it splits $e_i$ into $e_i$ and $e'_i$ such that
\begin{equation}\label{eqa:v-splitting-Topcode-matrices-2}
T_{code}\wedge \{e_i\}=T^1_{code}\uplus \big [T^2_{code}\uplus (x_i~e'_i~y_i)^{T}\big ].
\end{equation}
\end{defn}

An example for understanding e-splitting operation is shown in Fig.\ref{fig:e-split-matrix-1} (d) and (d-1). Let $\{e_{i_1},e_{i_2},\dots ,e_{i_n}\}=\{e_{i_j}\}^n_1\subset E^*$ of $T_{code}$, we do a series of e-splitting operations to each of $\{e_{i_j}\}^n_1$, the resultant Topcode-matrix is denoted as:
\begin{equation}\label{eqa:v-splitting-Topcode-matrices-1}
T_{code}\wedge \{e_{i_j}\}^n_1=\uplus^m_{s=1} T^s_{code}.
\end{equation}
If $T_{code}$ and each $T^s_{code}$ with $s\in [1,r]$ are connected, we call $T_{code}\wedge\{e_{i_j}\}^n_1$ to be \emph{$m$-e-connected}, and moreover, the smallest number $k$ of $m$ for which $T_{code}\wedge \{e_{i_j}\}^n_1$ is $k$-e-connected is denoted as $\kappa'(T_{code})=k$, and we say that $T_{code}$ is the \emph{e-$k$-code connectivity}.

\begin{figure}[h]
\centering
\includegraphics[width=8.6cm]{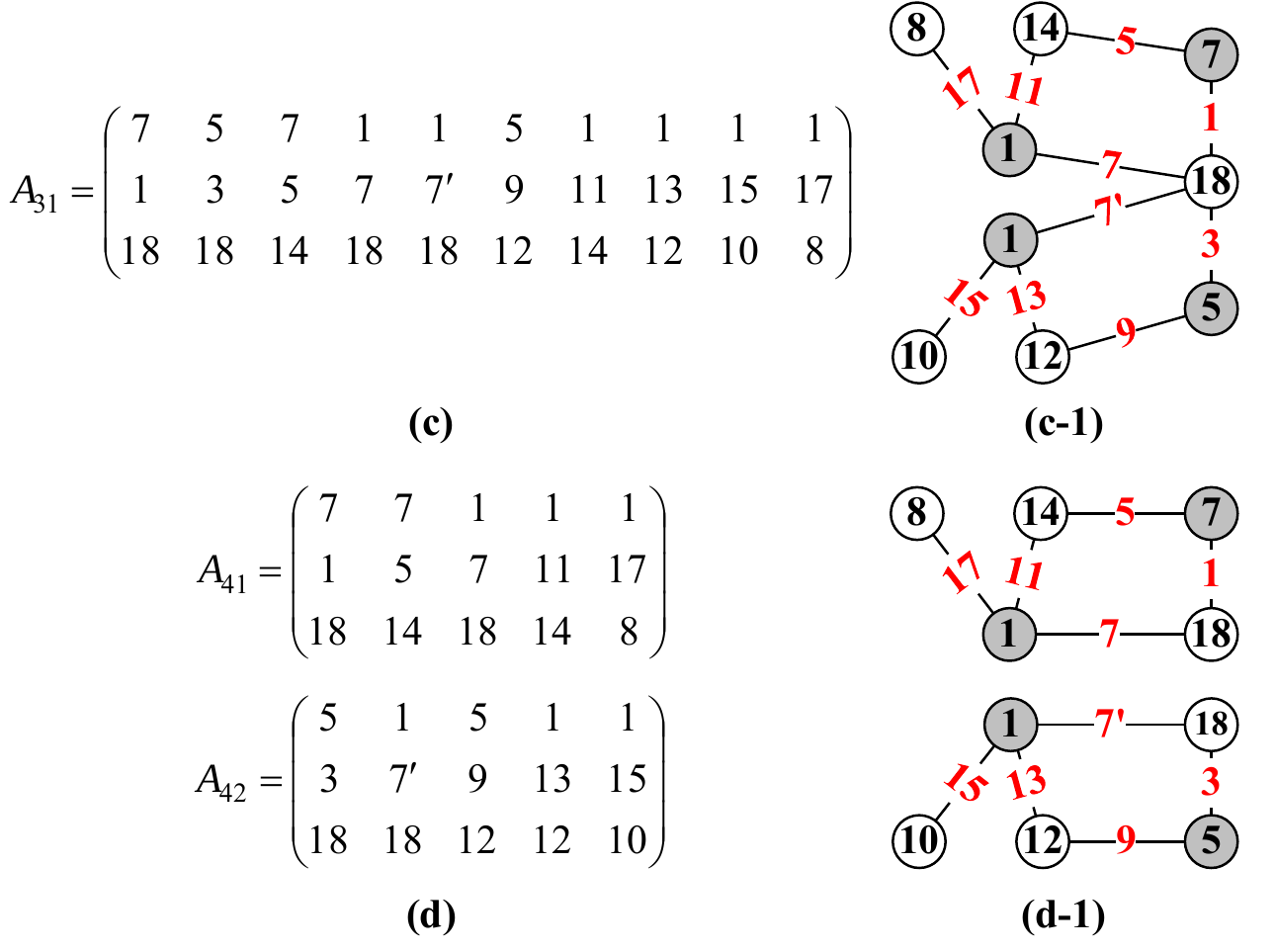}\\
\caption{\label{fig:e-split-matrix-1} {\small Two Topcode-matrices for understanding edge-splitting Topcode-matrices.}}
\end{figure}

\subsection{Topcode$^+$-matrix groups based on the additive v-operation}
We define another operation between the Topcode-matrices $T^1_{code},T^2_{code},\dots, T^m_{code}$ in this subsection. For a fixed positive integer $k$, if there exists a constant $M$, such that
\begin{equation}\label{eqa:group-operation-1}
(x_{i,r}+x_{j,r}-x_{k,r})~(\bmod~ M)=x_{\lambda,r}\in X_{\lambda}
\end{equation}
\begin{equation}\label{eqa:group-operation-2}
(y_{i,r}+y_{j,r}-y_{k,r})~(\bmod~ M)=y_{\lambda,r}\in Y_{\lambda}
\end{equation}
where $\lambda=i+j-k~(\bmod~ M)\in [1,m]$ and, $T^{\lambda}_{code}=(X_{\lambda}$ $E_{\lambda}$ $Y_{\lambda})^{T}$. Let $F_m=\{T^1_{code},T^2_{code},\dots, T^m_{code}\}$. Then we say (\ref{eqa:group-operation-1})+(\ref{eqa:group-operation-2}) to be an \emph{additive v-operation} on $F_m$, and we write this operation by ``$\oplus$'', and we have a matrix equation
\begin{equation}\label{eqa:group-operation-3}
T^i_{code}\oplus_k T^j_{code}=T^\lambda_{code}
\end{equation} based on the zero $T^k_{code}$ and $\lambda=i+j-k~(\bmod~M)$. By the additive v-operation defined in (\ref{eqa:group-operation-1}) and (\ref{eqa:group-operation-2}), if we have
\begin{asparaenum}[(i) ]
\item \emph{Every-zero}. Each element $T^k_{code}$ can be regarded as the \emph{zero} such that each $T^i_{code}\in F_m$
$$T^i_{code}\oplus_k T^k_{code}=T^i_{code};$$
\item \emph{Closure and uniqueness}. If $T^i_{code}\oplus_k T^j_{code}=T^\lambda_{code}$, $T^i_{code}\oplus_k T^j_{code}=T^\mu_{code}$, then $\lambda=\mu$. And
 $$T^i_{code}\oplus_k T^j_{code}\in F_m;$$
\item \emph{Inverse}. Each $T^i_{code}$ has its own \emph{inverse} $T^{i^{-1}}_{code}$ such that $$T^i_{code}\oplus_k T^{i^{-1}}_{code}=T^k_{code}.$$
\item \emph{Associative law}. $$T^i_{code}\oplus_k [T^j_{code}\oplus_k T^s_{code}]=[T^i_{code}\oplus_k T^j_{code}]\oplus_k T^s_{code}.$$
\end{asparaenum}
Then we call $F_m$ to be an \emph{every-zero additive associative Topcode-matrix group} (Topcode$^+$-matrix group for short), denoted as $\{F_m;\oplus _k \}$. In general, we write ``$\oplus$'' to replace ``$\oplus_k$'' if there is no confusion,

\vskip 0.2cm

An every-zero additive associative Topcode$^+$-matrix group $\{F_m;\oplus \}$ is shown in Fig.\ref{fig:topcode-matrix-group}, where the additive v-operation ``$\oplus $'' is defined in the equations (\ref{eqa:group-operation-1}) and (\ref{eqa:group-operation-2}). Moreover, if each Topcode-matrix in an every-zero additive associative Topcode$^+$-matrix group $\{F_m;\oplus \}$ is graphicable with a graph $G$, we call $\{F_m;\oplus \}$ an \emph{every-zero graph group} and rewrite $\{F_m;\oplus \}=\{F_m(G);\oplus \}$, see an example shown in Fig.\ref{fig:graph-group}.

\begin{figure}[h]
\centering
\includegraphics[width=7.2cm]{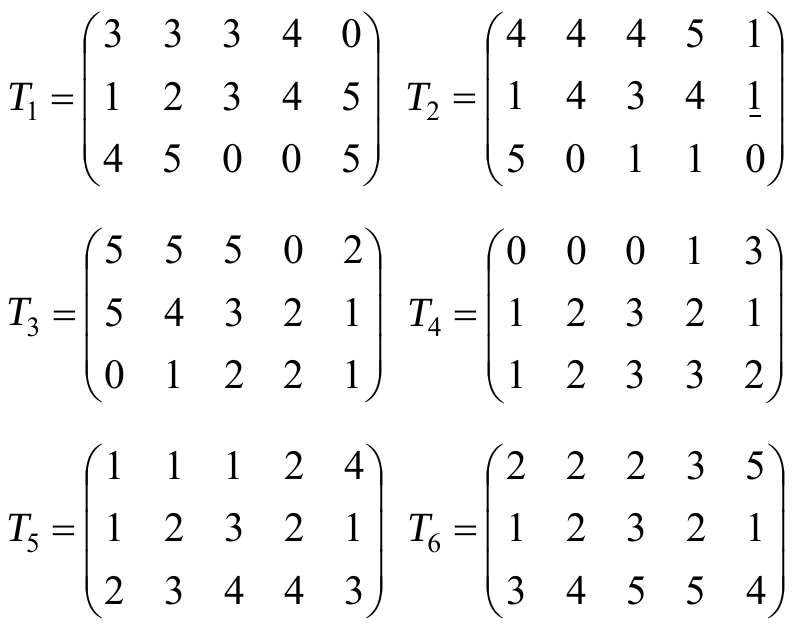}\\
\caption{\label{fig:topcode-matrix-group} {\small An every-zero additive associative Topcode$^+$-matrix group $\{F_6;\oplus \}$ with $F_6=\{T_1,T_2,T_3,T_4,T_5,T_6\}$ under modular $6$.}}
\end{figure}

\begin{figure}[h]
\centering
\includegraphics[width=8.6cm]{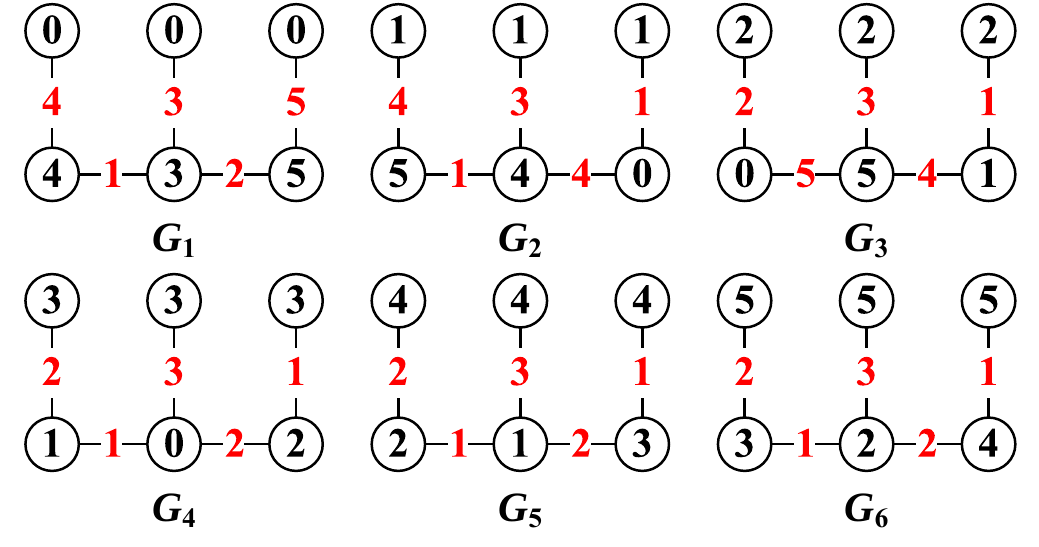}\\
\caption{\label{fig:graph-group} {\small An every-zero graph group $\{F_6(G);\oplus \}$ with $F_6(G)=\{G_i:i\in [1,6]\}$ under modular $6$, where each Topsnut-gpw $G_i$ corresponds a Topcode-matrix $T_i$ shown in Fig.\ref{fig:topcode-matrix-group}, $i\in [1,6]$.}}
\end{figure}

\subsection{Number string groups}

According to the Topcode-matrices $T_1,T_2,T_3,T_4,T_5,T_6$ shown in Fig.\ref{fig:topcode-matrix-group} and the rule Vo-1 shown in Fig.\ref{fig:4-rules}, we can write six TB-paws as follows
{\footnotesize
\begin{equation}\label{eqa:string-group-1}
\begin{array}{ccccc}
T_b(T_1)=333405432145005&T_b(T_2)=444511434150110\\
T_b(T_3)=555021234501221&T_b(T_4)=000131232112332\\
T_b(T_5)=111241232123443&T_b(T_6)=222351232134554
\end{array}
\end{equation}
}by the formula below
\begin{equation}\label{eqa:string-group-2}
{
\begin{split}
T_b(T_i)=&x_{i,1}x_{i,2}x_{i,3}x_{i,4}x_{i,5}e_{i,5}e_{i,4}e_{i,3}e_{i,2}e_{i,1}\\
&y_{i,1}y_{i,2}y_{i,3}y_{i,4}y_{i,5}
\end{split}}
\end{equation}
where $e_i=|x_i-y_i|$ with $i\in [1,6]$. The set $G_6=\{T_b(T_1)$, $T_b(T_2)$, $T_b(T_3)$, $T_b(T_4)$, $T_b(T_5)$, $T_b(T_6)\}$ forms an \emph{every-zero number string group} $\{G_6;\oplus \}$ based on the additive v-operation $\oplus$ defined in the equations (\ref{eqa:group-operation-1}) and (\ref{eqa:group-operation-2}). For example, we select randomly $T_b(T_3)$ as the zero, and write $[w_{a,j}+w_{b,j}-w_{3,j}]~(\bmod~6)=w_{\lambda,j}$ with $\lambda=a+b-3~(\bmod~6)$ in the following verification:

(1) $T_b(T_1)\oplus T_b(T_2)=T_b(T_6)$: $x_{1,j}+x_{2,j}-x_{3,j}=x_{6,j}~(\bmod~6)$ and $y_{1,j}+y_{2,j}-y_{3,j}=y_{6,j}~(\bmod~6)$ with $j\in [1,5]$.

(2) $T_b(T_1)\oplus T_b(T_3)=T_b(T_1)$.

(3) $T_b(T_1)\oplus T_b(T_4)=T_b(T_2)$: $x_{1,j}+x_{4,j}-x_{3,j}=x_{2,j}~(\bmod~6)$ and $y_{1,j}+y_{4,j}-y_{3,j}=y_{2,j}~(\bmod~6)$ with $j\in [1,5]$.

(4) $T_b(T_1)\oplus T_b(T_5)=T_b(T_3)$: $x_{1,j}+x_{5,j}-x_{3,j}=x_{3,j}~(\bmod~6)$ and $y_{1,j}+y_{5,j}-y_{3,j}=y_{3,j}~(\bmod~6)$ with $j\in [1,5]$.

(5) $T_b(T_1)\oplus T_b(T_6)=T_b(T_4)$: $x_{1,j}+x_{6,j}-x_{3,j}=x_{4,j}~(\bmod~6)$ and $y_{1,j}+y_{6,j}-y_{3,j}=y_{4,j}~(\bmod~6)$ with $j\in [1,5]$.

(6) $T_b(T_2)\oplus T_b(T_3)=T_b(T_2)$.

(7) $T_b(T_2)\oplus T_b(T_4)=T_b(T_3)$: $x_{2,j}+x_{4,j}-x_{3,j}=x_{3,j}~(\bmod~6)$ and $y_{2,j}+y_{4,j}-y_{3,j}=y_{3,j}~(\bmod~6)$ with $j\in [1,5]$.

(8) $T_b(T_2)\oplus T_b(T_5)=T_b(T_4)$: $x_{2,j}+x_{5,j}-x_{3,j}=x_{4,j}~(\bmod~6)$ and $y_{2,j}+y_{5,j}-y_{3,j}=y_{4,j}~(\bmod~6)$ with $j\in [1,5]$.

(9) $T_b(T_2)\oplus T_b(T_6)=T_b(T_5)$: $x_{2,j}+x_{5,j}-x_{3,j}=x_{5,j}~(\bmod~6)$ and $y_{2,j}+y_{5,j}-y_{3,j}=y_{5,j}~(\bmod~6)$ with $j\in [1,5]$.

(10) $T_b(T_3)\oplus T_b(T_r)=T_b(T_r)$ with $r\in [1,6]$.

(11) $T_b(T_4)\oplus T_b(T_5)=T_b(T_6)$: $x_{4,j}+x_{5,j}-x_{3,j}=x_{6,j}~(\bmod~6)$ and $y_{4,j}+y_{5,j}-y_{3,j}=y_{6,j}~(\bmod~6)$ with $j\in [1,5]$.

(12) $T_b(T_4)\oplus T_b(T_6)=T_b(T_1)$: $x_{4,j}+x_{6,j}-x_{3,j}=x_{1,j}~(\bmod~6)$ and $y_{4,j}+y_{6,j}-y_{3,j}=y_{1,j}~(\bmod~6)$ with $j\in [1,5]$.

(13) $T_b(T_5)\oplus T_b(T_6)=T_b(T_2)$: $x_{5,j}+x_{6,j}-x_{3,j}=x_{2,j}~(\bmod~6)$ and $y_{5,j}+y_{6,j}-y_{3,j}=y_{2,j}~(\bmod~6)$ with $j\in [1,5]$.

Thereby, we claim that the set $G_6$ is an every-zero number string group. It is not hard to obtain the algorithms for writing string groups from the rules shown in Fig.\ref{fig:4-rules} on matrix groups. But, there is no general way for random actions of writing string groups on matrix groups.

\vskip 0.4cm

\subsubsection{A technique for making number string groups}\quad We starting this technique from an initial number string $S_1=x^1_1x^1_2\cdots x^1_n$, where each $x^1_i$ is a non-negative integer with $i\in [1,n]$. We select randomly $x^1_{k_1},x^1_{k_2}, \dots, x^1_{k_m}$ from $S_1$, call them active numbers, and then make number string $S_i=x^i_1x^i_2\cdots x^i_n$, such that
$$
x^i_j=\left\{
\begin{array}{ll}
k+x^1_j~(\bmod~M),& \textrm{if}~ j\in \{k_1,k_2,\dots ,k_m\};\\
x^1_j,& \textrm{otherwise}.
\end{array}
\right.
$$ with $2\leq i$. So, we get a set $I_M=\{S_1, S_2,\dots ,S_M\}$. It is not hard to prove that $I_M$ forms an every-zero number string group by selecting randomly a zero $S_k$ and do the additive v-operation defined in the equations (\ref{eqa:group-operation-1}) and (\ref{eqa:group-operation-2}) to $I_M$, that is, $S_i \oplus_k S_j=S_\lambda$ defined by the following operation
\begin{equation}\label{eqa:number-string-group}
[x^i_{\alpha}+x^j_{\alpha}-x^k_{\alpha}]~(\bmod~M)=x^{\lambda}_{\alpha},
\end{equation}
where $\lambda=i+j-k~(\bmod~M)$ if $\alpha\in \{k_1,k_2,\dots ,k_m\}$, otherwise $x^{\lambda}_{\alpha}=x^1_{\alpha}$ for $\alpha\in [1,n]\setminus \{k_1,k_2,\dots ,k_m\}$. Hence, we get an \emph{every-zero number string group} $\{I_M;\oplus\}$.

The above technique of constructing every-zero number string groups contains probability. On the other hands, number string groups are easily used to encrypting different communities of a dynamic network at distinct time step, since we can make the groups mentioned here to some graphs, so it is called \emph{graph networking groups}.

\vskip 0.4cm

\subsubsection{Exchanging operations on number strings}\quad Let $T_b(T_{code})$ be a number string generated from a Topcode-matrix $T_{code}$ defined in Definition \ref{defn:Topcode-matrix}. So, we have another Topcode-matrix $T^*_{code}=C_{(c,l)(a,b)}(T_{code})$ by doing a series of column-exchanging operations $c_{(i_k,j_k)}$ with $k\in [1,a]$ and a series of XY-exchanging operations $l_{(i_s)}$ with $s\in [1,b]$ to the Topcode-matrix $T_{code}$. Thereby, the number string $T_b(T_{code})$ is changed by the operation $C_{(c,l)(a,b)}$, the resultant number string is denoted as $C_{(c,l)(a,b)}(T_b(T_{code}))$. Thereby, we can do the exchanging operations to an every-zero number string group $I^*_M=\{S^*_1, S^*_2,\dots, S^*_M\}$ with $S^*_i=C_{(c,l)(a,b)}(S_i)$ for $i\in [1,M]$.

\vskip 0.4cm

\subsection{Topcode$^-$-matrix groups defined by the subtractive v-operation}

Let $F_m=\{T^1_{code},T^2_{code},\dots, T^m_{code}\}$ be a set of Topcode-matrices with $T^i_{code}=(X_i~E_i~Y_i)^{T}$, $X_i=(x_{i,1}$ $x_{i,2}$ $\cdots$ $x_{i,q})$, $E_i=(e_{i,1}~e_{i,2}~\cdots~e_{i,q})$ and $Y_i=(y_{i,1}$ $y_{i,2}$ $\cdots$ $y_{i,q})$, $i\in [1,m]$. For a fixed Topcode-matrix $T^k_{code}\in F_m$, if there exists a constant $M$, such that
\begin{equation}\label{eqa:subtraction-group-operation-1}
(x_{i,r}-x_{j,r}+x_{k,r})~(\bmod~ M)=x_{\lambda,r}\in X_{\lambda}
\end{equation}
\begin{equation}\label{eqa:subtraction-group-operation-2}
(y_{i,r}-y_{j,r}+y_{k,r})~(\bmod~ M)=y_{\lambda,r}\in Y_{\lambda}
\end{equation}
where $\lambda=i-j+k~(\bmod~ M)\in [1,m]$ and, $T^{\lambda}_{code}=(X_{\lambda}~E_{\lambda}~Y_{\lambda})^{T}\in F_m$. Then two equations (\ref{eqa:subtraction-group-operation-1}) and (\ref{eqa:subtraction-group-operation-2}) defines a new operation, called the \emph{subtractive v-operation}, denoted as
$$T^i_{code}\ominus_k T^j_{code}=T^{\lambda}_{code}.$$ We can show:

(i) Each Topcode-matrix $T^k_{code}\in F_m$ is a zero of the subtractive v-operation, that is, $T^i_{code}\ominus_k T^k_{code}=T^{i}_{code}$.

(ii) $T^i_{code}\ominus_k T^j_{code}\in F_m$, and if $T^i_{code}\ominus_k T^j_{code}=T^{\lambda}_{code}$ and $T^i_{code}\ominus_k T^j_{code}=T^{\mu}_{code}$, then $\lambda =\mu$.

(iii) There exists $T^{i^{-1}}_{code}\in F_m$, such that $T^i_{code}\ominus_k T^{i^{-1}}_{code}=T^k_{code}$.

(vi) $[T^i_{code}\ominus_k T^j_{code}]\ominus_k T^{k}_{code}=T^i_{code}\ominus_k [T^j_{code}\ominus_k \ T^{k}_{code}]$.\\

Thereby, we get an \emph{every-zero subtractive Topcode-matrix group} (Topcode$^-$-matrix group) $\{F_m;\ominus\}$ based on the subtractive v-operation defined in the equations (\ref{eqa:subtraction-group-operation-1}) and (\ref{eqa:subtraction-group-operation-2}).

\vskip 0.4cm

Similarly, there are \emph{number string groups} or \emph{graphic groups} made by the subtractive v-operation. Other groups can be obtained by one of two equations (\ref{eqa:subtraction-group-operation-1}) and (\ref{eqa:subtraction-group-operation-2}), or one of two equations (\ref{eqa:group-operation-1}) and (\ref{eqa:group-operation-2}), such that a general matrix group $U_m=\{M(a^r_{i,j})_{m\times n}: ~r\in [1,m]\}$ holds true, where $M(a^r_{i,j})_{m\times n}$ is defined in (\ref{eqa:general-matrix}).

\section{Techniques for producing text-base codes}

In this section we introduce basic techniques for producing TB-paws from Topcode-matrices and Topsnut-matrices of Topsnut-gpws. As mentioned in the previous sections, the TB-paws made from those Topsnut-matrices can not rebuild up the original Topsnut-gpws, since Topsnut-gpws consist of topological structures and Topcode-matrices, but TB-paws are not related with topological structures.

\subsection{Basic rules for finding continuous fold lines}

Some basic rules shown in Fig.\ref{fig:4-rules} can be formed as algorithms. In general, we have a matrix $M(a_{i,j})_{m\times n}$ defined as:
\begin{equation}\label{eqa:general-matrix}
{
\begin{split}
M(a_{i,j})_{m\times n}=\left(\begin{array}{ccccc}
a_{1,1} & a_{1,2} & \cdots & a_{1,n}\\
a_{2,1} & a_{2,2} & \cdots & a_{2,n}\\
\cdots & \cdots & \cdots & \cdots\\
a_{m,1} & a_{m,2} & \cdots & a_{m,n}
\end{array}
\right)
\end{split}}
\end{equation}

\begin{figure}[h]
\centering
\includegraphics[width=8.2cm]{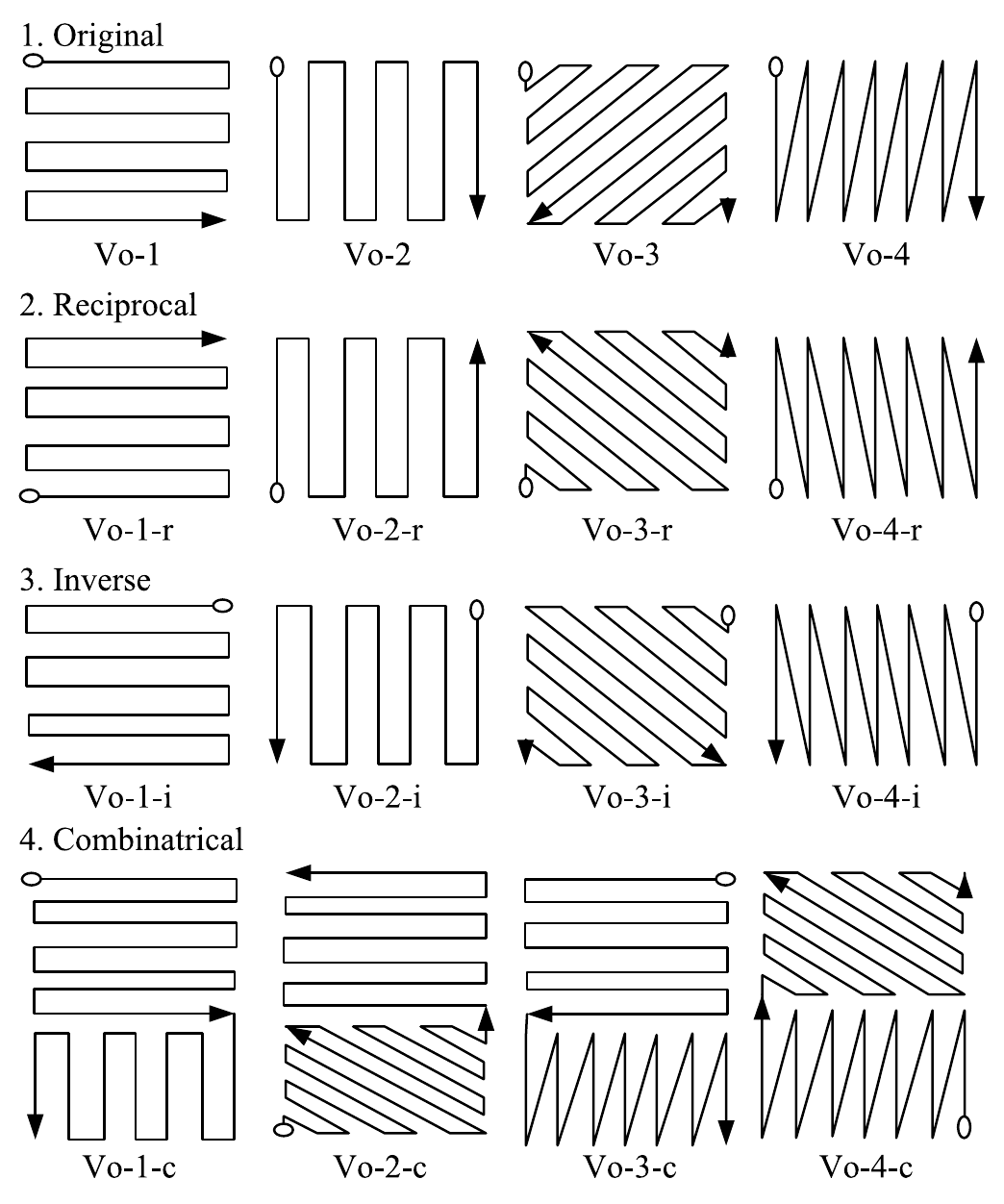}\\
\caption{\label{fig:4-rules} {\small Basic rules (also, adjacent TB-paw lines) for producing TB-paws from Topsnut-matrices, Topcode-matrices and Hanzi-GB2312-80 matrices.}}
\end{figure}
Based on the matrix (\ref{eqa:general-matrix}), the rule Vo-1 pictured in Fig.\ref{fig:4-rules} enables us to write out a TB-paw
\begin{equation}\label{eqa:Vo-1-1}
{
\begin{split}
&\quad T^e_b(M(a_{i,j})_{m\times n})=a_{1,1}a_{1,2} \cdots a_{1,n}\\
&a_{2,n}a_{2,n-1}\cdots a_{2,2}a_{2,1}a_{3,1}a_{3,2} \cdots a_{3,n}\\
&a_{4,n}a_{4,n-1}\cdots a_{4,2}a_{4,1}a_{5,1}a_{5,2}\cdots a_{5,n-1}a_{5,n}\\
& \cdots \cdots \cdots\\
& a_{m-1,1}a_{m-1,2} \cdots a_{m-1,n}a_{m,n}a_{m,n-1}\cdots a_{m,2}a_{m,1}
\end{split}}
\end{equation}
as $m$ is even, and furthermore we have another TB-paw
\begin{equation}\label{eqa:Vo-1-2}
{
\begin{split}
&\quad T^o_b(M(a_{i,j})_{m\times n})=a_{1,1}a_{1,2} \cdots a_{1,n}\\
&a_{2,n}a_{2,n-1}\cdots a_{2,2}a_{2,1}a_{3,1}a_{3,2} \cdots a_{3,n}\\
&a_{4,n}a_{4,n-1}\cdots a_{4,2}a_{4,1}a_{5,1}a_{5,2}\cdots a_{5,n-1}a_{5,n}\\
& \cdots \cdots \cdots\\
& a_{m,1}a_{m,2} \cdots a_{m,n}
\end{split}}
\end{equation}
as $m$ is odd. We can write a TB-paw line
\begin{equation}\label{eqa:TB-paw line}
L=(x_{1},y_{1})(x_{2},y_{2})\cdots (x_{mn},y_{mn}),
\end{equation} where $(x_{1},y_{1})$ is the \emph{initial point}, $(x_{mn},y_{mn})$ the \emph{terminal point} in $xOy$-plane, and $(x_{i},y_{i})\neq (x_{j},y_{j})$ if $i\neq j$. We have some particular TB-paw lines as:

(1) An \emph{adjacent TB-paw line} $L$ shown in (\ref{eqa:TB-paw line}) has: Two consecutive points $(x_{i},y_{i})(x_{i+1},y_{i+1})$ hold one of $|y_{i}-y_{i+1}|=1$ and $|x_{i}-x_{i+1}|=1$ (see Fig.\ref{fig:adjacent-TB-paw-line}(a)).

(2) A \emph{closed adjacent TB-paw line} is an adjacent TB-paw line with both the initial point and the terminal point are a point (see Fig.\ref{fig:adjacent-TB-paw-line}(b)).

(3) A \emph{closed TB-paw line} such that each point of the line is both the initial point and the terminal point (see Fig.\ref{fig:adjacent-TB-paw-line}(c)).

\begin{figure}[h]
\centering
\includegraphics[width=8.6cm]{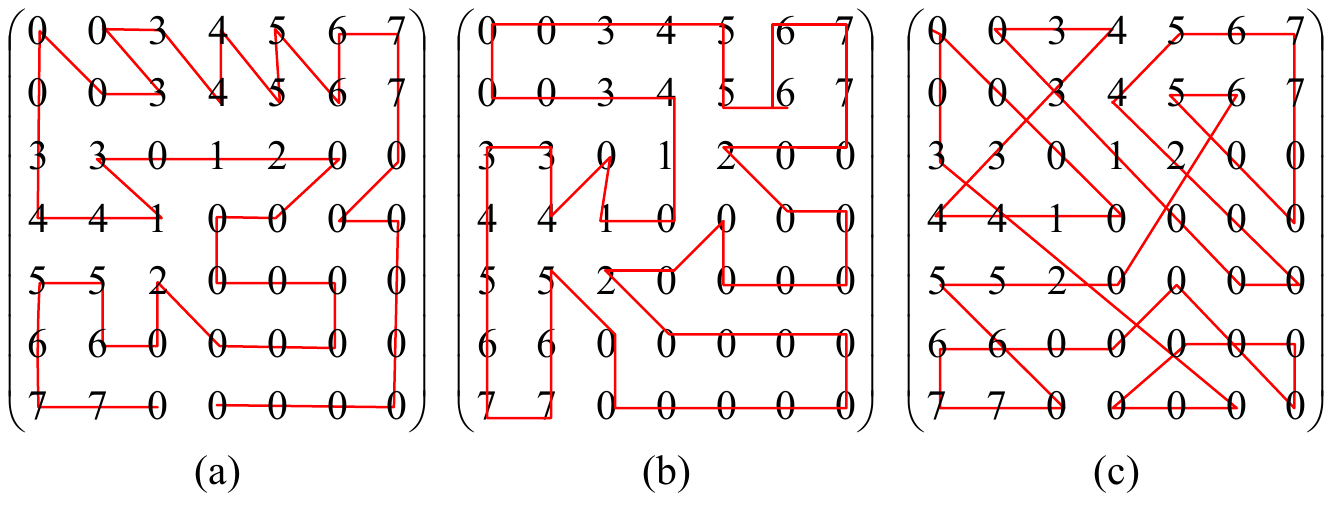}\\
\caption{\label{fig:adjacent-TB-paw-line} {\small (a) An adjacent TB-paw line; (b) a closed adjacent TB-paw line; (c) a closed TB-paw line, not adjacent.}}
\end{figure}

\begin{rem}\label{eqa:total-TB-paw-line}
(1) The TB-paw lines shown in Fig.\ref{fig:4-rules} can be written into algorithms of polynomial times.

(2) TB-paws made from the matrix (\ref{eqa:general-matrix}) corresponds a question, called ``\emph{Finding all total TB-paw lines in a lattice $P_m\times P_n$ for deriving TB-paws from matrices}'' stated exactly as: Let $P_m\times P_n$ be a lattice in $xOy$-plane. There are points $(i,j)$ on the lattice $P_m\times P_n$ with $i\in[1,m]$ and $j\in [1,n]$. If a fold-line $L$ with initial point $(a,b)$ and terminal point $(c,d)$ on $P_m\times P_n$ is internally disjoint and contains all points $(i,j)$ of $P_m\times P_n$, and each point $(i,j)$ appears in $L$ once only, we call $L$ a \emph{total TB-paw line}. We hope to find all possible total TB-paw lines of $P_m\times P_n$. Also, we can consider that $L$ consists of $L_1,L_2,\dots ,L_m$ with $m\geq 2$.

(3) There are random TB-paw lines based on probabilistic methods for deriving probabilistic TB-paws.
\end{rem}

\subsection{Topsnut- and Topcode-matrices related with graphs}

\vskip 0.4cm

\subsubsection{Topcode-matrices from public keys and private keys}\quad Since a Topsnut-gpw $G$ can be split into two parts: one is a \emph{public key} $G_{pub}$ and another one is a \emph{private key} $G_{pri}$, such that $G=G_1\cup G_2$ is an \emph{authentication}, see an example shown in Fig.\ref{fig:2-public-private-keys} for understanding this concept. Thereby, we have a Topsnut-matrix
$$A_{vev}(G)=A_{vev}(G_{pub})\uplus A_{vev}(G_{pri})$$
made by two Topsnut-matrices $A_{vev}(G_{pub})$ and $A_{vev}(G_{pri})$. On the other hands, these Topsnut-matrices are just three e-valued Topcode-matrices $T_{code}(G)$, $T_{code}(G_{pub})$ and $T_{code}(G_{pri})$. In general, we, according to the matrix equation (\ref{eqa:matrix-operation}), can write a Topcode-matrix $T_{code}$ defined in Definition \ref{defn:Topcode-matrix} in the following form
\begin{equation}\label{eqa:matrix-operation-xx}
{
\begin{split}
T_{code}=\uplus ^m_{i=1}T^i_{code}=T^1_{code}\uplus T^2_{code}\uplus \cdots \uplus T^m_{code}
\end{split}}
\end{equation}
where each $T^i_{code}$ can be regarded as a \emph{public key} or a \emph{private key} to an authentication $T_{code}$.

\begin{figure}[h]
\centering
\includegraphics[width=8cm]{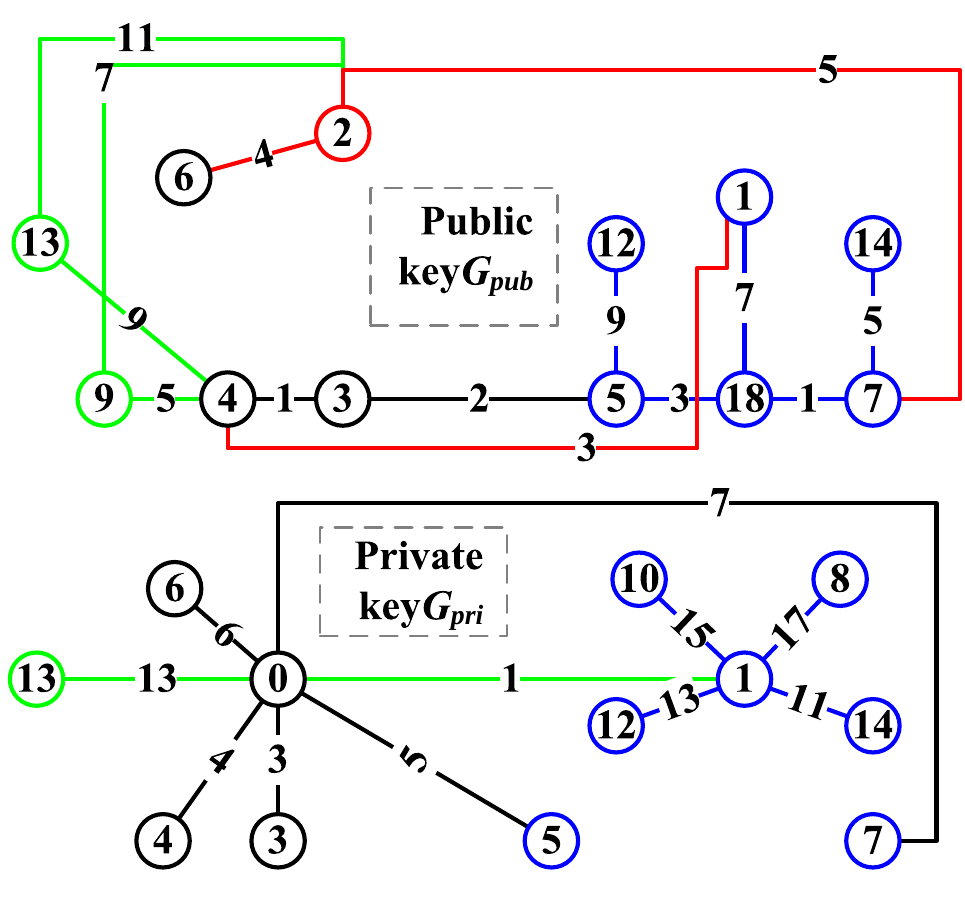}\\
\caption{\label{fig:2-public-private-keys} {\small A public key $G_{pub}$ and a private key $G_{pri}$ can be certified successfully by an authentication $G$ shown in Fig.\ref{fig:combinatorial-labellings}, that is, $G=G_{pub}\cup G_{pri}$, where the set $\{0,1,3,4,5,6,7,12,13,14\}$ is the common boundaries of $G_{pub}$ and $G_{pri}$.}}
\end{figure}

\vskip 0.4cm

\subsubsection{Complementary Topcode-matrices}\quad Let $T_{code}=(X~E~Y)^{T}$ and $\overline{T}_{code}=(\overline{X}~\overline{E}~\overline{Y})^{T}$ be two Topcode-matrices.

A Topcode-matrix $T_{code}$ defined in Definition \ref{defn:Topcode-matrix} is \emph{bipartite} if each graph corresponding $T_{code}$ has no odd-cycle.

We say $\overline{T}_{code}$ to be a \emph{complementary Topcode-matrix} of a Topcode-matrix $T_{code}$ if $T_{code}\uplus \overline{T}_{code}$ holds $E^*\cap \overline{E}^*=\emptyset$ and corresponds a complete graph $K_n$ with $(XY)^*=(\overline{X}\overline{Y})^*$ and $n=|(XY)^*|$. See some examples shown in Fig.\ref{fig:dual-twin} and \ref{fig:dual-twin-matrix}.
\begin{figure}[h]
\centering
\includegraphics[width=8.6cm]{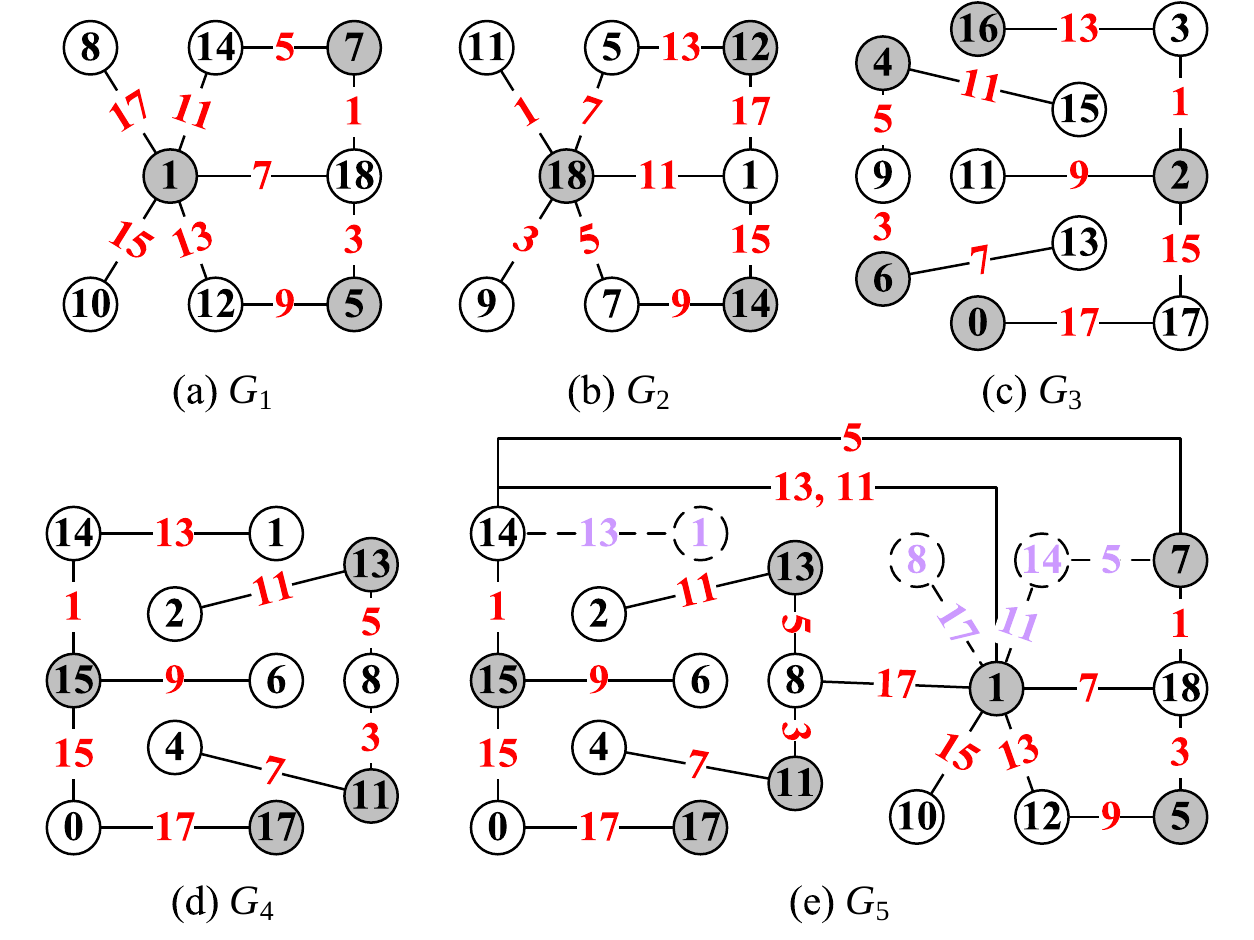}\\
\caption{\label{fig:dual-twin} {\small (a) and (b) are a pair of complementary Topsnut-gpws; (c) and (d) are a pair of complementary Topcode-matrices; (a) and (c) are a pair of twin Topsnut-gpws; and $G_5=G_1\cup G_2$.}}
\end{figure}

\begin{figure}[h]
\centering
\includegraphics[width=8.6cm]{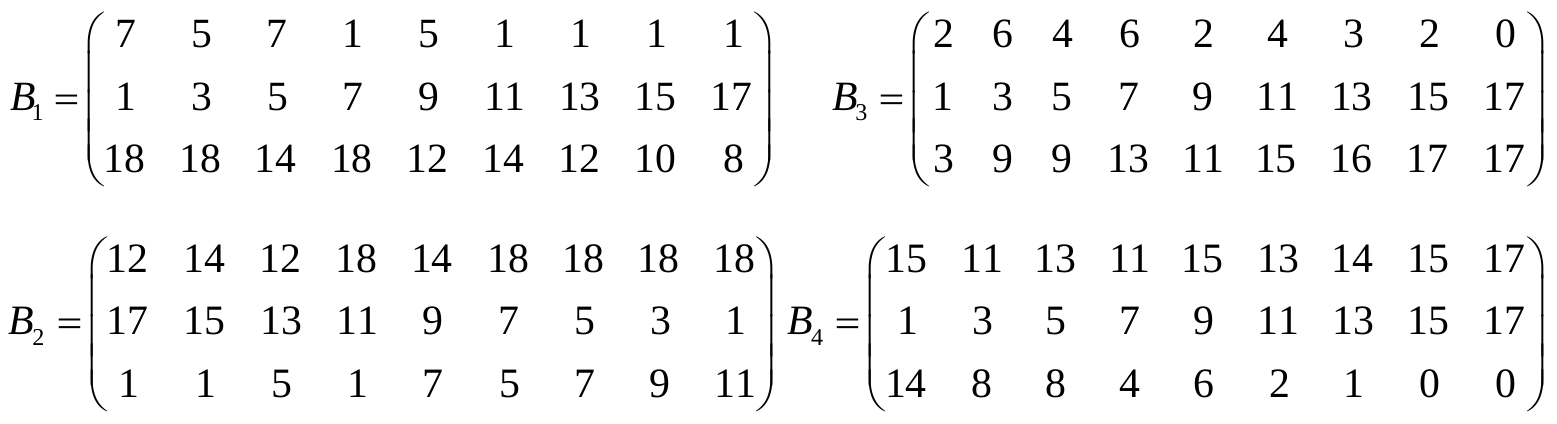}\\
\caption{\label{fig:dual-twin-matrix} {\small $B_1$ and $B_2$ are a pair of complementary Topcode-matrices; $B_1$ and $B_3$ are a pair of twin Topcode-matrices; $B_3$ and $B_4$ are a pair of complementary Topcode-matrices. Here, each Topcode-matrix $B_i$ corresponds a graph $G_i$ shown in Fig.\ref{fig:dual-twin} with $i\in [1,4]$.}}
\end{figure}

\vskip 0.4cm

\subsubsection{Twin odd-graceful Topcode-matrices}\quad As known, an odd-graceful Topcode-matrix $T_{code}=(X~E~Y)^{T}_{3\times q}$ holds $(XY)^*\subset [0,2q-1]$ and $e_i=|x_i-y_i|$ with $i\in [1,q]$, as well as $E^*=[1,2q-1]^o$. And there is another Topcode-matrix $\overline{T}_{code}=(\overline{X}~\overline{E}~\overline{Y})^{T}_{3\times q}$ holds $(\overline{X}\overline{Y})^*\subset [0,2q]$ and $e'_i=|x'_i-y'_i|$ for $i\in [1,q]$ such that $\overline{E}^*=[1,2q-1]^o$ holds true. If $(XY)^*\cup (\overline{X}\overline{Y})^*= [0,2q]$, where $(XY)^*$ and $(\overline{X}\overline{Y})^*$ are two sets of all different numbers in $X,Y,\overline{X}$ and $\overline{Y}$ respectively, then we say $(T_{code},\overline{T}_{code})$ a \emph{twin odd-graceful Topcode-matrix matching} (see examples shown in Fig.\ref{fig:twin-odd-matrices}).

\begin{figure}[h]
\centering
\includegraphics[width=8.6cm]{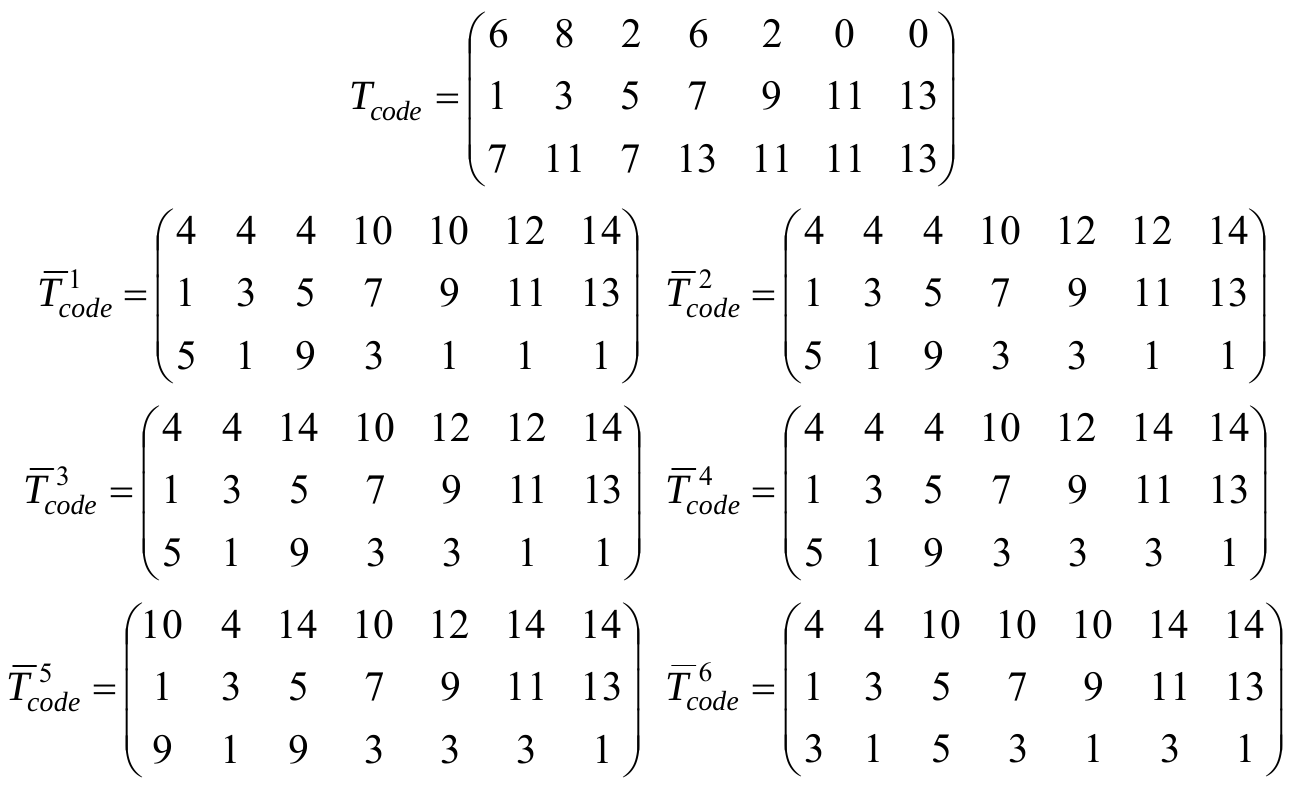}\\
\caption{\label{fig:twin-odd-matrices} {\small Twin odd-graceful Topcode-matrices $(T_{code},\overline{T}^k_{code})$ with $k\in [1,6]$.}}
\end{figure}

\vskip 0.4cm

\subsubsection{Topcode-matrices of line Topsnut-gpws}\quad We show the line graphs in Fig.\ref{fig:graph-group-matching}. For $i\in [1,6]$, each line Topsnut-gpw $G_{im}$ is the \emph{line graph} of the Topsnut-gpw $G_i$ shown in Fig.\ref{fig:graph-group}. It is not hard to write the Topcode-matrices of these line Topsnut-gpws. Furthermore, $(G_{i},G_{im})$ is a \emph{Topsnut-gpw matching}, or an \emph{encrypting authentication}.
\begin{figure}[h]
\centering
\includegraphics[width=8.6cm]{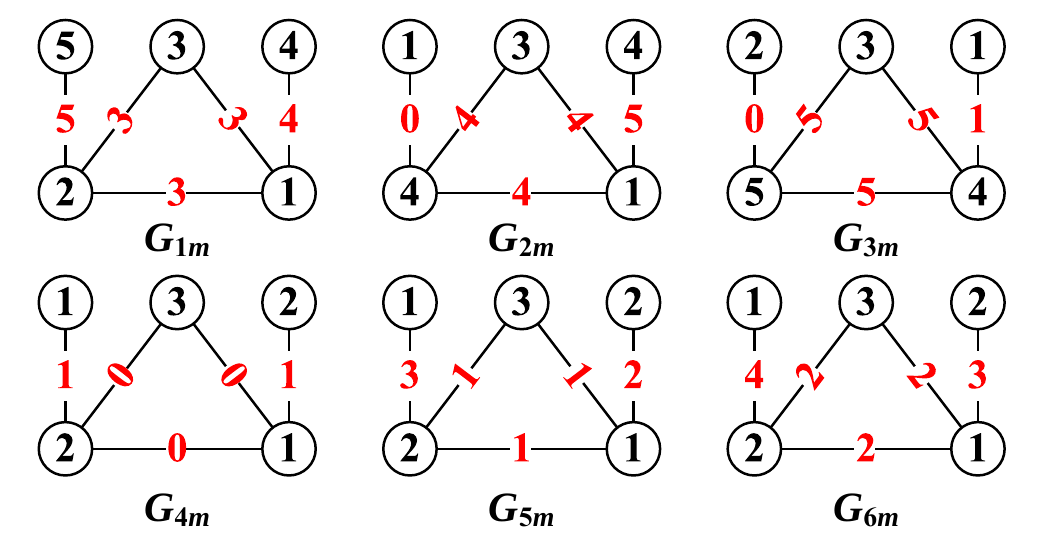}\\
\caption{\label{fig:graph-group-matching} {\small Each line Topsnut-gpw $G_{im}$ is the line graph of the Topsnut-gpw $G_i$ shown in Fig.\ref{fig:graph-group} with $i\in [1,6]$.}}
\end{figure}

A graceful Topcode-matrix $T_{code}$ corresponds a graceful graph $G$, and another graceful Topcode-matrix $T^L_{code}$ corresponds another graceful graph $H$. We call $(T_{code},T^L_{code})$ to be double graceful if $H$ is the \emph{line graph} of $G$.

\vskip 0.4cm

\subsubsection{Various $H$-complementary matchings based on Topcode-matrices}\quad Let $G$ be a connected graph having two proper subgraphs $H_1$ and $H_2$ with $|V(H_i)|\geq 2$ and $|E(H_i)|\geq 1$ for $i=1,2$. We have:
\begin{asparaenum}[\textrm{Condition}-1. ]
\item \label{eqa:same-v-set} $V(H)=V(H_1)=V(H_2)$;
\item \label{eqa:not-join-v-set} $V(H)=V(H_1)\cup V(H_2)$ and $V(H_1)\cap V(H_2)=\emptyset $;
\item \label{eqa:join-v-set} $V(H)=V(H_1)\cup V(H_2)$ and $V(H_1)\cap V(H_2)\neq \emptyset $;
\item \label{eqa:not-join-e-set} $E(H)=E(H_1)\cup E(H_2)$ and $E(H_1)\cap E(H_2)=\emptyset $;
\item \label{eqa:join-e-set} $E(H)=E(H_1)\cup E(H_2)$ and $E(H_1)\cap E(H_2)\neq \emptyset $.
\end{asparaenum}
Suppose that $G$ corresponds a Topcode-matrix $T_{code}$ defined in Definition \ref{defn:Topcode-matrix}, and $H_i$ corresponds a Topcode-matrix $T^i_{code}$ defined in Definition \ref{defn:Topcode-matrix} with $i=1,2$. There are the following $H$-complementary matchings:
\begin{asparaenum}[\textrm{ComMat}-1. ]
\item An e-proper $H$-complementary matching $(T^1_{code},T^2_{code})$ holds Condition-\ref{eqa:same-v-set}, Condition-\ref{eqa:not-join-e-set} and $T_{code}=T^1_{code}\uplus T^2_{code}$ true.
\item A v-joint e-proper $H$-complementary matching $(T^1_{code},T^2_{code})$ holds Condition-\ref{eqa:join-v-set}, Condition-\ref{eqa:not-join-e-set} and $T_{code}=T^1_{code}\uplus T^2_{code}$ true.
\item An e-joint $H$-complementary matching $(T^1_{code},T^2_{code})$ holds Condition-\ref{eqa:same-v-set}, Condition-\ref{eqa:join-e-set} and $T_{code}=T^1_{code}\uplus T^2_{code}$ true.
\item A ve-joint $H$-complementary matching $(T^1_{code},T^2_{code})$ holds Condition-\ref{eqa:join-v-set}, Condition-\ref{eqa:join-e-set} and $T_{code}=T^1_{code}\uplus T^2_{code}$ true.
\end{asparaenum}

\section{Pan-matrices}

A pan-matrix has its own elements from a set $Z$, for instance, $Z$ is a letter set, or a Hanzi (Chinese characters) set, or a graph set, or a poem set, or a music set, even a picture set, \emph{etc.}. In this section, we will discuss pan-matrices made by Hanzi, or Chinese idioms.

\subsection{Hanzi-matrices}

``Chinese character'' is abbreviated as ``Hanzi'' hereafter, and we use the Chinese code GB2312-80 in \cite{GB2312-80} to express a Hanzi, and write this Hanzi as $H_{abcd}$, where $abcd$ is a Chinese code in \cite{GB2312-80}.

\begin{defn}\label{defn:Topsnut-matrix}
\cite{Yao-Mu-Sun-Sun-Zhang-Wang-Su-Zhang-Yang-Zhao-Wang-Ma-Yao-Yang-Xie2019} A \emph{Hanzi-GB2312-80 matrix} (or Hanzi-matrix) $A_{han}(H)$ of a Hanzi-sentence $H=\{ H_{a_ib_ic_id_i}\}^m_{i=1}$ made by $m$ Hanzis $H_{a_1b_1c_1d_1}$, $H_{a_2b_2c_2d_2}$, $\dots$, $H_{a_mb_mc_md_m}$ is defined as
\begin{equation}\label{eqa:a-formula}
\centering
{
\begin{split}
A_{han}(H)&= \left(
\begin{array}{ccccc}
a_{1} & a_{2} & \cdots & a_{m}\\
b_{1} & b_{2} & \cdots & b_{m}\\
c_{1} & c_{2} & \cdots & c_{m}\\
d_{1} & d_{2} & \cdots & d_{m}
\end{array}
\right)=\left(\begin{array}{c}
A\\
B\\
C\\
D
\end{array} \right)\\
&=(A~B~C~D)^{T}_{4\times m}
\end{split}}
\end{equation}\\
where
\begin{equation}\label{eqa:three-vectors}
{
\begin{split}
&A=(a_1 ~ a_2 ~ \cdots ~a_m), B=(b_1 ~ b_2 ~ \cdots ~b_m)\\
&C=(c_1 ~ c_2 ~ \cdots ~c_m), D=(d_1 ~ d_2 ~ \cdots ~d_m)
\end{split}}
\end{equation}
where each Chinese code $a_ib_ic_id_i$ is defined in \cite{GB2312-80}.\qqed
\end{defn}

See two Hanzi-GB2312-80 matrices shown in Fig.\ref{fig:6-hanzi-matrices}. Another Hanzi-GB2312-80 matrix is as follows:

\begin{equation}\label{eqa:Hanzi-sentence-GB2312-80}
\centering
{
\begin{split}
A_{han}(G^*)&= \left(
\begin{array}{ccccccccc}
4 & 4 & 2 & 2 & 5 & 4 & 4 & 4 & 3\\
0 & 0 & 6 & 5 & 2 & 4 & 7 & 4 & 8\\
4 & 4 & 3 & 1 & 8 & 7 & 3 & 1 & 2\\
3 & 3 & 5 & 1 & 2 & 6 & 4 & 1 & 9
\end{array}
\right)
\end{split}}
\end{equation}
according to a Hanzi-sentence $G^*=H_{4043}$ $H_{4043}$ $H_{2635}$ $H_{2511}$ $H_{5282}$ $H_{4476}$ $H_{4734}$ $H_{4411}$ $H_{3829}$ shown in Fig.\ref{fig:ren-ren-hao-gong} (b)(1). This Hanzi-sentence $G^*$ induces another matrix as follows:

{\small
\begin{equation}\label{eqa:Hanzi-sentence-code}
\centering
{
\begin{split}
A_{han}(G^*)&= \left(
\begin{array}{ccccccccc}
4 & 4 & 5 & 5 & 5 & 5 & 4 & 5 & 5\\
E & E & 9 & 1 & 2 & 9 & E & 9 & E\\
B & B & 7 & 6 & 1 & 2 & 0 & 1 & 7\\
A & A & D & C & 9 & 9 & B & A & 3
\end{array}
\right)
\end{split}}
\end{equation}
}
by Chinese code of Chinese dictionary, see Fig.\ref{fig:ren-ren-hao-gong} (a).

\begin{figure}[h]
\centering
\includegraphics[width=8.6cm]{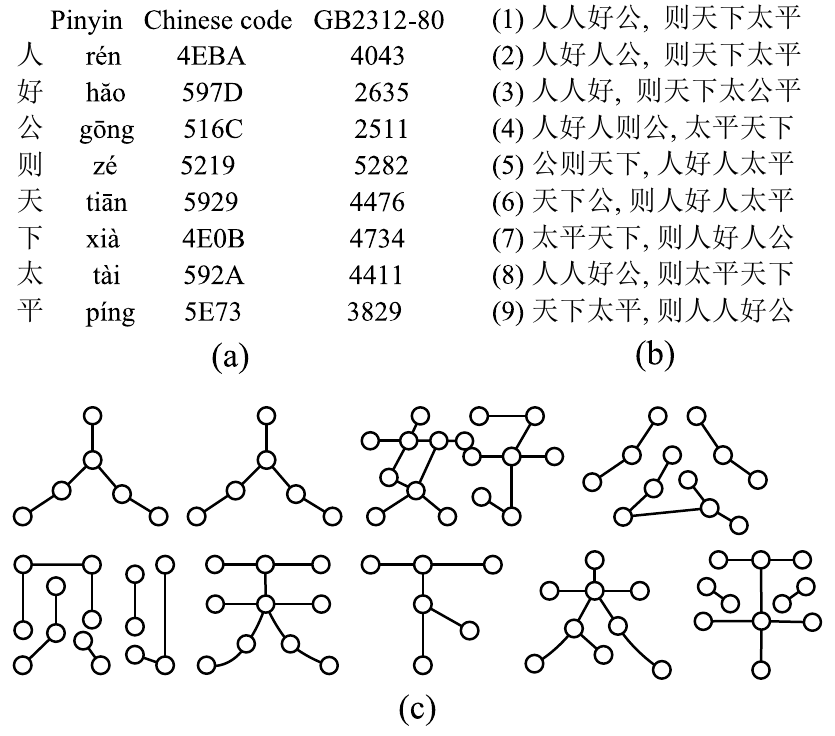}\\
\caption{\label{fig:ren-ren-hao-gong} {\small In \cite{Yao-Mu-Sun-Sun-Zhang-Wang-Su-Zhang-Yang-Zhao-Wang-Ma-Yao-Yang-Xie2019}: (a) Pinyin, Chinese code and Hanzi-GB2312-80 of Hanzis; (b) nine Hanzi-sentences made by nine Hanzis $H_{4043}$, $H_{4043}$, $H_{2635}$, $H_{2511}$, $H_{5282}$, $H_{4476}$, $H_{4734}$, $H_{4411}$ and $H_{3829}$ shown in (a); (c) a topological model of nine Hanzis, called \emph{Hanzi-graphs}.}}
\end{figure}

\subsection{Adjacent ve-value matrices}

An adjacent ve-value matrix $A(G)=(a_{i,j})_{(p+1)\times (p+1)}$ is defined on a colored $(p,q)$-graph $G$ with a coloring/labelling $g: V(G)\cup E(G)\rightarrow [a,b]$ for $V(G)=\{x_1,x_2,\dots ,x_p\}$ in the way:

(i) $a_{1,1}=0$, $a_{1,j+1}=g(x_j)$ with $j\in[1,p]$, and $a_{k+1,1}=g(x_k)$ with $k\in[1,p]$.

(ii) $a_{i+1,i+1}=0$ with $i\in[1,p]$.

(iii) For an edge $x_{ij}=x_ix_j\in E(G)$ with $i,j\in[1,p]$, then $a_{i+1,j+1}=g(x_{ij})$, otherwise $a_{i+1,j+1}=0$.

The Tosnut-gpw $G_2$ shown in Fig.\ref{fig:combinatorial-labellings} has its own adjacent ve-value matrix $A(G_2)$ shown in Fig.\ref{fig:5-other-rules}.

\begin{figure}[h]
\centering
\includegraphics[width=8.6cm]{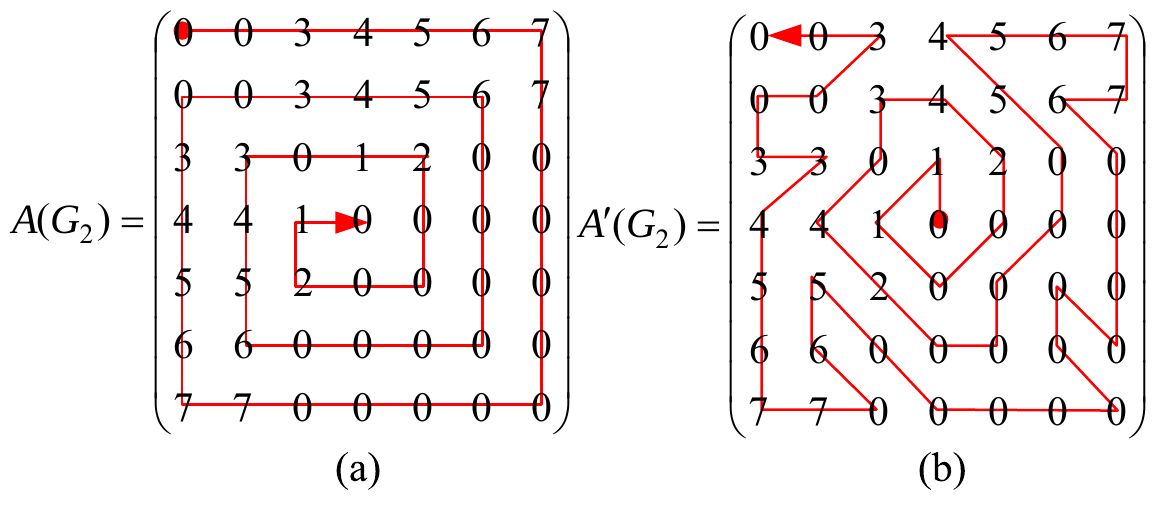}\\
\caption{\label{fig:5-other-rules} {\small Other rules on two adjacent ve-value matrices defined by the Tosnut-gpw $G_2$ shown in Fig.\ref{fig:combinatorial-labellings}.}}
\end{figure}

According to a red line $L$ under the adjacent ve-value matrix $A(G_2)$ shown in Fig.\ref{fig:5-other-rules}(a), we can write a TB-paw
$${
\begin{split}
T_b(A(G_2))=&0034567700000000077654300\\
&345600000006543012000210
\end{split}}
$$ and another TB-paw
$${
\begin{split}
T_b(A'(G_2))=&011002430420000054567760\\
&0000000000560776543300300
\end{split}}
$$ is obtained along a red line $L'$ under the adjacent ve-value matrix $A'(G_2)$ shown in Fig.\ref{fig:5-other-rules} (b).

We have a matrix set $\{A(G_2;k):~k\in [1,7]\}$ with $A(G_2;1)$ $=A(G_2)$ shown in Fig.\ref{fig:5-other-rules}(a), and each matrix $A(G_2;k)=$ $(a^k_{i,j})_{7\times 7}$ holds $a^k_{i,j}=k+a^1_{i,j}~(\bmod~7)$ with $k\in [1,7]$, and we can use $\{A(G_2;k):$~$k\in [1,7]\}$ to form an every-zero additive associative Topcode$^+$-matrix group by the additive v-operation defined in the equations (\ref{eqa:group-operation-1}) and (\ref{eqa:group-operation-2}). Correspondingly, each $A(G_2;k)$ derives a TB-paw $T^k_b(A(G_2))$ along the red line $L$ under the adjacent ve-value matrix $A(G_2)$ shown in Fig.\ref{fig:5-other-rules}(a), so the number string set $\{T^k_b(A(G_2)):~k\in [1,7]\}$, also, forms an every-zero additive associative number string group by the additive v-operation.

\begin{figure}[h]
\centering
\includegraphics[width=8.6cm]{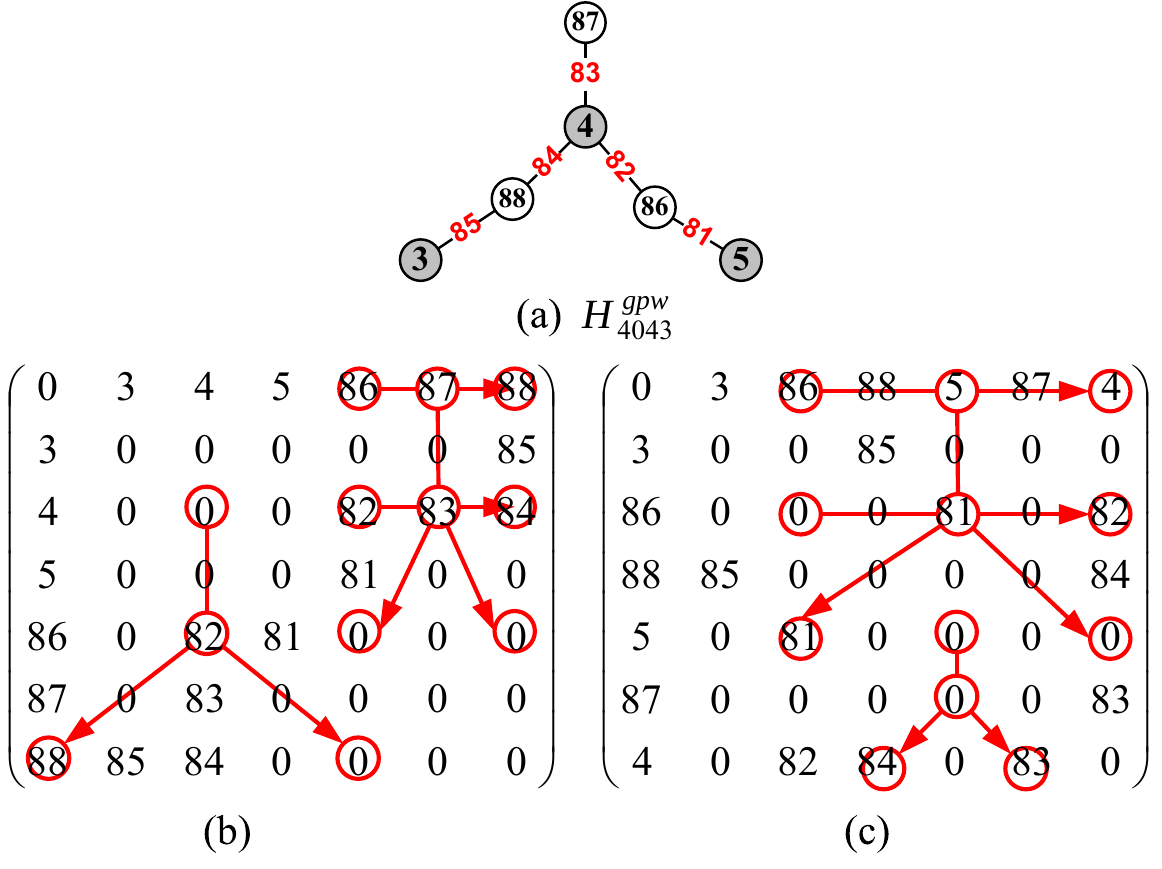}\\
\caption{\label{fig:ren-ve-value-matrix} {\small A Hanzi $H_{4043}$ with its Hanzi-gpw $H^{gpw}_{4043}$ and two adjacency ve-value matrices $A^{(1)}(H_{4043})$ and $A^{(2)}(H_{4043})$ cited from \cite{Yao-Mu-Sun-Sun-Zhang-Wang-Su-Zhang-Yang-Zhao-Wang-Ma-Yao-Yang-Xie2019}.}}
\end{figure}

By an adjacency ve-value matrix $A^{(1)}(H_{4043})$ shown in Fig.\ref{fig:ren-ve-value-matrix}, we have a TB-paw $T_b(11)\uplus T_b(12)$, where $$T_b(11)=868788828384870830830$$ and $T_b(12)=00820888200$. Another adjacency ve-value matrix $A^{(2)}(H_{4043})$ gives us another TB-paw $T_b(21)\uplus T_b(22)$, where $T_b(21)=86885874008108250810818100$ and $T_b(22)=0084083$. In this example, we have shown a technique based on adjacent ve-value matrices for producing TB-paws by Hanzis, and it is possible to generalize this technique to be other more complex cases.

\subsection{Linear system of Hanzi equations} Let us see an example shown in Fig.\ref{fig:Hanzi-linear-system-1} and Fig.\ref{fig:Hanzi-linear-system-2}. Using a Hanzi-string $T=H_{4476}H_{4734}H_{4662}H_{4311}$ (it is a \emph{Chinese idiom}), we set up a Hanzi-GB2312-80 matrix (also, coefficient matrix) $A(T)$ shown in Fig.\ref{fig:Hanzi-linear-system-1}, and get a system $Y=A(T)X$ of linear Hanzi equations. Next, we put $x_1,x_2,x_3,x_4$ into $X$ with $x_1=2,x_2=0,x_3=1,x_4=6$ of a Hanzi $H_{2016}$ from \cite{GB2312-80}, and then solve another Hanzi $H_{6532}$ by the system $Y=A(T)X$, and moreover $Y=X$ when $A(T)$ is the unit matrix, and $X=A^{-1}(T)Y$ as if the determinant $|A(T)|\neq 0$.

\begin{figure}[h]
\centering
\includegraphics[width=6.8cm]{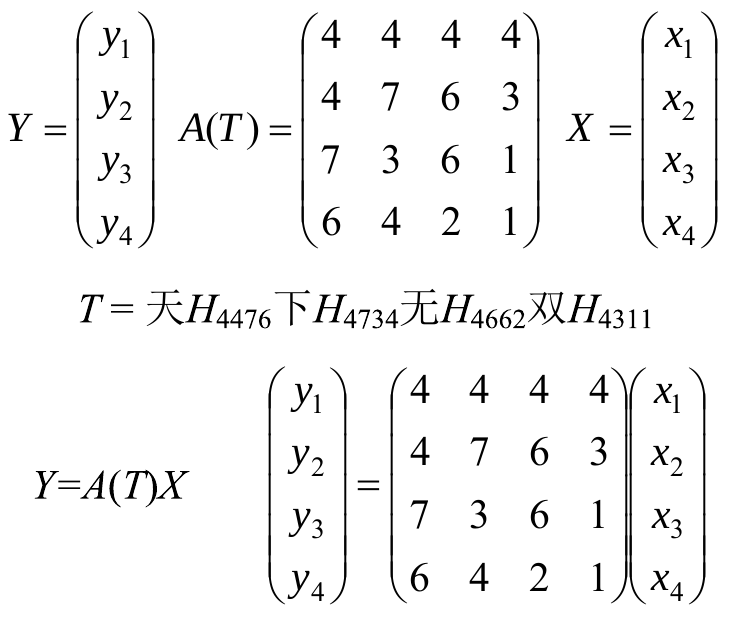}\\
\caption{\label{fig:Hanzi-linear-system-1} {\small A Hanzi-string $T=H_{4476}H_{4734}H_{4662}H_{4311}$, and a system $Y=A(T)X$ of linear Hanzi equations made by a known Hanzi $H_{x_1x_2x_3x_4}$ and a unknown Hanzi $H_{y_1y_2y_3y_4}$ based on \cite{GB2312-80}.}}
\end{figure}

\begin{figure}[h]
\centering
\includegraphics[width=7.2cm]{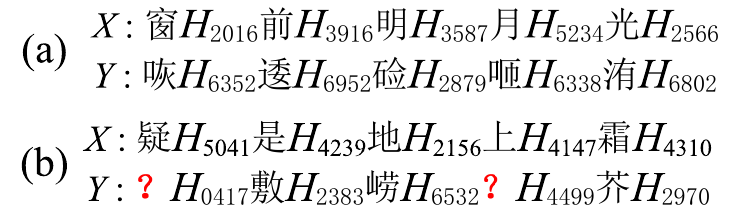}\\
\caption{\label{fig:Hanzi-linear-system-2} {\small Based on the system $Y=A(T)X$ shown in Fig.\ref{fig:Hanzi-linear-system-1}, some $H_{y_1y_2y_3y_4}$ in (b) are unknown in \cite{GB2312-80}.}}
\end{figure}
In Fig.\ref{fig:Hanzi-linear-system-2}, we have
\begin{equation}\label{eqa:c3xxxxx}
y_i=\beta_{i,1}x_1+\beta_{i,2}x_2+\beta_{i,3}x_3+\beta_{i,4}x_4~(\bmod~10)
\end{equation}
for $i\in [1,4]$, where $(\beta_{1,1}~\beta_{1,2}~\beta_{1,3}~\beta_{1,4})=(4~4~4~4)$, $(\beta_{2,1}~\beta_{2,2}~\beta_{2,3}~\beta_{2,4})=(4~7~6~3)$, $(\beta_{3,1}~\beta_{3,2}~\beta_{3,3}~\beta_{3,4})=(7~3~6~1)$ and $(\beta_{4,1}~\beta_{4,2}~\beta_{4,3}~\beta_{4,4})=(6~4~2~1)$. From Fig.\ref{fig:Hanzi-linear-system-2}(a), let
$${
\begin{split}
&X_1=(2~0~1~6)^{T},X_2=(3~9~1~6)^{T},X_3=(3~5~8~7)^{T},\\
&X_4=(5~2~3~4)^{T},X_5=(2~5~6~6)^{T};
\end{split}}
$$
so we have solved
$${
\begin{split}
&Y_1=(6~3~5~2)^{T},Y_2=(6~9~5~2)^{T},Y_3=(2~8~7~9)^{T},\\
&Y_4=(6~3~3~8)^{T},Y_5=(6~8~0~2)^{T},
\end{split}}
$$
and get five equations
\begin{equation}\label{eqa:Hanzi-linear-system-3}
Y_i=A(T)X_i,~i\in [1,5].
\end{equation}
Finally, we get a matrix equation by the \emph{union-addition operation} ``$\uplus$'' as:
\begin{equation}\label{eqa:Hanzi-linear-system-33}
\uplus^5_{i=1}Y_i=\uplus^5_{i=1}A(T)X_i=A(T)\uplus^5_{i=1}X_i.
\end{equation}

Through the matrix equation (\ref{eqa:Hanzi-linear-system-33}), we have translated a Chinese sentence $C_{pub}=H_{2016}H_{3916}H_{3587}H_{5234}H_{2566}$ (as a public key) into another Chinese sentence $C_{pri}=H_{6352}H_{6952}H_{2879}H_{6338}H_{6802}$ (as a private key), see Fig.\ref{fig:Hanzi-linear-system-3}. From these two sentences $C_{pub}$ and $C_{pri}$, we get two TB-paws
$$T_b(C_{pub})=20163916358752342566$$
and
$$T_b(C_{pri})=63526952287963386802.$$
Clearly, for a fixed public key $C_{pub}$, we may have many private keys $C_{pri}$ from different Hanzi-GB2312-80 matrices.
\begin{figure}[h]
\centering
\includegraphics[width=8cm]{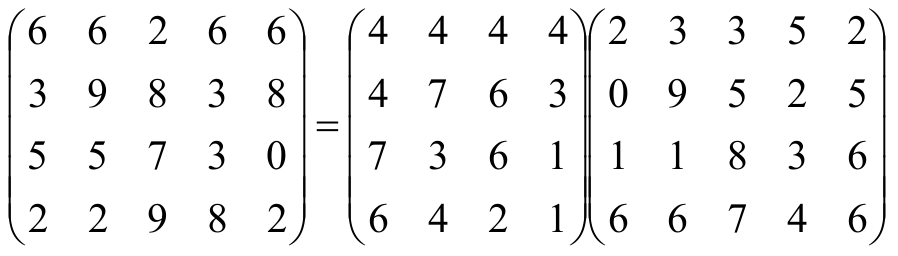}\\
\caption{\label{fig:Hanzi-linear-system-3} {\small A matrix equation based on the matrix equation (\ref{eqa:Hanzi-linear-system-33}).}}
\end{figure}

We generalize the matrix equation (\ref{eqa:Hanzi-linear-system-33}) by setting matrices $Y'_i=(y_{i,1}~y_{i,2}~ y_{i,3}~y_{i,4})^{T}_{4\times 1}$, $X'_i=(x_{i,1}~x_{i,2}~ x_{i,3}~x_{i,4})^{T}_{4\times 1}$ and $(a_{i,j})_{4\times 4}$ with $0\leq a_{i,j}\leq 9$, such that $X_{4\times m}=\uplus^m_{i=1}X'_i$ and $Y_{4\times m}=\uplus^m_{i=1}Y'_i$. We define the dot product $y_{i,j}$ of two vectors $(a_{i,1}~a_{i,2}~a_{i,3}~a_{i,4})$ and $(x_{j,1}~x_{j,2}~ x_{j,3}~x_{j,4})^{T}$ as
$$
{
\begin{split}
y_{i,j}&=(a_{i,1}~a_{i,2}~a_{i,3}~a_{i,4})\bullet (x_{j,1}~x_{j,2}~ x_{j,3}~x_{j,4})^{T}\\
&=a_{i,1}x_{j,1}+a_{i,2}x_{j,2}+a_{i,3}x_{j,3}+a_{i,4}x_{j,4}~(\bmod~10),
\end{split}}
$$ so we get a matrix equation below
\begin{equation}\label{eqa:generalize-Hanzi-linear-system}
Y_{4\times m}=(a_{i,j})_{4\times 4}X_{4\times m},
\end{equation}
where $Y_{4\times m}$ is an \emph{unknown matrix}, $(a_{i,j})_{4\times 4}$ is a \emph{coefficient matrix} and $X_{4\times m}$ is a \emph{known matrix}.

\subsection{Equations based on Hanzi-matrices}

We define another type of Hanzi-matrices as follows:

\begin{equation}\label{eqa:han-matrix-1}
\centering
{
\begin{split}
A^r_{han}&= \left(
\begin{array}{ccccc}
a_{1} & b_{1} & c_{1} & d_{1}\\
a_{2} & b_{2} & c_{2} & d_{2}\\
\cdots & \cdots& \cdots &\cdots\\
a_{m} & b_{m} & c_{m} & d_{m}
\end{array}
\right)=\left(\begin{array}{c}
X_1\\
X_2\\
\cdots\\
X_m
\end{array} \right)\\
&=(X_1~X_2~\cdots ~X_m)^{T}_{m\times 4}
\end{split}}
\end{equation}
where $X_k=(a_{k}~ b_{k}~ c_{m}~ d_{k})$ with $k\in [1,m]$ corresponds a Hanzi $H_{a_{k}b_{k}c_{m}d_{k}}$ in \cite{GB2312-80}, and
\begin{equation}\label{eqa:han-matrix-11}
\centering
{
\begin{split}
A^c_{han}&= \left(
\begin{array}{ccccc}
x_{1} & x_{2} & \cdots & x_{n}\\
y_{1} & y_{2} & \cdots & y_{n}\\
z_{1} & z_{2} & \cdots & z_{n}\\
w_{1} & w_{2} & \cdots & w_{n}
\end{array}
\right)\\
&=(Y_1~Y_2~\cdots ~Y_n)_{4\times n}
\end{split}}
\end{equation}
where $Y_j=(x_{j}~ y_{j}~ z_{j}~ w_{j})^{T}$ with $j\in [1,n]$ corresponds a Hanzi $H_{x_{j}y_{j}z_{j}w_{j}}$ in \cite{GB2312-80}. Thereby, we define
\begin{equation}\label{eqa:han-matrix-12}
\centering
{
\begin{split}
&\quad A^{r(\bullet) c}_{han}=A^r_{han}(\bullet) A^c_{han}\\
&=(X_1~X_2~\cdots ~X_m)^{T}_{m\times 4}(\bullet) (Y_1~Y_2~\cdots ~Y_n)_{4\times n}
\end{split}}
\end{equation}
where $A^{r\times c}_{han}=(\alpha_{i,j})_{m\times n}$ and ``$(\bullet)$'' is an abstract operation. We define two operations that differ from the additive v-operation defined in (\ref{eqa:group-operation-1}) and (\ref{eqa:group-operation-2}):

\vskip 0.2cm

(i) \emph{\textbf{Multiplication ``$\bullet $'' on components of two vectors}.} We define
\begin{equation}\label{eqa:han-matrix-multiplication}
{
\begin{split}
X_k\bullet Y_j&=(a_{k}~ b_{k}~ c_{k}~ d_{k})\bullet (x_{j}~ y_{j}~ z_{j}~ w_{j})^{T}\\
&=a_{k,j}b_{k,j}c_{k,j}d_{k,j}
\end{split}}
\end{equation}
where $a_{k,j}=a_{k}\cdot x_{j}~(\bmod~10)$, $b_{k,j}=b_{k}\cdot y_{j}~(\bmod~10)$, $c_{k,j}=c_{k}\cdot z_{j}~(\bmod~10)$ and $d_{k,j}=d_{k}\cdot w_{j}~(\bmod~10)$. So $A^{r\bullet c}_{han}=A^r_{han}\bullet A^c_{han}$ (see an example shown in Fig.\ref{fig:6-hanzi-matrices} and Fig.\ref{fig:7-hanzi-matrices}).

(ii) \emph{\textbf{Addition ``$\oplus$'' on components of two vectors}.} We define
\begin{equation}\label{eqa:han-matrix-addition}
{
\begin{split}
X_k\oplus Y_j&=(a_{k}~ b_{k}~ c_{m}~ d_{k})\oplus (x_{j}~ y_{j}~ z_{j}~ w_{j})^{T}\\
&=\alpha_{k,j}\beta_{k,j}\gamma_{k,j}\delta_{k,j}
\end{split}}
\end{equation}
where $\alpha_{k,j}=a_{k}+x_{j}~(\bmod~10)$, $\beta_{k,j}=b_{k}+y_{j}~(\bmod~10)$, $\gamma_{k,j}=c_{m}+z_{j}~(\bmod~10)$ and $\delta_{k,j}=d_{k}+w_{j}~(\bmod~10)$, that is $A^{r\oplus c}_{han}=A^r_{han}\oplus A^c_{han}$ (see an example shown in Fig.\ref{fig:6-hanzi-matrices} and Fig.\ref{fig:7-hanzi-matrices}).

\vskip 0.2cm

\begin{figure}[h]
\centering
\includegraphics[width=8.6cm]{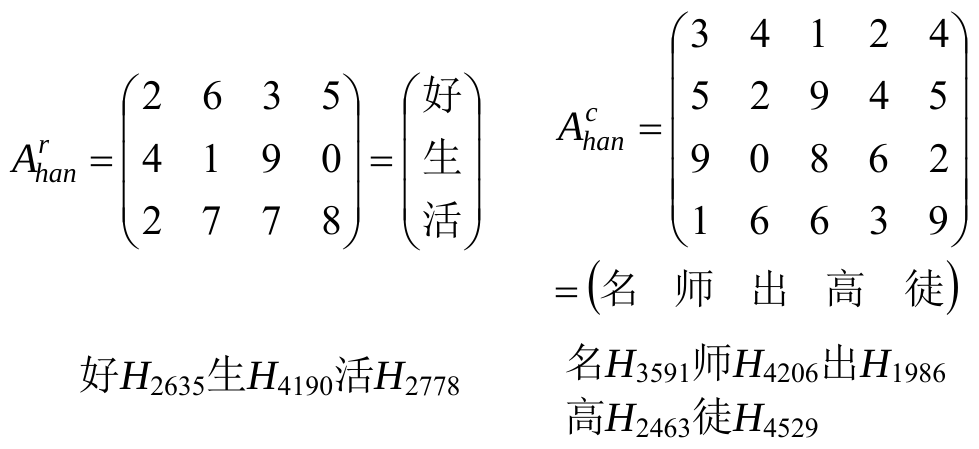}\\
\caption{\label{fig:6-hanzi-matrices} {\small Examples for the multiplication and the addition on components of two vectors.}}
\end{figure}

\begin{figure}[h]
\centering
\includegraphics[width=8.6cm]{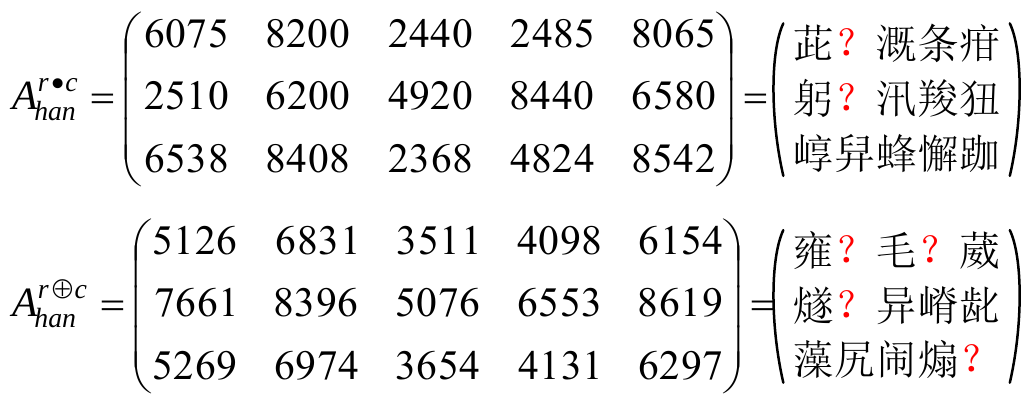}\\
\caption{\label{fig:7-hanzi-matrices} {\small Two matrices obtained from Fig.\ref{fig:6-hanzi-matrices}.}}
\end{figure}

Observe Fig.\ref{fig:7-hanzi-matrices}, there are several ``?'' in two matrices $A^{r\bullet c}_{han}=(\alpha_{i,j})_{3\times 5}$ and $A^{r\oplus c}_{han}=(\alpha_{i,j})_{3\times 5}$, since no Hanzi with Chinese codes in \cite{GB2312-80} corresponds to these ``?''.

Thereby, each element of the matrix $A^{r\bullet c}_{han}=(\alpha_{i,j})_{m\times n}$ (or $A^{r\oplus c}_{han}=(\alpha_{i,j})_{m\times n}$) is a number of four bytes which corresponds a Hanzi in \cite{GB2312-80}, so we call $A^{r\bullet c}_{han}=(\alpha_{i,j})_{m\times n}$ (or $A^{r\oplus c}_{han}$) a \emph{Hanzi-matrix}. On the other hands, we can regard $A^{r\bullet c}_{han}=(\alpha_{i,j})_{m\times n}$ (or $A^{r\oplus c}_{han}$) as an \emph{authentication} for the \emph{public key} $A^r_{han}$ and the \emph{private key} $A^c_{han}$.

\begin{figure}[h]
\centering
\includegraphics[width=8cm]{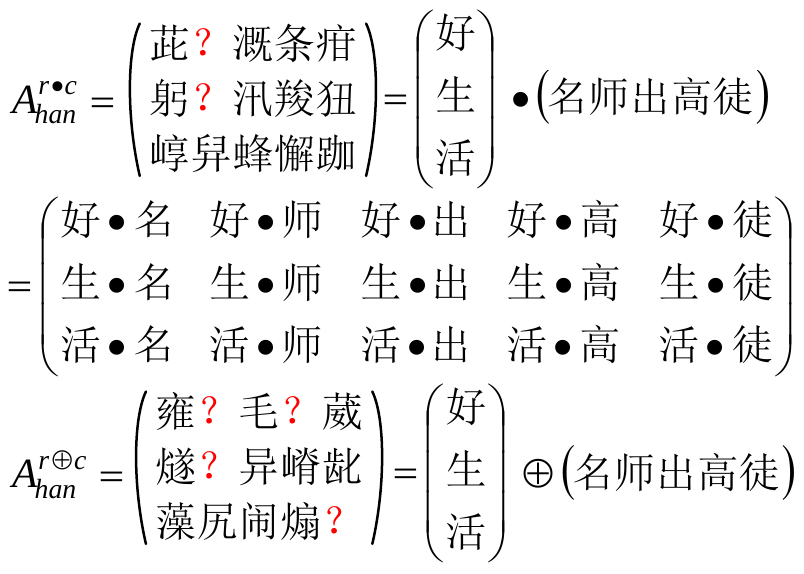}\\
\caption{\label{fig:8-hanzi-matrices} {\small Two pure Hanzi matrices from Fig.\ref{fig:7-hanzi-matrices}.}}
\end{figure}

\begin{figure}[h]
\centering
\includegraphics[width=8.6cm]{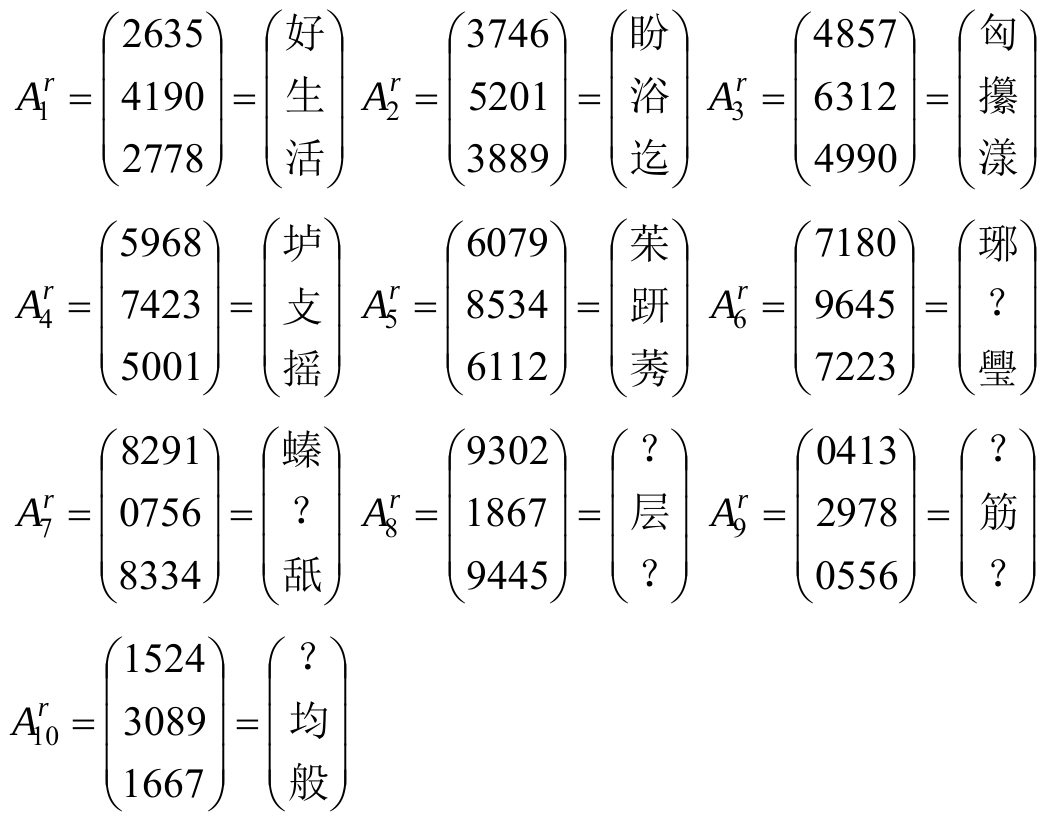}\\
\caption{\label{fig:hao-sheng-huo-group} {\small A Hanzi-group based on the additive v-operation under $\bmod~10$.}}
\end{figure}

If we have an encoding of Chinese characters, which contains $10^4$ different Chinese characters, then our Hanzi-matrices contain no ``?''.

\subsection{Matrices with elements to be graphs or Topsnut-gpws}

\vskip 0.4cm

\subsubsection{Coefficient matrices are number matrices}\quad A graph-matrix is one with every element to be a graph or a Topsnut-gpw. Let
\begin{equation}\label{eqa:coefficient-matrix}
\centering
{
\begin{split}
A_{coe}= \left(
\begin{array}{ccccc}
a_{1,1} & a_{1,2} & \cdots & a_{1,n}\\
a_{2,1} & a_{2,2} & \cdots & a_{2,n}\\
\cdots & \cdots & \cdots & \cdots \\
a_{m,1} & a_{m,2} & \cdots & a_{m,n}\\
\end{array}
\right)
\end{split}}
\end{equation}
and two graph vectors $Y(H)=(H_1~H_2~\cdots ~H_m)^{T}$ and $X(G)=(G_1$ $G_2$ $\cdots $ $G_n)^{T}$, where $G_i$ and $H_j$ are colored graphs with $i\in [1,n]$ and $j\in [1,m]$. We have a system of graph equations
\begin{equation}\label{eqa:graph-equation}
\centering
Y(H)=A_{coe}\cup X(G)
\end{equation}
with $H_i=\bigcup^n_{j=1}a_{i,j}G_j$, where ``$\bigcup$'' is the union operation on graphs of graph theory, and $a_{i,j}G_j$ is a disconnected graph having $a_{i,j}$ components that are isomorphic to $G_i$, that is
{\footnotesize
\begin{equation}\label{eqa:graph-equations}
\centering
{
\begin{split}
\left(\begin{array}{c}
H_1\\
H_2\\
\cdots\\
H_m
\end{array} \right)= \left(
\begin{array}{ccccc}
a_{1,1} & a_{1,2} & \cdots & a_{1,n}\\
a_{2,1} & a_{2,2} & \cdots & a_{2,n}\\
\cdots & \cdots & \cdots & \cdots \\
a_{m,1} & a_{m,2} & \cdots & a_{m,n}\\
\end{array}
\right)\bigcup \left(\begin{array}{c}
G_1\\
G_2\\
\cdots\\
G_n
\end{array} \right)
\end{split}}
\end{equation}
}

\vskip 0.4cm

\subsubsection{Coefficient matrices have elements being graphs}\quad We define a \emph{graph coefficient matrix} $A_{graph}$ by
\begin{equation}\label{eqa:graph-coefficient-matrix}
\centering
{
\begin{split}
A_{graph}= \left(
\begin{array}{ccccc}
G_{1,1} & G_{1,2} & \cdots & G_{1,q}\\
G_{2,1} & G_{2,2} & \cdots & G_{2,q}\\
\cdots & \cdots & \cdots & \cdots \\
G_{p,1} & G_{m,2} & \cdots & G_{p,q}\\
\end{array}
\right)_{p\times q}
\end{split}}
\end{equation}
and two \emph{graph vectors} $Y(L)=(L_1~L_2~\cdots ~L_p)^{T}$ and $X(T)=(T_1$ $T_2$ $\cdots $ $T_q)^{T}$, where $T_i$ and $L_j$ are colored graphs with $i\in [1,q]$ and $j\in [1,p]$. Thereby, we have a system of graph equations
\begin{equation}\label{eqa:graph-coefficient-matrix-equation}
\centering
Y(L)=A_{graph}(\ast) X(T)
\end{equation}
with $L_i=\bigcup^q_{j=1}[G_{i,j}(\ast)T_j]$, where ``$(\ast)$'' is a graph operation on colored graphs in graph theory (Ref. \cite{Bang-Jensen-Gutin-digraphs-2007, Bondy-2008}).

\vskip 0.4cm

\subsubsection{Topcode-matrices with variables}\quad We show such a matrix in Fig.\ref{fig:pan-matrix-000}, and get a TB-paw $T_b(H_{4043};s,i)$ from the Topcode-matrix $T(H^{gpw}_{4043})$ made by a Topsnut-gpw $H^{gpw}_{4043}$ based on a Hanzi $H_{4043}$ as follows:
$${
\begin{split}
&T_b(H_{4043};s,i)=(s-1)(s-i-2)(i+1)(s-i-1)\\
&s(s-i)i(i+1)(s-i-3)(s-2)(s-i-4)(i+2).
\end{split}}$$

\begin{figure}[h]
\centering
\includegraphics[width=8.6cm]{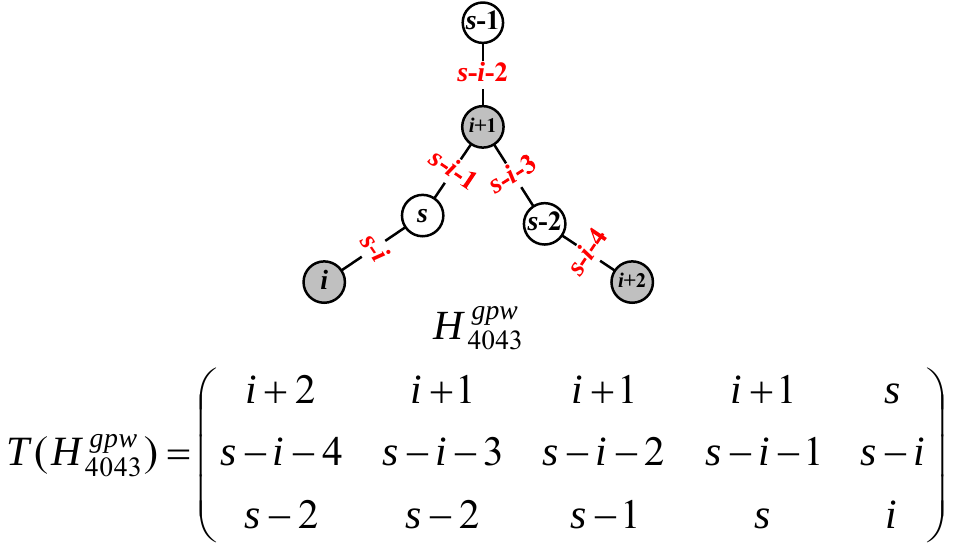}\\
\caption{\label{fig:pan-matrix-000} {\small A Topcode-matrix with elements of variables $s$ and $i$ cited from \cite{Yao-Mu-Sun-Sun-Zhang-Wang-Su-Zhang-Yang-Zhao-Wang-Ma-Yao-Yang-Xie2019}.}}
\end{figure}

\vskip 0.4cm

\subsubsection{Topcode-matrices from analytic Hanzis}\quad In Fig.\ref{fig:analytic-hanzi-matrix}, a Hanzi $H_{4585}$ was put into xOy-plane, so we get a new type of matrices, called \emph{analytic matrices}.

\begin{figure}[h]
\centering
\includegraphics[width=8.6cm]{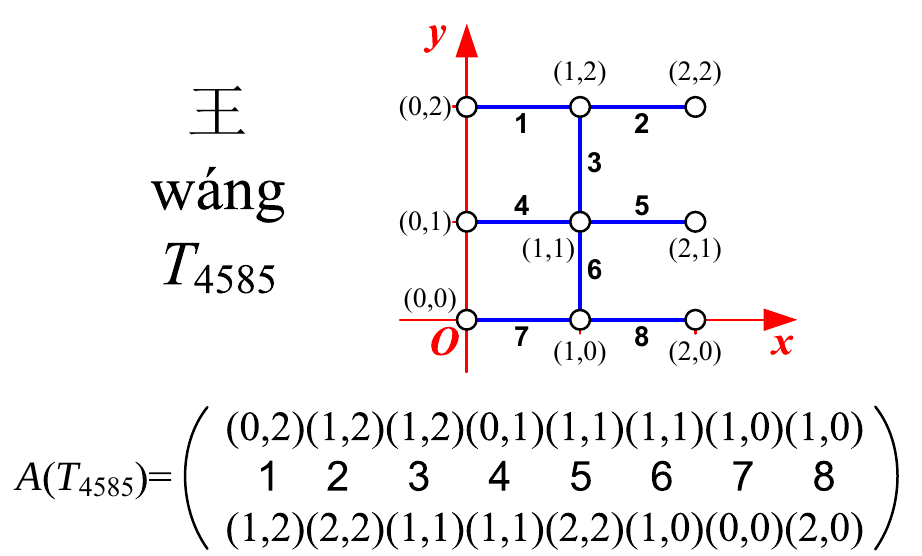}\\
\caption{\label{fig:analytic-hanzi-matrix} {\small A Hanzi $H_{4585}$ and its analytic graph, as well as its Topcode-matrix with elements to be numbers and coordinates in xOy-plane cited from \cite{Yao-Mu-Sun-Sun-Zhang-Wang-Su-Zhang-Yang-Zhao-Wang-Ma-Yao-Yang-Xie2019}.}}
\end{figure}

\vskip 0.4cm

\subsubsection{Dynamic Topcode-matrices from dynamic networks}\quad Let $\mathcal {N}(t)$ be a dynamic network at time step $t$. Thereby, $\mathcal {N}(t)$ has an its own \emph{dynamic Topcode-matrix} $T_{code}(\mathcal {N}(t))$ at each time step $t$, where $T_{code}(\mathcal {N}(t))=(X(t)~E(t)~Y(t))^{T}$ with $X(t)=(x^t_1$ $x^t_2$ $\cdots $ $x^t_{q(t)})$, $E(t)=(e^t_1~e^t_2~\cdots ~e^t_{q(t)})$ and $Y(t)=(y^t_1~y^t_2~\cdots ~y^t_{q(t)})$. Conversely, $T_{code}(\mathcal {N}(t))$ may corresponds two or more dynamic networks at time step $t$, and these dynamic networks match each other (Ref. \cite{Li-Gu-Dullien-Vinyals-Kohli-arXiv2019}). Obviously, there are interesting topics on dynamic Topcode-matrices.

\section{Connections between matrices and groups}

As to answer ``What applications do Topcode-matrices have?'' proposed in the previous section, we will try to join our matrices by graphical methods in this section. Conversely, some contents are colored/labelled graphs with the matrices introduced here. We will deal with these two problems: Label a graph $H$ by the elements of a graph group $\{F(G);\oplus\}=\{G_{x,i}$, $G_{e,j}$, $G_{y,k}:1\leq i,j,k\leq q\}$, conversely, the graph $H$ is made by the graph group $\{F(G);\oplus\}$. Thereby, we get $H$'s Topcode-matrix as follows
\begin{equation}\label{eqa:Topcode-matrix}
\centering
{
\begin{split}
M_{atrix}(H)&= \left(
\begin{array}{ccccc}
G_{x,1} & G_{x,2} & \cdots & G_{x,q}\\
G_{e,1} & G_{e,2} & \cdots & G_{e,q}\\
G_{y,1} & G_{y,2} & \cdots & G_{y,q}
\end{array}
\right)\\
&=
\left(\begin{array}{c}
G_X\\
G_E\\
G_Y
\end{array} \right)=(G_X~G_E~G_Y)^{T}
\end{split}}
\end{equation}where three graph vectors $G_X=(G_{x,1}~G_{x,2}~\cdots ~G_{x,q})$, $G_E=(G_{e,1}$ $G_{e,2}$ $\cdots ~G_{e,q})$ and $G_Y=(G_{y,1}~G_{y,2}~\cdots ~G_{y,q})$. We say $M_{atrix}(H)$ to be \emph{graphicable}.

\subsection{Connections between particular Topcode-matrices}

\begin{thm}\label{thm:equivalent-Topcode-matrices}
A Topcode-matrix $T_{code}$ Definition \ref{defn:Topcode-matrix} and corresponding a tree is a set-ordered graceful Topcode-matrix if and only if it is equivalent with

(1) a set-ordered odd-graceful Topcode-matrix.

(2) a set-ordered edge-magic total Topcode-matrix.

(3) a 6C-Topcode-matrix (resp. a 6C-odd-Topcode-matrix).
\end{thm}
\begin{proof} Let $T_{code}=(X~E~Y)^{T}_{3\times q}$ be a set-ordered graceful Topcode-matrix, where \emph{v-vector} $X=(x_1 ~ x_2 ~ \cdots ~x_q)$ and \emph{v-vector} $Y=(y_1 ~ y_2 ~\cdots ~ y_q)$ hold $\max X< \min Y$ true such that $1\leq e_i=y_i-x_i\leq q$, and moreover $E=(e_1 ~ e_2 ~ \cdots ~e_q)=(1 ~ 2 ~ \cdots ~q)$. For the purpose of statement, we write $\{x^k_i:i\in[1,q]\}=X^k=(x^k_1 ~ x^k_2 ~ \cdots ~x^k_q)$, $\{e^k_i:i\in[1,q]\}=E^k=(e^k_1 ~ e^k_2 ~ \cdots ~e^k_q)$ and $\{y^k_i:i\in[1,q]\}=Y^k=(y^k_1 ~ y^k_2 ~ \cdots ~y^k_q)$ with $k\geq 0$, where $X^0=X$, $E^0=E$ and $Y^0=Y$. By the property of a set-ordered graceful Topcode-matrix, we have $(XY)^*=\{x_1,x_2,\dots,x_s\}\cup \{y_1,y_2,\dots,y_t\}$ with $s+t=p$, and moreover $0=x_1<x_i<x_{i+1}\leq s-1$ for $i\in[2,s-1]$ and $s\leq y_j<y_{j+1}<y_t=p$ for $j\in[1,t-2]$. We have $y_{t-j+1}=s+p-y_j$ with $j\in[1,t]$

\vskip 0.4cm

\textbf{Proof of claim (1). }\emph{The proof of ``if''.} We make a set-ordered odd-graceful Topcode-matrix $T^1_{code}=(X^1$ $E^1$ $Y^1)^{T}_{3\times q}$ in the way: For $i\in [1,q]$, set up $x^1_i=2x_i$ for $x^1_i\in X^1$ and $y^1_i=2y_i-1$ for $y^1_i\in Y^1$, as well as $e^1_i=|y^1_i-x^1_i|$. Clearly, $\max X^1< \min Y^1$, and $1\leq e^1_i=y^1_i-x^1_i$, such that $e^1_i=y^1_i-x^1_i=2(y_i-x_i)-1\in [1,2q-1]^o$. We claim that $T^1_{code}$ is a set-ordered odd-graceful Topcode-matrix. Write this procedure as $F_1$, so we have $T^1_{code}=F_1(T_{code})$, a transformation on two Topcode-matrices.

\emph{The proof of ``only if''.} Suppose that $T^1_{code}=(X^1$ $E^1$ $Y^1)^{T}_{3\times q}$ is a set-ordered odd-graceful Topcode-matrix with $X^1=\{x^1_i:i\in[1,q]\}$, $Y^1=\{y^1_i:i\in[1,q]\}$ and $e^1_i=y^1_i-x^1_i\in E^1=[1,2q-1]^o$, here $\max X^1< \min Y^1$. So, each $x^1_i$ must be even, and each $y^1_i$ must be odd. We make a transformation $f$: $x_i=\frac{1}{2}x^1_i$ and $y_i=\frac{1}{2}(y^1_i+1)$. Then we get $x_i=\frac{1}{2}x^1_i\in X=(x_1 ~ x_2 ~ \cdots ~x_q)$ and $y_i=\frac{1}{2}(y^1_i+1)\in Y=(y_1 ~ y_2 ~\cdots ~ y_q)$, as well as
$${
\begin{split}
e_i&=y_i-x_i=\frac{1}{2}(y^1_i+1)-\frac{1}{2}x^1_i\\
&=\frac{1}{2}(y^1_i+1-x^1_i)\in [1,q],
\end{split}}$$ which means $E=(1 ~ 2 ~ \cdots ~q)$, so $T_{code}=(X~E~Y)^{T}_{3\times q}$ is a set-ordered graceful Topcode-matrix.

\vskip 0.4cm

\textbf{Proof of claim (2). }\emph{Necessary.} We have a matrix $T^2_{code}=(X^2$ $E^2$ $Y^2)^{T}_{3\times q}$ defined as: $x^2_i=x_i+1$ for $x^2_i\in X^2$ and $i\in [1,s]$, and $y^2_j=y_{t-j+1}+1$ for $y^2_i\in Y^2$ with $j\in [1,t]$, as well as $e^2_i=e_i+p$ with $i\in [1,q]$. Clearly, $\max X^2=\max \{x^2_i:i\in [1,q]\}< \min Y^2=\min \{y^2_i:i\in [1,q]\}$, and
$${
\begin{split}
x^2_i+ e^2_i+y^2_i&=x_i+1+ e_i+p+y_{t-j+1}+1\\
&=x_i+1+y_i-x_i+p+s+p-y_i+1\\
&=s+2p+2
\end{split}}$$ a constant, so we claim that $T^2_{code}$ to be a set-ordered edge-magic total Topcode-matrix (also, a super edge-magic total Topcode-matrix). Write this transformation as $F_2$, so $T^2_{code}=F_2(T_{code})$.

\emph{Sufficiency.} Let $T^2_{code}=(X^2$ $E^2$ $Y^2)^{T}_{3\times q}$ be a set-ordered edge-magic total Topcode-matrix, and $x^2_i+ e^2_i+y^2_i=s+2p+2$ for $i\in [1,q]$, and $E^2=[p+1,p+q]$, where $(X^2Y^2)^*$ is the set of all different numbers of $X^2$ and $Y^2$, and $(X^2Y^2)^*=\{x^2_1,x^2_2,\dots,x^2_s\}\cup \{y^2_1,y^2_2,\dots,y^2_t\}$ with $p=s+t$. Notice that $T^2_{code}$ holds $\max X^2=\max \{x^2_i:i\in [1,q]\}< \min Y^2=\min \{y^2_i:i\in [1,q]\}$. Now, we set $x_i=x^2_i-1$ with $i\in [1,s]$ and $y_j=y^2_{t-j+1}-1$ with $j\in [1,t]$, and set $e_i=e^2_i-p$ with $i\in [1,q]$. Then we get a matrix $T_{code}=(X$ $E$ $Y)^{T}_{3\times q}$. Clearly, $\{e_i:i\in [1,q]\}=[1,q]$ and furthermore
$${
\begin{split}
e_i&=e^2_i-p=(s+2p+2)-(x^2_i+y^2_i)-p\\
&=s+p+2-(x_i+1+s+p-y^2_{t-i+1}+1)\\
&=y_i-x_i
\end{split}}$$ and $E^*=[1,q]$. Thereby, $T_{code}$ is a set-ordered graceful Topcode-matrix, as desired.

\vskip 0.4cm

\textbf{Proof of claim (3). }\emph{The proof of ``if''.} A Topcode-matrix $T^3_{code}$ is a \emph{6C-Topcode-matrix} if there are the following facts:

(0) $(X^3Y^3)^*\cup E^3=[1,p+q]$.

\emph{Solution. } We, for $i\in [1,q]$, set
\begin{equation}\label{eqa:6c-matrix-vs-set-ordered-graceful-matrix-11}
x^3_i=x_i+1,\quad y^3_i=y_i+1
\end{equation}
which indicates $(X^3Y^3)^*=[1,p]$; and
\begin{equation}\label{eqa:6c-matrix-vs-set-ordered-graceful-matrix-22}
e^3_i=e_{q-i+1}+p,
\end{equation}
which means $E^3=[p+1,p+q]$.

(1) (e-magic) There exists a constant $k$ such that $e^3_i+|x^3_i-y^3_i|=k$ for $i\in [1,q]$.

\emph{Solution. } We have
$${
\begin{split}
e^3_i+|x^3_i-y^3_i|&=e_{q-i+1}+p+y_i+1-(x_i+1)\\
&=e_{q-i+1}+p+e_i\\
&=p+q+1,
\end{split}}$$
 since $e_i+e_{q-i+1}=q+1$.

(2) $e^3_i=|x^3_j-y^3_j|$ or $e^3_i=2M-|x^3_j-y^3_j|$ for each $i\in [1,q]$ and some $j\in [1,q]$, and $M=\theta(q,|(XY)^*|)$.

\emph{Solution. } From
$${
\begin{split}
e^3_i&=e_{q-i+1}+p=p+q+1-e_i\\
&=p+1+e_j=|(X^3Y^3)^*|+1+y_j-x_j,
\end{split}}$$
this holds true.

(3) (ee-balanced) Let $s(e^3_i)=|x^3_i-y^3_i|-e^3_i$ for $i\in [1,q]$. There exists a constant $k'$ such that each $e^3_i$ with $i\in [1,q]$ matches with another $e^3_j$ holding $s(e^3_i)+s(e^3_j)=k'$ (or $2(q+|(X^3Y^3)^*|)+s(e^3_i)+s(e^3_j)=k'$, or $(|(X^3Y^3)^*|+q+1)+s(e^3_i)+s(e^3_j)=k'$) true.

\emph{Solution. } We have
$${
\begin{split}
s(e^3_i)&=|x^3_i-y^3_i|-e^3_i=y_i-x_i-(e_{q-i+1}+p)\\
&=e_i-(q+1-e_i+p)\\
&=2e_i-(p+q+1),
\end{split}}$$
and
$${
\begin{split}
s(e^3_{q-i+1})&=|x^3_{q-i+1}-y^3_{q-i+1}|-e^3_{q-i+1}\\
&=y_{q-i+1}-x_{q-i+1}-(e_i+p)\\
&=e_{q-i+1}-e_i-p\\
&=q+1-2e_i-p,
\end{split}}$$
as well as $s(e^3_i)+s(e^3_{q-i+1})=q+1-p-(p+q+1)=-2p$.

(4) (EV-ordered) $\min (X^3Y^3)^*>\max E^3$, or $\max (X^3Y^3)^*$ $<\min E^3$, or $(X^3Y^3)^*$ $\subseteq E^3$, or $E^3$ $\subseteq(X^3Y^3)^*$, or $(X^3Y^3)^*$ is an odd-set and $E^3$ is an even-set.

\emph{Solution. } Clearly, $\max (X^3Y^3)^*=p<p+1=\min E^3$.

(5) (ve-matching) there exists a constant $k''$ such that each $e^3_i\in E^3$ matches with $w\in (X^3Y^3)^*$ such that $e^3_i+w=k''$, and each vertex $x^3_i\in (X^3Y^3)^*$ matches with $e^3_t\in E^3$ such that $x^3_i+e^3_t=k''$.

\emph{Solution. } Notice that $(X^3Y^3)^*=[1,p]$, $E^3=[1,q]$. Take $w=e_i$, $e^3_i+w=e_{q-i+1}+p+e_i=q+1+p$, also, we have $w$ matches with $e^3_t=e^3_i$.

(6) $\max X^3<\min Y^3$ holds true.

\emph{Solution. } It is obvious, since $x^3_i=x_i+1<y^3_i=y_i+1$.

Write the above transformation as $F_3$, we get $T^3_{code}=F_3(T_{code})$.

\vskip 0.2cm

\emph{The proof of ``only if''.} Since a 6C-Topcode-matrix $T^3_{code}=(X^3~E^3~Y^3)^{T}_{3\times q}$ holds the facts shown in the proof of ``if'' above, so we, by (\ref{eqa:6c-matrix-vs-set-ordered-graceful-matrix-11}) and (\ref{eqa:6c-matrix-vs-set-ordered-graceful-matrix-22}), can set $x_i=x^3_i-1$, $y_i=y^3_i-1$, so $(XY)^*=[0,p-1]$ and $e_i=e^3_{q-i+1}-p=e_i+p-p$, so $E^*=[1,q]$; and $\max X<\min Y$ from $x_i=x^3_i-1<y^3_i-1=y_i$. Thereby, $T^3_{code}$ is equivalent to a set-ordered graceful Topcode-matrix $T_{code}=(X~E~Y)^{T}_{3\times q}$.

The proof of ``a 6C-odd-Topcode-matrix'' is similar with that of ``a 6C-Topcode-matrix'', so we omit it.
\end{proof}

\begin{rem} \label{rem:xxxxx}
Let $G(T_{code})$ be the set of graphs $G_1,G_2,\dots, G_m$ derived by a Topcode-matrix $T_{code}$ defined in Definition \ref{defn:Topcode-matrix}.

1. If $T_{code}$ is equivalent to each of Topcode-matrices $T^i_{code}$ defined in Definition \ref{defn:Topcode-matrix} with $G(T^i_{code})=\{G^i_1,G^i_2,\dots, G^i_m\}$ and $i\in [1,n]$, then there exists a transformation $f_i$ such that $G^i_j=f_i(G_j)$ with $j\in [1,m]$. Thereby, we claim that the set $G(T_{code})$ is equivalent to each set $G(T^i_{code})$ according to $T^i_{code}=f_i(T_{code})$ with $i\in [1,n]$.

2. Refereing \cite{Li-Gu-Dullien-Vinyals-Kohli-arXiv2019}, it shows that each graph $G_j$ matches with the graph $G^i_j$ from $T^i_{code}=f_i(T_{code})$, $j\in [1,m]$.

3. Each graph of $G(T_{code})$ is similar with others of $G(T_{code})$ by the Topcode-matrix $T_{code}$ of view (Ref. \cite{Li-Gu-Dullien-Vinyals-Kohli-arXiv2019}).

4. If $H^*\in G(T_{code})$ holds that $|V(H^*)|\leq |V(G_i)|$, then $H^*$ is unique.
\end{rem}

\subsection{Groups to graph networks}

Translating groups into ``graph networks'' here are really topological connections between the groups, rather than ``Graph networks'' investigated in \cite{Battaglia-27-authors-arXiv1806-01261v2}. Using the additive v-operation defined in (\ref{eqa:group-operation-1}) and (\ref{eqa:group-operation-2}) for Topcode$^+$-matrix groups, we can mapping all elements of a Topcode$^+$-matrix group to some graph network $H$ with evaluated colorings on its vertices and edges, see an example is shown in Fig.\ref{fig:group-vs-graphical}, in which $H_{\textrm{under}}$ is the underlying graph of the graph network $H$, and each of (a), (b) and (c) is $H$, where (a) and (b) are called \emph{ve-graceful evaluated colorings} of $H$, since there exists a function $f:V(H)\rightarrow F_6=\{T_i:i\in [1,6]\}$ such that $f(V(H))=F_6$ and $f(E(H))=F_6\setminus \{T_6\}$; however, (c) has no vertices were labelled with $T_1$ and $T_3$, so we call (c) an \emph{e-graceful evaluated coloring} of $H$. We can show the following result:

\begin{thm}\label{thm:e-graceful-weighted-value}
Given a tree $T$ of $n$ vertices and an every-zero Topcode$^+$-matrix group $\{F_n;\oplus\}$, then $T$ admits an e-graceful evaluated coloring on $\{F_n;\oplus\}$.
\end{thm}

\begin{rem}\label{eqa:e-graceful-weighted-value}
\begin{asparaenum}[1. ]
\item Theorem \ref{thm:e-graceful-weighted-value} tells us that there are two or more trees of $n$ vertices having e-graceful evaluated colorings based on an every-zero Topcode$^+$-matrix group $\{F_n;\oplus\}$. Conversely, there are two or more every-zero Topcode$^+$-matrix groups such that a tree has two or more e-graceful evaluated colorings based on them.
\item ``Graph network'' used here is not equivalent with that defined in \cite{Battaglia-27-authors-arXiv1806-01261v2}, roughly speaking, our graph network has: (i) Attribute: properties that can be encoded as a vector, set, or even another graph;
(ii) Attributed: edges and vertices have attributes associated with them; (iii) Global attribute: a graph-level attribute.
\item A graph network $H$ with its underlying graph to be a tree of $n$ vertices admits a \emph{ve-graceful evaluated coloring} $f$ if there is a Topcode$^+$-matrix group $\{F_n;\oplus\}$ (resp. a graph group, or a number string group), where $F_n=\{T_i:i\in[1,n]\}$, such that $f:V(H)\rightarrow \{F_n;\oplus\}$, and the vertex coloring set $f(V(H))=F_n$ and the edge coloring set $f(E(H))=F_n\setminus \{T_n\}$. If $f(V(H))$ is a proper set of $F_n$ and $f(E(H))=F_n\setminus \{T_n\}$, we call $f$ an \emph{e-graceful evaluated coloring}.

\quad In Fig.\ref{fig:group-vs-graphical}, the graph network (a) has its own Topcode-matrix $T_{code}(a)$ as follows
\begin{equation}\label{eqa:Topcode-matrix-in-matrices-a}
\centering
{
\begin{split}
T_{code}(a)= \left(
\begin{array}{ccccc}
T_{3} & T_{1} &T_{6} & T_{1}& T_{1}\\
T_{1} & T_{2} &T_{3} & T_{4}& T_{5}\\
T_{5} & T_{2} &T_{4} & T_{4}& T_{5}
\end{array}
\right)
\end{split}}
\end{equation} and moreover (b) and (c) have their own Topcode-matrices $T_{code}(b)$ and $T_{code}(c)$ as follows
\begin{equation}\label{eqa:Topcode-matrix-in-matrices-b}
\centering
{
\begin{split}
T_{code}(b)= \left(
\begin{array}{ccccc}
T_{6} & T_{3} &T_{5} & T_{3}& T_{3}\\
T_{1} & T_{2} &T_{3} & T_{4}& T_{5}\\
T_{4} & T_{2} &T_{1} & T_{4}& T_{5}
\end{array}
\right)
\end{split}}
\end{equation}
and
\begin{equation}\label{eqa:Topcode-matrix-in-matrices-c}
\centering
{
\begin{split}
T_{code}(c)= \left(
\begin{array}{ccccc}
T_{2} & T_{6} &T_{5} & T_{6}& T_{6}\\
T_{1} & T_{2} &T_{3} & T_{4}& T_{5}\\
T_{5} & T_{2} &T_{4} & T_{4}& T_{5}
\end{array}
\right)
\end{split}}
\end{equation}
Notice that all elements of three Topcode-matrices $T_{code}(a)$, $T_{code}(b)$ and $T_{code}(c)$ are matrices too, so these three are Topcode-matrices in Topcode-matrices. In general, the elements of a Topcode-matrix defined here may be number strings, or graphs, or Topsnut-gpws, or Chinese characters, \emph{etc.}, this is the reason we call them \emph{pan-matrices}.

\item We can generalize the definition of ve-graceful evaluated coloring to over one hundred colorings/labellings in \cite{Gallian2018}. For example, if there exists a coloring $h$ on a tree $G$ of $n$ vertices and Topcode$^+$-matrix group $\{F_{2n};\oplus\}$ with $F_{2n}=\{G_i:i\in [1,2n]\}$, such that $h$: $V(T)\rightarrow F_{2n}$ holds the vertex coloring set $h(V(G))\subset F_{2n}$ with $|h(V(G))|=n$ and the edge coloring set $h(E(G))=\{G_1$, $G_{3}$, $G_{5}$, $\dots $, $G_{2n-1}\}$, then we call $h$ a \emph{ve-odd-graceful evaluated coloring}, and moreover $h$ is called an \emph{e-odd-graceful evaluated coloring} if $|h(V(G))|<n$ and $h(E(G))=\{G_1$, $G_{3}$, $G_{5}$, $\dots $, $G_{2n-1}\}$.
\end{asparaenum}
\end{rem}

\begin{figure}[h]
\centering
\includegraphics[width=7.6cm]{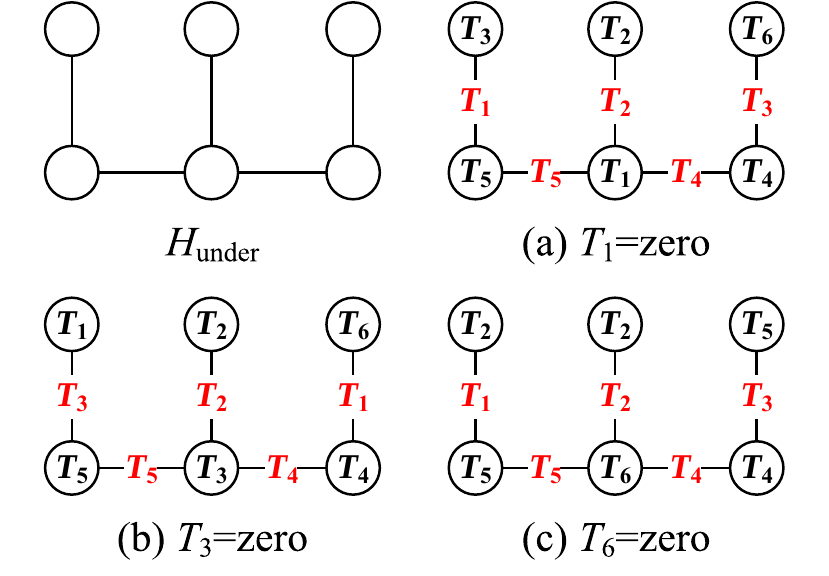}\\
\caption{\label{fig:group-vs-graphical} {\small Three graph networks (a), (b) and (c) are based on an every-zero Topcode$^+$-matrix group $\{T_i:i\in [1,6]\}$ shown in Fig.\ref{fig:topcode-matrix-group} by the additive v-operation defined in (\ref{eqa:group-operation-1}) and (\ref{eqa:group-operation-2}).}}
\end{figure}

\begin{thm}\label{thm:ve-graceful-evaluated-coloring}
Let $\{F_p(G),\oplus\}$ be an every-zero $\varepsilon$-group, where $F_p(G)=\{G_1,G_2,\dots, G_p\}$, and $\varepsilon$-group is one of every-zero Topcode$^+$-matrix groups (Topcode-groups), every-zero number string groups, every-zero Topsnut-gpw groups and Hanzi-groups. If a tree $T$ of $p$ vertices admits a set-ordered graceful labelling defined in Definition \ref{defn:set-ordered-odd-graceful-labelling}, then $T$ admits a ve-graceful evaluated coloring on $\{F_p(G),\oplus\}$.
\end{thm}
\begin{proof} Let $T$ be a tree having $p$ vertices and the bipartition $(X,Y)$, where $X=\{x_i:~i\in [1,s]\}$ and
$Y=\{y_j:j\in [1,t]\}$ with $s+t=p$. By the hypothesis of the theorem, $T$ admits a set-ordered graceful labelling $f$ with $f(x_i)=i-1$ for $i\in [1,s]$ and $f(y_j)=s-1+j$ for $j\in [1,t]$, and $f(x_iy_j)=f(y_j)-f(x_i)=s+j-i-2$ for each edge $x_iy_j\in E(T)$. Clearly, $f(X)<f(Y)$, and moreover $f(y_j)+f(y_{t-j+1})=2s+t-1=s+p-1$ for $j\in [1,t]$.
We define a labelling $g$ for $T$ as follows:
$g(x_i)=1+f(x_i)$ for $i\in [1,s]$, $g(y_j)=1+f(y_{t-j+1})$ for
$j\in [1,t]$. Notice that the tree $T$ has $|E(T)|=p-1$ edges, for each edge $x_iy_j\in E(T)$,
\begin{equation}\label{eqa:graph-group-trees}
{
\begin{split}
g(x_iy_j)=&g(x_i)+g(y_j)=2+f(x_i)+f(y_{t-j+1})\\
&=f(x_i)+s+p+1-f(y_j)\\
&=s+p+1-\big [f(y_j)-f(x_i)\big ]\\
&=s+p+1-f(x_iy_j),
\end{split}}
\end{equation} we obtain two sets $S_1=\{s+p,s+p-1,\dots
,s+p-1-(s-1),s+p+1-(s+1)\}$ and $S_2=\{p-1,p-2,\dots ,s+1\}$ from (\label{eqa:graph-group-trees}). Under modulo
$(p-1)$, the set $S_1$ distributes a set $S\,'_1=[1,s]$. Thereby,
$g(E(T))=\{g(x_iy_j)=g(x_i)+g(y_j)~ (\bmod ~p-1): x_iy_j\in
E(T)\}=[1,p-1]$. Moreover,

If $g(x_i)+g(y_j)-k\leq 0$, thus, we have a set $\{1-k,2-k,\dots ,0\}$, take modulo $(p-1)$, immediately, we get a set $\{p-k,p-k+1,\dots ,p-1\}$; if $g(x_i)+g(y_j)-k\geq 0$, then we get $\{1,2,3,\dots ,p-k\}$. Thereby,
\begin{equation}\label{eqa:ve-graceful-evaluated-coloring}
{
\begin{split}
\quad &\{g(x_iy_j)-k~ (\bmod ~p-1): x_iy_j\in E(T)\}\\
=&\{g(x_i)+g(y_j)-k~ (\bmod ~p-1): x_iy_j\in E(T)\}\\
=&[1,p-1].
\end{split}}
\end{equation}

We have a coloring $h: V(T)\rightarrow \{F(G),\oplus\}$, such that $h(x_i)=G_{g(x_i)}$, $h(y_j)=G_{g(y_j)}$, and such that $h(x_iy_j)=G_{g(x_iy_j)-k}\in \{F_p(G),\oplus\}$ by the operation (\ref{eqa:ve-graceful-evaluated-coloring}). So, $h$ is a ve-graceful evaluated coloring of $T$ on $\{F_p(G),\oplus\}$.
\end{proof}

\begin{figure}[h]
\centering
\includegraphics[width=8.6cm]{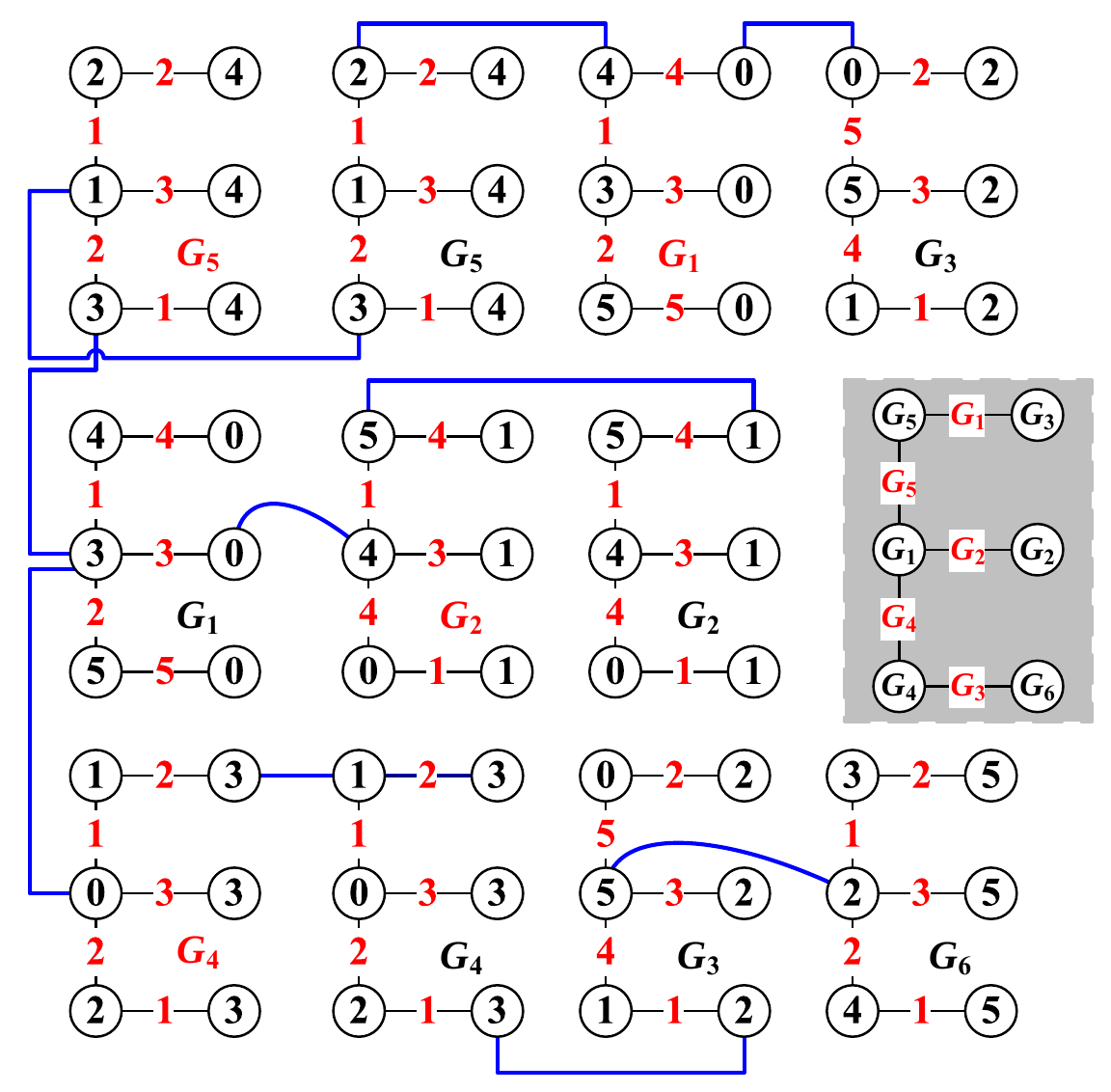}\\
\caption{\label{fig:label-with-graphs} {\small The underlying graph $H_{\textrm{under}}$ shown in Fig.\ref{fig:group-vs-graphical} is labelled with an every-zero graph group shown in Fig.\ref{fig:graph-group} to producing a Tosnut-gpw $G$ on gray rectangle, and $G$ is equivalent to $H$ shown in Fig.\ref{fig:group-vs-graphical} (a).}}
\end{figure}

\vskip 0.4cm

Observe the Tosnut-gpw $G$ shown in Fig.\ref{fig:label-with-graphs}, we can see there are two or more Tosnut-gpws like $G$ based on different blue joining-edges, so we collect these Tosnut-gpws like $G$ into a set $N_{graph}(G)$. It may be interesting to find some relationships between the graphs in $N_{graph}(G)$. By the similar method, we can show the following result:

\begin{thm}\label{thm:e-odd-graceful-weighted-value}
Given a tree $T$ of $n$ vertices and an every-zero Topcode$^+$-matrix group $\{F_{2n};\oplus\}$, then $T$ admits an e-odd-graceful evaluated coloring on $\{F_{2n};\oplus\}$. Moreover, if $T$ admits a set-ordered graceful labelling defined in Definition \ref{defn:set-ordered-odd-graceful-labelling}, then $T$ admits a ve-odd-graceful evaluated coloring on $\{F_{2n};\oplus\}$.
\end{thm}

\subsection{Overall security mechanism to networks}

We can set up a \emph{protection mechanism} to a colored network $\mathcal {N}(t)$ with a graph group $\{F_p(G);\oplus\}$ as follows: Each vertex $w$ of $\mathcal {N}(t)$ has its own neighbor set $N_{ei}(w)=\{w_1,w_2,\dots ,w_d\}$, where $d=\textrm{deg} (w)$. If a vertex $u$ out of $\mathcal {N}(t)$ will join with the vertex $w$ of $\mathcal {N}(t)$ (i.e. $u$ wants to visit $w$), then $u$ must to get admitted permit $per(w_i)$ of each vertex $w_i\in N_{ei}(w)$ with $i\in [1,d]$. The set $\{per(w_i):w_i\in N_{ei}(w)\}$ is the key for the vertex $u$ out of $\mathcal {N}(t)$ visiting the vertex $w$ of $\mathcal {N}(t)$. We call the protection mechanism as an \emph{overall security mechanism} of $\mathcal {N}(t)$.

The methods to realize our overall security mechanism to encrypt networks are:

(i) Encrypting the vertices and edges of a network by various groups introduced here, such as every-zero number string groups, every-zero Topcode$^+$-matrix groups, every-zero Topcode$^-$-matrix groups, pan-Topcode-matrix groups and graph groups (Abelian additive groups on Topsnut-gpws), and so on. Using the elements of a group encrypts the vertices and edges of a network.

(ii) If a graph set $G(T_{code})$, in which each graph has its Topcode-matrix $T_{code}$, admits a connection between elements of $G(T_{code})$, so this connection and $G(T_{code})$ can be used to encrypt a network.

(iii) Use a group set $\mathcal {F}$ with its elements are various groups introduced here to encrypt a network $\mathcal {N}(t)$, any pair of two elements of $\mathcal {F}$ is equivalent to each other, such as equivalent Abelian additive groups on Topsnut-gpws.

\section{Conclusion and discussion}

Motivated from convenient, useful and powerful matrices used in computation and investigation of today's networks, we have introduced \emph{Topcode-matrices} that differ from popular matrices applied in linear algebra and computer science. Topcode-matrices can use numbers, letters, Chinese characters, sets, graphs, algebraic groups \emph{etc.} as their elements. One important thing is that Topcode-matrices of numbers can derive easily number strings, since number strings are TB-paws used in information security. And we have obtained several operations on Topcode-matrices, especially, our additive operation and subtractive operation help us to build up every-zero Topcode-matrix groups and every-zero number string groups, and so on.

Topsnut-matrices, Topcode-matrices, Hanzi-matrices, adjacency ve-value matrices and pan-Topcode-matrices have been discussed. Particular Topcode-matrices introduced here are graceful Topcode-matrix, set-ordered graceful Topcode-matrix, odd-graceful Topcode-matrix, set-ordered odd-graceful Topcode-matrix, elegant Topcode-matrix, edge-magic total Topcode-matrix, odd-edge-magic total Topcode-matrix, edge sum-difference Topcode-matrix, strongly graceful Topcode-matrix, strongly odd-graceful Topcode-matrix, $(k,d)$-graceful Topcode-matrix, and so on. Some connections between these Topcode-matrices have been proven. The definition of a topological graphic password (Topsnut-gpw) is: \emph{A Topsnut-gpw consists of a no-colored $(p,q)$-graph $G$ and an evaluated Topcode-matrix $T_{code}$ such that $G$ can be colored by $T_{code}$}. No doubt, Topcode-matrices will provide us more new things in operations, parameters and topics of graph theory.

For further investigation of Topcode-matrices we propose the following questions:

\begin{asparaenum}[\textbf{\textrm{Question}} 1. ]
\item Characterize a Topcode-matrix $T_{code}$ such that each simple graph $G$ derived by $T_{code}$ (that is $A_{vev}(G)=T_{code}$) is connected.
\item Each graceful Topcode-matrix $T_{code}$ (resp. each odd-graceful Topcode-matrix) corresponds a tree admitting a graceful labelling (resp. an odd-graceful labelling) when $|(XY)^*|=q+1=|E^*|+1$. Conversely, each tree corresponds a graceful Topcode-matrix (resp. an odd-graceful Topcode-matrix). We conjecture: \emph{Each tree with diameter less than three admits a splitting graceful (resp. odd-graceful) coloring}. Some trees were shown to support this conjecture in \cite{Yao-Mu-Sun-Zhang-Yang-Wang-Wang-Su-Ma-Sun-2019}.
\item Find a set of Topcode-matrices $T^1_{code}$, $T^2_{code}$, $\dots$, $T^m_{code}$, such that there exists a transformation $F_{i,j}$ holding $T^j_{code}=F_{i,j}(T^i_{code})$ for any pair of $i,j\in [1,m]$, see some examples shown in the proof of Theorem \ref{thm:equivalent-Topcode-matrices}.
\item A Topcode-matrix $T_{code}$ defined in Definition \ref{defn:Topcode-matrix} corresponds a graph. We have a transformation $T'_{code}=f(T_{code})$ by a function $f:$ $x'_i=f(x_i)$, $e'_i=f(e_i)$ and $y'_i=f(y_i)$ with $i\in [1,q]$ such that $x'_i,e'_i,y'_i$ are non-negative integers, where $T'_{code}=(X'~E'~Y')^{T}$. There are: (c-1) $x'_i\neq y'_i$ for each $(x'_i\quad e'_i\quad y'_i)^{T}$; (c-2) $e'_i\neq e'_j$ if both $e'_i,e'_j$ have a common end. We call $T'_{code}$ a \emph{v-proper non-negative integer Topcode-matrix} of $T_{code}$ if (c-1) holds true, an \emph{e-proper non-negative integer Topcode-matrix} of $T_{code}$ if (c-2) holds true, a \emph{ve-proper non-negative integer Topcode-matrix} of $T_{code}$ if both (c-1) and (c-2) hold true. Let $\Delta(T_{code})=\max \{|N_{ei}(w_i)|:w_i\in (XY)^*\}$, find:

\quad (1-1) The number $\chi=\min\max (X'Y')^*$ over all v-proper non-negative integer Topcode-matrices of $T_{code}$. It was conjectured that $\chi\leq \lceil \frac{\Delta(T_{code})+|(X'Y')^*|}{2}\rceil$ (Ref. \cite{Bondy-2008}), where $T'_{code}=(X'~E'~Y')^{T}$ is a maximum clique Topcode-matrix of $T_{code}$.

\quad (1-2) The number $\min\max (X'Y')^*$ over all v-distinguishing Topcode-matrices that are v-proper non-negative integer Topcode-matrices of $T_{code}$.

\quad (1-3) The number $\min\max (X'Y')^*$ over all adjacent v-distinguishing Topcode-matrices that are v-proper non-negative integer Topcode-matrices of $T_{code}$.

\quad (2-1) The number $\min\max (E')^*$ over all e-proper non-negative integer Topcode-matrices of $T_{code}$.

\quad (2-2) The number $\chi'_{as}=\min\max (E')^*$ over all adjacent e-distinguishing Topcode-matrices which are e-proper non-negative integer Topcode-matrices of $T_{code}$. It was conjectured that $\chi'_{as}\leq \Delta(T_{code})+2$ (Ref. \cite{Zhang-Liu-Wang-Applied-Mathematics-2002}).

\quad (3-1) The number $\chi''=\min\max (X'Y')^*\cup (E')^*$ over all ve-proper non-negative integer Topcode-matrices of $T_{code}$. It was conjectured $\chi''\leq \Delta(T_{code})+2$.

\quad (3-2) The number $\chi''_{as}=\min\max (X'Y')^*\cup (E')^*$ over all adjacent total-distinguishing Topcode-matrices which are ve-proper non-negative integer Topcode-matrices of $T_{code}$. It was conjectured that $\chi''_{as}\leq \Delta(T_{code})+3$.

\item Let $G(T_{code})$ be the set of graphs corresponding to a graceful/odd-graceful Topcode-matrix $T_{code}$. (i) There exists at least a simple graph (having no multiple-edge and self-edge) in $G(T_{code})$. (ii) If there is a graph $H\in G(T_{code})$ holding $|V(H)|\geq |V(G)|$ for any $G\in G(T_{code})$, so $H$ is a tree. (iii) What connections are there among the graphs of $G(T_{code})$? (vi) If $|V(H^*)|\leq |V(G)|$ for a graph $H^*\in G(T_{code})$ and any $G\in G(T_{code})$, then $H^*$ has no two vertices $u,v$ holding $f(u)=f(v)$ true, where $f$ is a graceful/odd-graceful labelling of $H^*$.
\item Find conditions to two Topcode-matrices $T_{code}(G)$ corresponding a graph $G$ and $T'_{code}(H)$ corresponding a graph $H$, such that $G\cong H$ if $C_{(c,l)(a,b)}(T_{code}(G))=T'_{code}(H)$ defined in (\ref{eqa:graphs-isomorphic-Topcode-matrices}).
\item \label{problem:infinite-sets} \cite{Yao-bing2019-8-4-8-6} For three given positive integers $M,m,n$, any infinite non-negative integer string $\{a_p\}^{\infty}_1$ contains $m~(\geq 4)$ non-negative integer string segments $H_i=a_{i,1}a_{i,2}\dots a_{i,n}$ with $i\in[1,m]$, and these non-negative integer strings form a number string group based on the additive v-operation: There are a sequence $r_1,r_2,\dots,r_b$ and a fixed $k\in [1,m]$ such that
\begin{equation}\label{eqa:string-group}
[a_{i,r_t}+a_{j,r_t}-a_{k,r_t}]~(\bmod~ M)=a_{\lambda,r_t}
\end{equation}
where $\lambda=i+j-k~(\bmod~ M)\in [1,m]$ and $t\in [1,b]$.

\quad (\ref{problem:infinite-sets}.1) \cite{Yao-bing2019-8-4-8-6} Does any Chinese paragraph appear in an infinite string of Chinese characters? We can prove a simple result about this problem:

\begin{thm}\label{thm:lemma-equivalent-proof}
There exists an infinite number string of numbers $0$, $1$, $2$, $3$, $4$, $5$, $6$, $7$, $8$, $9$ such that it contains any number string $k(1)$ $k(2)$ $\cdots k(m_k)$ of numbers $0$, $1$, $2$, $3$, $4$, $5$, $6$, $7$, $8$, $9$, each $k(j)=$ $a_{k(j)}b_{k(j)}c_{k(j)}d_{k(j)}$ with $j\in [1,m_k]$ is the subscript of a Chinese character $H_{a_{k(j)}b_{k(j)}c_{k(j)}d_{k(j)}}$ in \cite{GB2312-80}.
\end{thm}
\begin{proof} As known, ``GB2312-80 Encoding of Chinese characters'' in \cite{GB2312-80} contains $M=6763$ Chinese characters, where each Chinese character $H_{abcd}$ has its own code $abcd$ of numbers $0$, $1$, $2$, $3$, $4$, $5$, $6$, $7$, $8$, $9$. Let $T_b(Hanzi)$ be an infinite string of numbers $0$, $1$, $2$, $3$, $4$, $5$, $6$, $7$, $8$, $9$, and let $G_k=$ $H_{k(1)}H_{k(2)}\cdots H_{k(m_k)}$ be a paragraph written by $m_k$ Chinese characters defined in \cite{GB2312-80}, we call $m_k$ to be the length of the Chinese paragraph $G_k$, where $k(j)=$ $a_{k(j)}b_{k(j)}c_{k(j)}d_{k(j)}$ with $j\in [1,m_k]$. So, we have $2^{M\cdot m_k}$ paragraphs with length $m_k$, each $G_k$ can be expressed as a number string
$${
\begin{split}
&\quad k(1)k(2)\cdots k(m_k)=a_{k(1)}b_{k(1)}c_{k(1)}d_{k(1)}\\
&a_{k(2)}b_{k(2)}c_{k(2)}d_{k(2)}\cdots a_{k(m_k)}b_{k(m_k)}c_{k(m_k)}d_{k(m_k)}
\end{split}}$$ with $4m_k$ bytes. We insert $k(1)$ $k(2)$ $\cdots $ $k(m_k)$ into $T_b(Hanzi)$, the new number string contains all Chinese paragraphs written by Chinese characters $H_{abcd}$ with code $abcd$ defined in \cite{GB2312-80}, also, contains $T_b(Hanzi)$. The proof of this theorem is complete
\end{proof}

\begin{thm}\label{thm:infinite-sequence}
There exists an infinite sequence of numbers $0$, $1$, $2$, $3$, $4$, $5$, $6$, $7$, $8$, $9$ such that it contains each number string group based on the additive v-operation ``$\oplus$'' defined in (\ref{eqa:group-operation-1}) and (\ref{eqa:group-operation-2}), or the subtractive v-operation ``$\ominus$'' defined in (\ref{eqa:subtraction-group-operation-1}) and (\ref{eqa:subtraction-group-operation-2}).
\end{thm}

\quad In fact, let $S_{eque}=\{x_i\}^\infty_1$ be an infinite sequence with $x_i\in [0,9]$, and let $S_i=a_{i}a_{i+1}a_{i+2}\cdots a_{i+m}$ be in an every-zero number string group $\{F_m;\oplus\}$ (resp. $\{F_m;\ominus\}$). We insert each $S_i$ into $S_{eque}$ to get a new infinite sequence $S^*_{eque}$.

\quad (\ref{problem:infinite-sets}.2) Does any infinite number string $S^*$ of umbers $0$, $1$, $2$, $3$, $4$, $5$, $6$, $7$, $8$, $9$ contains any number string $S_\Delta=k(1)$ $k(2)$ $\cdots$ $k(m_k)$, that is $S_\Delta\subset S^*$, where each $$k(j)=a_{k(j)}b_{k(j)}c_{k(j)}d_{k(j)}$$ with $j\in [1,m_k]$ to be a Chinese character defined in \cite{GB2312-80}?

\quad (\ref{problem:infinite-sets}.3) Is any infinite triangular planar graph 4-colorable? Does any infinite triangular planar graph admitting a 4-coloring contain any planar graph having finite number of vertices? Here, an infinite triangular planar graph has infinite vertices, no outer face (or, its outer face degrades into the \emph{infinite point}) and each inner face is a triangle. See an infinite triangular planar graph with a 4-coloring is shown in Fig.\ref{fig:4-colorable}, so any infinite triangular planar graph admitting a 4-coloring can be tiled fully by four triangles of Fig.\ref{fig:4-colorable}(a). It may be interesting to study infinite planar graphs, or other type of infinite graphs.

\begin{figure}[h]
\centering
\includegraphics[width=8.6cm]{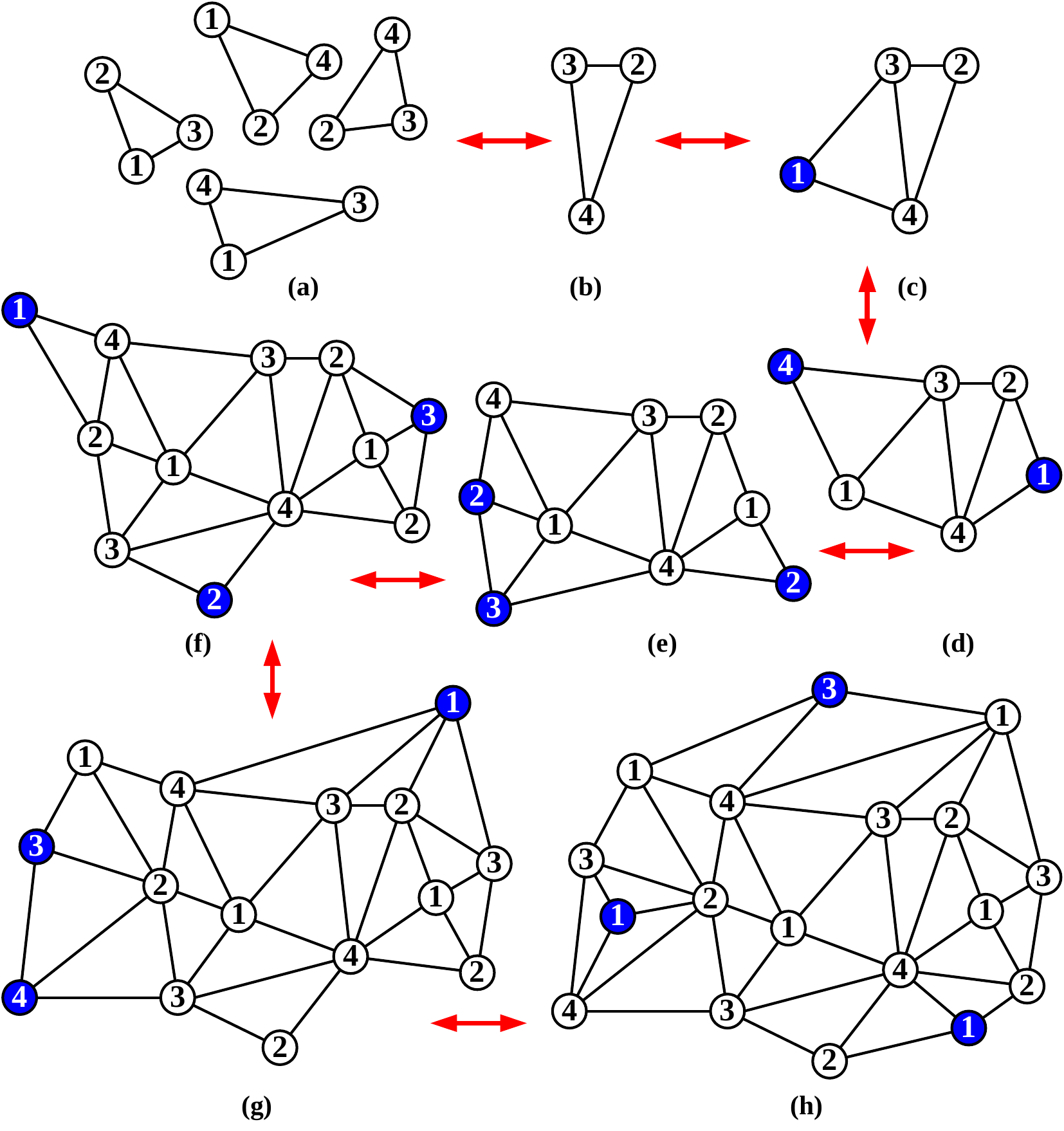}\\
\caption{\label{fig:4-colorable} {\small A process of constructing an infinite triangular planar graph admitting a 4-coloring.}}
\end{figure}

\item For a set $\{x_{i,j}:i\in[1,m], j\in [1,n]\}$, we write a string $S=x_{i_1,j_1}x_{i_1,j_1}\cdots x_{i_p,j_p}$ with $p=mn$, and $x_{i_s,j_s}\neq x_{i_t,j_t}$ if one of $i_s\neq i_t$ and $j_s\neq j_t$ to be true. If any pair of two consecutive elements $x_{i_k,j_k}x_{i_{k+1},j_{k+1}}$ holds $|i_k-i_{k+1}|\leq 1$ and $|j_k-j_{k+1}|\leq 1$ true, we call $S$ an \emph{adjacent string}. Find all possible adjacent strings for the given set $\{x_{i,j}:i\in[1,m], j\in [1,n]\}$.
\item Every tree $T$ of $p$ vertices corresponds an every-zero Topcode$^+$-matrix group $\{F_p;\oplus\}$ with $p\geq 1$, such that $T$ has a ve-graceful evaluated coloring on $\{F_p;\oplus\}$.
\item Suppose that a graceful Topcode-matrix $T_{code}$ corresponds a set $F(T_{code})$ of graphs, where $F(T_{code})=\{G:~A_{vev}(G)=T_{code}\}$. Find some connection among the graphs of $F(T_{code})$.
\item For a given odd-graceful Topcode-matrix $T_{code}$, find all Topcode-matrices $T^k_{code}$ such that $(T_{code},\overline{T}_{code})$ to be pairs of twin odd-graceful Topcode-matrices.
\item If the elements of the graph coefficient matrix $A_{graph}$ defined in (\ref{eqa:graph-coefficient-matrix}) forms a graph group, how about the graph vector $Y(L)$ in (\ref{eqa:graph-coefficient-matrix-equation})?
\item As known, a graph $G$ with a coloring/labelling has its own Topcode-matrix $A_{vev}(G)$, and $A_{vev}(G)$ derives $M$ number strings $T^j_b(G)=a_1a_2\cdots a_n$ with $j\in [1,M]$. Conversely, each number string $T^j_b(G)=a_1a_2\cdots a_n$ should reconstruct the original graph $G$. Does any number string produce a graph?
\item Observe Fig.\ref{fig:contract-splitting}, we define a multiple-edge complete graph $K^*_n$ of $n$ vertices in the way: vertex set $V(K^*_n)=\{x_i:~i\in [1,n]\}$, and edge set
$$E(K^*_n)=\bigcup_{i,j\in[1,n],i\neq j} E_{ij},$$ where $E_{ij}=\{e^s_{ij}=x_ix_j:~s\in[1,m_{ij}]\}$ with $m_{ij}\geq 1$. Thereby, a multiple-edge complete graph $K^*_n$ of $n$ vertices has $M$ edges in total, where $M=\sum_{i,j\in[1,n],i\neq j} m_{ij}$. How to split the vertices of a multiple-edge complete graph $K^*_n$ of $M$ edges into a previously given $n$-colorable simple graph $G$ of $M$ edges? Conversely, it is nor hard to v-contract the vertices of a $n$-colorable simple graph to obtain a multiple-edge complete graph of $n$ vertices. (1) Suppose that $K^*_n$ is a public key, then it is not easy to find a private key in coding theory. (2) A maximal planar graph can be contracted into a multiple-edge complete graph $K^*_4$, so it may be interesting to find conditions for $K^*_4$, such that $K^*_4$ can be v-split into any given maximal planar graph. If this problem holds true, we get a method to prove 4CC (Ref. \cite{Jin-Xu-2019-1-978-7-03-060377-7}).

\begin{figure}[h]
\centering
\includegraphics[width=8.6cm]{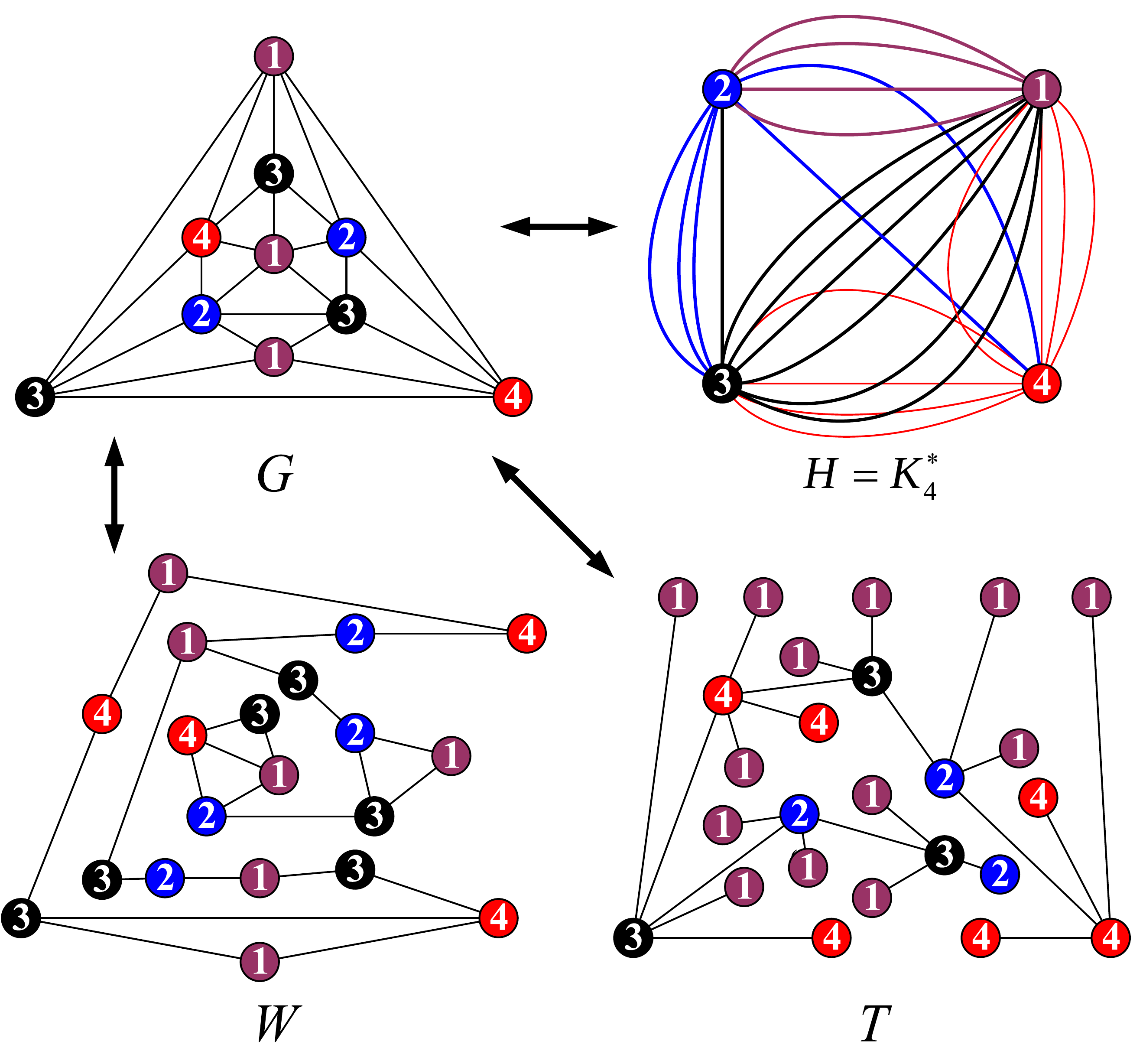}\\
\caption{\label{fig:contract-splitting} {\small A 4-colorable maximal planar graph $G$ can be contracted into a multiple-edge complete graph $K^*_4$, conversely, $K^*_4$ can be v-split into $G$. The graph $W$ is obtained from $G$ by the 2-degree v-splitting operation, and $T$ obtained from $G$ by a series of v-splitting operations is a tree with maximum leaves.}}
\end{figure}

\quad Let $P=xyz$ be a path of 2 length in a planar graph $G$. We split $y$ into two $y'$ and $y''$, such that $N_{ei}(y)=N_{ei}(y')\cup N_{ei}(y'')$ with $x,z\in N_{ei}(y')$ (it is allowed that $N_{ei}(y'')=\emptyset$), and add three edges $y'y''$, $xy''$ and $zy''$, the resultant graph $(G\wedge y)+\{xy'',y'y'',zy''\}$ is still a planar graph, and has a rhombus $xy'zy''$ (see Fig.\ref{fig:Lingxing-operation}). So we have another way for showing possibly 4CC the rhombus expanded-contracted operation system introduced in \cite{wang-hong-yu-2018-doctor-thesis}. One example about the rhombus expanded-contracted operation system is shown in Fig.\ref{fig:Lingxing-operation-WHY}.

\begin{figure}[h]
\centering
\includegraphics[width=8.6cm]{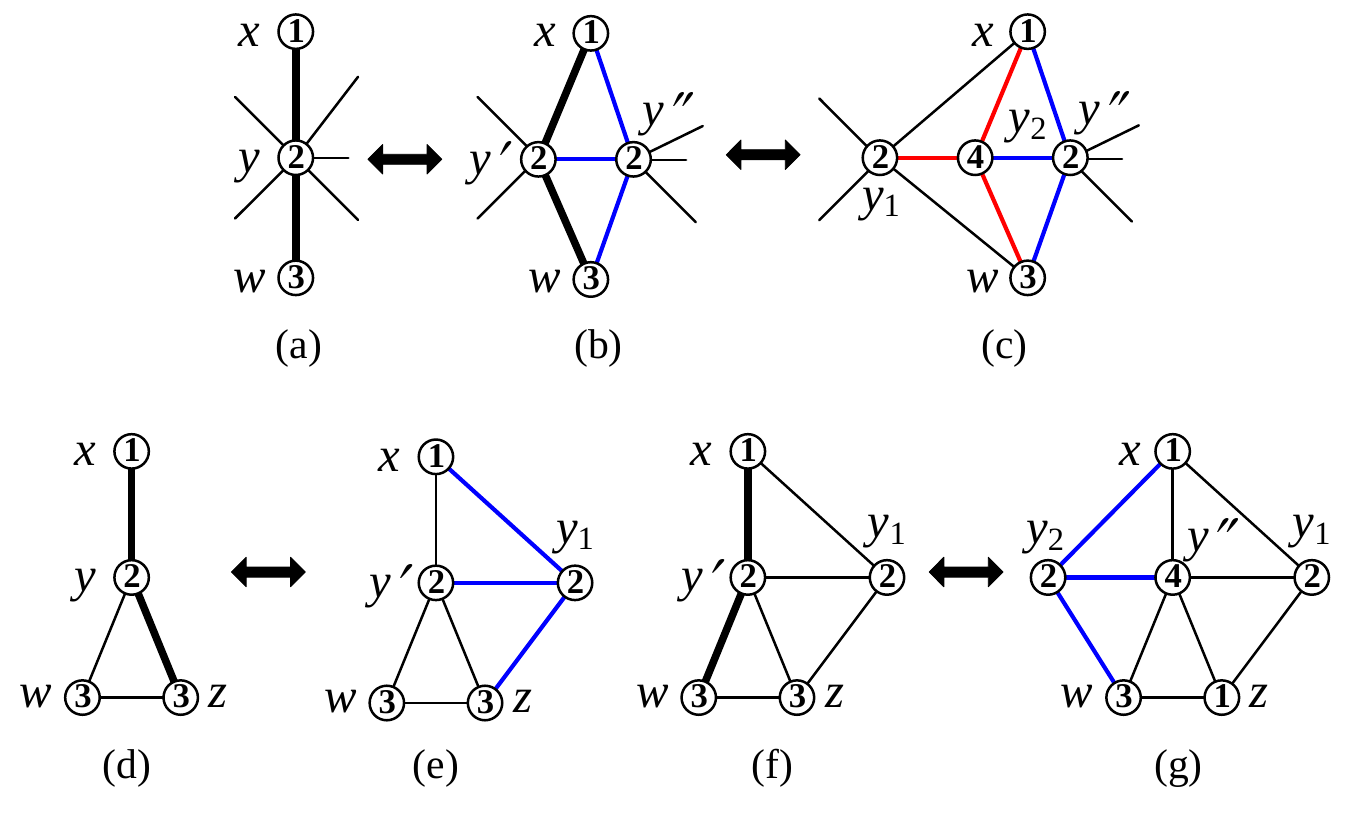}\\
\caption{\label{fig:Lingxing-operation} {\small The rhombus expanded-contracted operation system introduced in \cite{wang-hong-yu-2018-doctor-thesis}, where the vertex $y$ is split into $y'$ and $y''$ from (a) to (b), and the vertex $y'$ is split into $y_1$ and $y_2$ from (b) to (c); conversely, identifying two vertices $y_1$ and $y_2$ into one vertex $y'$ from (c) to (b), identifying two vertices $y'$ and $y''$ into one vertex $y$ from (b) to (a). }}
\end{figure}

\begin{figure}[h]
\centering
\includegraphics[width=8.6cm]{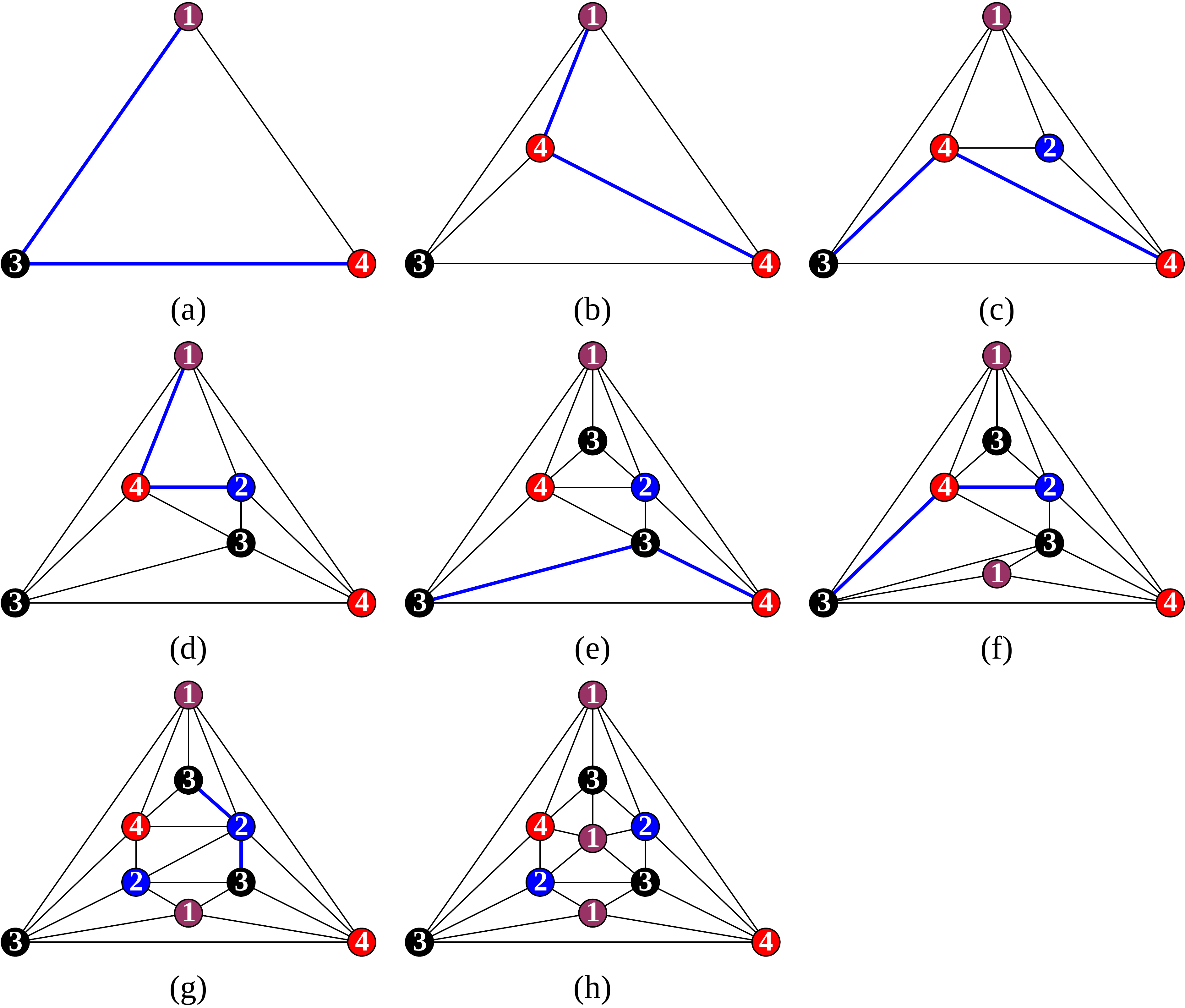}\\
\caption{\label{fig:Lingxing-operation-WHY} {\small A maximal triangular planar graph admitting a 4-coloring obtained by the rhombus expanded-contracted operation system introduced in \cite{wang-hong-yu-2018-doctor-thesis}.}}
\end{figure}

\end{asparaenum}

\vskip 0.4cm

We have known that a Topcode-matrix may corresponds two or more colored graphs, so Theorem \ref{thm:equivalent-Topcode-matrices} means that a collection of colored graphs corresponding an $A$-type Topcode-matrix is equivalent to another collection of colored graphs corresponding a $B$-type Topcode-matrix. Topcode-matrices will distribute us more applications for security of networks and information, since a Topcode-matrix may corresponds many colored/labelled graphs, also, Topsnut-gpws. Topcode-matrices have contained many unsolved mathematical problems, so they can force attackers to give up breaking down networks as the sizes of Topcode-matrices are enough larger.

\section*{Acknowledgment}

The research was supported by the National Natural Science Foundation of China under grants No. 61163054, No. 61363060 and No. 61662066.




\begin{thebibliography}{1}
\bibitem{Xiaogang-Wen-2018-Sun-Yat-sen-University}Xiaogang Wen. Four physical revolutions and the second quantum revolution (talk). The sixth issue of the high-level academic report of the Furan Forum, Sun Yat-sen University South Campus Auditorium, December 29, 2018.
\bibitem{Li-Gu-Dullien-Vinyals-Kohli-arXiv2019}Yujia Li, Chenjie Gu, Thomas Dullien, Oriol Vinyals, Pushmeet Kohli. Graph Matching Networks for Learning the Similarity of Graph Structured Objects. Proceedings of the 36 th International Conference on Machine
Learning, Long Beach, USA, 2019. arXiv:1904.12787v1 [cs.LG] 29 Apr 2019.
\bibitem{Battaglia-27-authors-arXiv1806-01261v2}Peter W. Battaglia, Jessica B. Hamrick, Victor Bapst, Alvaro Sanchez-Gonzalez, Vinicius Zambaldi, Mateusz Malinowski, Andrea Tacchetti, David Raposo, Adam Santoro, Ryan Faulkner, Caglar Gulcehre, Francis Song, Andrew Ballard, Justin Gilmer, George Dahl, Ashish Vaswani, Kelsey Allen, Charles Nash4, Victoria Langston, Chris Dyer, Nicolas Heess, Daan Wierstra, Pushmeet Kohli, Matt Botvinick, Oriol Vinyals, Yujia Li, Razvan Pascanu. Relational inductive biases, deep learning, and graph networks. arXiv:1806. 01261v2 [cs. LG] 11 Jun 2018.
\bibitem{Harary-book} F. Harary. Graph Theory. \emph{Addison-Wesley}, 1969.
\bibitem{Bang-Jensen-Gutin-digraphs-2007}Jorgen Bang-Jensen, Gregory Gutin. Digraphs Theory, Algorithms and Applications. Springer-Verlag, 2007.
\bibitem{Bondy-2008} J. A. Bondy, U. S. R. Murty. Graph Theory. Springer London, 2008.
\bibitem{Gallian2018} Joseph A. Gallian. A Dynamic Survey of Graph Labeling. \emph{The electronic journal of
combinatorics}, Twenty-first edition, December 21 (2018), \# DS6. (502 pages, 2643 reference papers, over 200 graph labellings)
\bibitem{Suo-Zhu-Owen-2005} Xiaoyuan Suo, Ying Zhu, G. Scott. Owen. Graphical Password: A Survey. In: Proceedings of Annual
Computer Security Applications Conference (ACSAC), Tucson, Arizona. IEEE (2005) 463-472. (10 pages, 38 reference
papers)
\bibitem{Biddle-Chiasson-van-Oorschot-2009}R. Biddle, S. Chiasson, and P. C. van Oorschot. Graphical passwords: Learning from the First Twelve Years. ACM Computing Surveys, \textbf{44} (4), Article 19:1-41. Technical Report TR-09-09, School of Computer Science, Carleton University, Ottawa, Canada. 2009. (25 pages, 145 reference papers)
\bibitem{Gao-Jia-Ye-Ma-2013}Haichang Gao, Wei Jia, Fei Ye and Licheng Ma. A Survey on the Use of Graphical Passwords in Security. Journal Of Software, Vol. \textbf{8} (7), July 2013, 1678-1698. (21 pages, 88 reference papers)




\bibitem{Wang-Xu-Yao-2016} Hongyu Wang, Jin Xu, Bing Yao. Exploring New Cryptographical Construction Of Complex Network Data. IEEE First International Conference on Data Science in Cyberspace. IEEE Computer Society, (2016):155-160.
\bibitem{Wang-Xu-Yao-Key-models-Lock-models-2016}Hongyu Wang, Jin Xu, Bing Yao. The Key-models And Their Lock-models For Designing New Labellings Of Networks.Proceedings of 2016 IEEE Advanced Information Management, Communicates, Electronic and Automation Control Conference (IMCEC 2016) 565-5568.

\bibitem{GB2312-80} ``GB2312-80 Encoding of Chinese characters'' cited from The Compilation Of National Standards For Character Sets And Information Coding, China Standard Press, 1998.

\bibitem{Sun-Zhang-Zhao-Yao-2017}Hui Sun, Xiaohui Zhang, Meimei Zhao and Bing Yao. New Algebraic Groups Produced By Graphical Passwords Based On Colorings And Labellings. ICMITE 2017, MATEC Web of Conferences \textbf{139}, 00152 (2017), DOI: 10. 1051/matecconf/201713900152
\bibitem{Yao-Sun-Zhao-Li-Yan-2017}Bing Yao, Hui Sun, Meimei Zhao, Jingwen Li, Guanghui Yan. On Coloring/Labelling Graphical Groups For Creating New Graphical Passwords. (ITNEC 2017) 2017 IEEE 2nd Information Technology, Networking, Electronic and Automation Control Conference. (2017) 1371-1375.
\bibitem{Yao-Zhang-Sun-Mu-Sun-Wang-Wang-Ma-Su-Yang-Yang-Zhang-2018arXiv}Bing Yao, Xiaohui Zhang, Hui Sun, Yarong Mu, Yirong Sun, Xiaomin Wang, Hongyu Wang, Fei Ma, Jing Su, Chao Yang, Sihua Yang, Mingjun Zhang. Text-based Passwords Generated From Topological Graphic Passwords. arXiv: 1809. 04727v1 [cs.IT] 13 Sep 2018.
\bibitem{Yao-Mu-Sun-Sun-Zhang-Wang-Su-Zhang-Yang-Zhao-Wang-Ma-Yao-Yang-Xie2019}Bing Yao, Yarong Mu, Yirong Sun, Hui Sun, Xiaohui Zhang, Hongyu Wang, Jing Su, Mingjun Zhang, Sihua Yang, Meimei Zhao, Xiaomin Wang, Fei Ma, Ming Yao, Chao Yang, Jianming Xie. Using Chinese Characters To Generate Text-Based Passwords For Information Security. arXiv:1907.05406v1 [cs.IT] 11 Jul 2019.
\bibitem{Yao-Mu-Sun-Zhang-Yang-Wang-Wang-Su-Ma-Sun-2019}Bing Yao, Yarong Mu, Yirong Sun, Mingjun Zhang, Sihua Yang, Hongyu Wang, Xiaomin Wang, Jing Su, Fei Ma, Hui Sun. Splitting Graceful And Pan-graceful Codes Towards Information Security. submitted 2019.
\bibitem{Yao-Zhang-Wang2010}Bing Yao, Zhong-fu Zhang and Jian-fang Wang. Some results on spanning trees[J]. Acta Mathematicae Applicatae Sinica, English Series, 2010, 26(4).607-616. DOI: 10.1007/s10255-010-0011-4
\bibitem{Yao-Sun-Zhang-Mu-Sun-Wang-Su-Zhang-Yang-Yang-2018arXiv}Bing Yao, Hui Sun, Xiaohui Zhang, Yarong Mu, Yirong Sun, Hongyu Wang, Jing Su, Mingjun Zhang, Sihua Yang, Chao Yang. Topological Graphic Passwords And Their Matchings Towards Cryptography. arXiv: 1808. 03324v1 [cs.CR] 26 Jul 2018.
\bibitem{Jin-Xu-2019-1-978-7-03-060377-7}Jin Xu. On theory of maximal planar graphs (first of two volumes). Science Press (Chinese), 2019, ISBN 978-7-03-060377-7.
\bibitem{Bing-Yao-Cheng-Yao-Zhao2009}Bing Yao, Hui Cheng, Ming Yao and Meimei Zhao. A Note on Strongly Graceful Trees. Ars Combinatoria \textbf{92} (2009), 155-169.
\bibitem{Zhou-Yao-Chen-Tao2012}Xiangqian Zhou, Bing Yao, Xiang'en Chen and Haixia Tao. A proof to the odd-gracefulness of all lobsters. \emph{Ars Combinatoria} \textbf{103} (2012), 13-18.
\bibitem{wang-hong-yu-2018-doctor-thesis}Hongyu Wang. The structure and theoretical analysis of a topological graphic cipher. Doctoral dissertation. Peking University, 2018.6.
\bibitem{Yao-bing2019-8-4-8-6}Bing Yao. Discrete Intelligent Computing Expert Committee Annual Meeting, 2019 China Artificial Intelligence Society, Lanzhou Jiaotong University, 2019.8.4-8.6.
\bibitem{Zhang-Liu-Wang-Applied-Mathematics-2002}Zhang Zhong-fu, Liu Lin-zhong and Wang Jian-fang. Adjacent Strong Edge Coloring of Graphs. Applied Mathematics Letters. 15(5)(2002), 623-626.
\end{thebibliography}
%

\end{CJK}

\end{document}